\documentclass[a4paper,UKenglish,cleveref, autoref,numberwithinsect]{lipics-v2021}
\usepackage{hyperref}
\usepackage[table]{xcolor}

\nolinenumbers

\usepackage{amssymb}
\usepackage{graphicx}

\usepackage{url}
\usepackage{amsfonts}
\usepackage{mathtools}
\usepackage{setspace}
\usepackage{amsmath}
\usepackage{xspace}
\usepackage{latexsym}
\usepackage{verbatim}
\usepackage{paralist}
\usepackage{tikz}
\usetikzlibrary{positioning,arrows}

\tikzstyle{every node} =
[draw = none, fill = white, thin]

\tikzstyle{noall} =
[draw = none, fill = none]
\tikzstyle{nodraw} =
[draw = none, fill = white]
\tikzstyle{nofill} =
[draw = black, fill = black]

\tikzstyle{cnode} =
[rectangle, draw = black, thin,align=center]

\newcommand{\details}[1]{}

\newcommand{\AAWA}{\text{\sffamily AAWA}}
\newcommand{\NAAWA}{\text{$n$\sffamily AAWA}}
\newcommand{\NAWA}{\text{\sffamily NAWA}}
\newcommand{\NNAWA}{\text{$n$\sffamily NAWA}}
\newcommand{\HLTL}{\text{\sffamily HyperLTL}}
\newcommand{\AHLTL}{\text{\sffamily A-HyperLTL}}

\newcommand{\SHLTL}{\text{\sffamily  HyperLTL$_{S}$}}
\newcommand{\CHLTL}{\text{\sffamily HyperLTL$_{C}$}}
\newcommand{\HPDL}{\text{\sffamily HyperPDL-$\Delta$}}
\newcommand{\HQPTL}{\text{\sffamily HyperQPTL}}
\newcommand{\HU}{\text{\sffamily H$_\mu$}}
\newcommand{\FO}{\text{\sffamily FO[$<$]}}
\newcommand{\FOF}{\text{\sffamily FO$_f$[$<$]}}
\newcommand{\FOPLUS}{\text{\sffamily FO$_f$[$<$, $+$]}}
\newcommand{\FOINPLUS}{\text{\sffamily FO[$<$, $+$]}}
\newcommand{\FOE}{\text{\sffamily FO[$<$,E]}}
\newcommand{\MSOE}{\textup{\sffamily S1S[E]}}
\newcommand{\MSO}{\text{\sffamily S1S}}

\newcommand{\PDL}{\text{\sffamily PDL}}
\newcommand{\QPTL}{\text{\sffamily QPTL}}
\newcommand{\CTL}{\text{\sffamily CTL}}
\newcommand{\CTLStar}{\text{\sffamily CTL$^{*}$}}
\newcommand{\LTL}{\text{\sffamily LTL}}
\newcommand{\HCTLStar}{\text{\sffamily HyperCTL$^{*}$}}

\newcommand{\Lang}{\mathcal{L}}
\newcommand{\Logic}{\textit{L}}
\newcommand{\AP}{\textsf{AP}}
\newcommand{\Ku}{\mathcal{K}}

\newcommand{\Au}{\mathcal{A}}
\newcommand{\AT}{\textsf{A}}
\newcommand{\ET}{\textsf{E}}

\newcommand \tpl[1]{\langle #1 \rangle}
\newcommand{\Var}{\textsf{VAR}}

\newcommand{\Lab}{\textit{Lab}}
\newcommand{\stfr}{\textit{stfr}}
\newcommand{\SUCC}{\textit{succ}}

\newcommand{\cl}{{\textit{cl}}}
\newcommand{\Dom}{{\textit{Dom}}}

\newcommand{\Atoms}{{\textit{Atoms}}}

\newcommand{\Inst}{{\textit{L}}}
\newcommand{\halt}{{\textit{halt}}}
\newcommand{\rec}{{\textit{rec}}}
\newcommand{\init}{{\textit{init}}}
\newcommand{\Beg}{\textit{beg}}
\newcommand{\Pad}{\textit{pad}}

\newcommand{\inc}{{\mathsf{inc}}}
\newcommand{\dec}{{\mathsf{dec}}}
\newcommand{\zero}{{\mathsf{if\_zero}}}
\newcommand{\instr}{{\textit{op}}}

\newcommand{\Until}{\textsf{U}}
\newcommand{\PUntil}{\textsf{S}}

\newcommand{\Next}{\textsf{X}}
\newcommand{\PNext}{\textsf{Y}}
\newcommand{\Always}{\textsf{G}}

\newcommand{\Eventually}{\textsf{F}}

\newcommand{\und}{\textsf{und}}

\def\N{{\mathbb{N}}}
\def\B{{\mathbb{B}}}
\def\S{{\mathcal{S}}}
\def\C{{\mathcal{C}}}
\def\F{{\mathcal{F}}}
\def\U{{\mathcal{U}}}
\def\M{{\mathcal{M}}}
\def\V{{\mathcal{V}}}

\newcommand{\true}{\texttt{true}}
\newcommand{\false}{\texttt{false}}

\def\PSPACE{{\sc Pspace}}

\newcommand{\Rel}[2]{\ensuremath{#1[#2]}}
\newcommand{\DefinedAs}{\ensuremath{\,\stackrel{\text{\textup{def}}}{=}\,}}

\title{Expressiveness and Decidability of Temporal Logics for Asynchronous Hyperproperties}

\titlerunning{Expressiveness and decidability  of temporal logics for asynchronous hyperproperties}

\author{Laura Bozzelli}{University of Napoli ``Federico II'', Napoli, Italy}{}{}{}{}
\author{Adriano Peron}{University of Napoli ``Federico II'', Napoli, Italy}{}{}{}{}
\author{C\'esar S\'anchez}{IMDEA Software Institute, Madrid, Spain}{}{}{}{}

 \authorrunning{L.\ Bozzelli,      A.\ Peron and C.\ S\'anchez}



\begin{document}

\maketitle

\begin{abstract}
  Hyperproperties are properties of systems that relate  different
  executions traces, with many applications from security to
  symmetry, consistency models of concurrency, etc.
  In recent years, different linear-time logics for specifying
  \emph{asynchronous} hyperproperties have been investigated.
  Though model checking of these logics is undecidable, useful
  decidable fragments have been identified with applications e.g.~for
  asynchronous security analysis.
  In this paper, we address  expressiveness and decidability
  issues of temporal logics for asynchronous hyperproperties.
  We compare the expressiveness of these logics together with
  the extension $\MSOE$ of $\MSO$ with the equal-level predicate by
  obtaining an almost complete expressiveness picture.
  We also study the expressive power of these logics when interpreted
  on singleton sets of traces.
  We show that for two asynchronous extensions of $\HLTL$, checking
  the existence of a singleton model is already undecidable, and for
  one of them, namely Context $\HLTL$ ($\CHLTL$), we establish a
  characterization of the singleton models in terms of the extension
  of standard $\FO$ over traces with addition.
  This last result generalizes the well-known equivalence between
  $\FO$ and $\LTL$.
  Finally, we identify new boundaries on the decidability of model
  checking $\CHLTL$.
\end{abstract}

\section{Introduction}\label{sec:Intro}

\noindent \textbf{Hyperproperties.}
In the last decade, a novel specification paradigm has been introduced  that generalizes  traditional  trace properties by properties of sets of traces, the so called \emph{hyperproperties}~\cite{ClarksonS10}.
Hyperproperties relate execution traces of a reactive system and are
useful in many settings.
In the area of information flow control, hyperproperties can formalize
security policies (like
noninterference~\cite{goguen1982security,McLean96} and observational
determinism~\cite{ZdancewicM03}) which compare observations made by an
external low-security agent along traces resulting from different
values of not directly observable inputs.
These security requirements go, in general, beyond regular properties
and cannot be expressed in classical regular temporal logics such as
\LTL~\cite{Pnueli77}, \CTL, and \CTLStar~\cite{EmersonH86}.
Hyperproperties also have applications in other settings, such as the
symmetric access to critical resources in distributed
protocols~\cite{FinkbeinerRS15}, consistency models in concurrent
computing~\cite{BonakdarpourSS18}, and distributed
synthesis~\cite{FinkbeinerHLST20}.

In the context of model checking of finite-state reactive systems, many temporal logics for hyperproperties have been proposed~\cite{DimitrovaFKRS12,ClarksonFKMRS14,BozzelliMP15,Rabe2016,FinkbeinerH16,CoenenFHH19,GutsfeldMO20}
for which model checking is decidable, including
$\HLTL$~\cite{ClarksonFKMRS14}, $\HCTLStar$~\cite{ClarksonFKMRS14},
$\HQPTL$~\cite{Rabe2016,CoenenFHH19}, and $\HPDL$~\cite{GutsfeldMO20}
which extend $\LTL$, $\CTLStar$, $\QPTL$~\cite{SistlaVW87}, and
$\PDL$~\cite{FischerL79}, respectively, by explicit first-order
quantification over traces and trace variables to refer to multiple
traces at the same time. The semantics of all these logics is \emph{synchronous}
and the temporal modalities are evaluated by a
lockstepwise traversal of all the traces assigned to the quantified
trace variables.

A different approach for the formalization of synchronous hyper logics is based on hyper variants
of monadic second-order logic over traces or trees~\cite{CoenenFHH19}. For the linear-time setting,
we recall the logic $\MSOE$~\cite{CoenenFHH19} (and its first-order fragment $\FOE$~\cite{Finkbeiner017})
which syntactically  extends monadic second-order logic of one successor $\MSO$ with the \emph{equal-level predicate}
$E$, which relate the same time points on different traces. Another class of hyperlogics is obtained
by adopting a \emph{team semantics} for standard temporal logics, in particular, $\LTL$~\cite{KrebsMV018,Luck20,VirtemaHFK021}.     \vspace{0.1cm}

\noindent \textbf{Asynchronous extensions of Hyper logics.} Many
application domains require asynchronous properties that relate traces
at distinct time points which can be arbitrarily far from each
other.
For example, asynchronous specifications are needed to reason about a
multithreaded environment in which threads are not scheduled in
lockstep, and traces associated with distinct threads progress at
different speed.
Asynchronous hyperproperties are also useful in information-flow
security where an observer is not time-sensitive, so the observer
cannot distinguish consecutive time points along an execution having
the same observation.
This again requires asynchronously matching sequences of observations
along distinct execution traces.
A first systematic study of asynchronous hyperproperties is done
in~\cite{GutsfeldOO21}, where two powerful and expressively equivalent
linear-time asynchronous formalisms are introduced: the temporal
fixpoint calculus $\HU$ and an automata-theoretic formalism where the
quantifier-free part of a specification is expressed by the class of
parity multi-tape Alternating Asynchronous Word Automata
(\AAWA)~\cite{GutsfeldOO21}.
While the expressive power of the quantifier-part of $\HLTL$ is just
that of $\LTL$ over tuples of traces of fixed arity
(\emph{multi-traces}), $\AAWA$ allow to specify very expressive
non-regular multi-trace properties.
As a matter of fact, model checking against $\HU$ or its $\AAWA$-based
counterpart is undecidable even for the quantifier alternation-free
fragment.
In~\cite{GutsfeldOO21}, two decidable subclasses of parity $\AAWA$ are
identified which express only multi-trace $\omega$-regular properties and lead to two $\HU$ fragments with decidable model
checking. 
More recently, other temporal logics~\cite{BaumeisterCBFS21,BozzelliPS21} which syntactically extend \HLTL\ have been introduced
for expressing  asynchronous hyperproperties.
\emph{Asynchronous $\HLTL$}
($\AHLTL$)~\cite{BaumeisterCBFS21}, useful for asynchronous security analysis,   models asynchronicity  by means of an additional quantification layer over the so called \emph{trajectories}.
Intuitively, a trajectory controls the relative speed at which traces progress by choosing
at each instant which traces move and which traces stutter.
The general logic also has an undecidable model-checking problem,
but~\cite{BaumeisterCBFS21} identifies practical decidable fragments,
and reports an empirical evaluation.
\emph{Stuttering $\HLTL$}
(\SHLTL) and \emph{Context $\HLTL$} (\CHLTL)
are introduced in~\cite{BozzelliPS21} as more expressive
asynchronous extensions of $\HLTL$.
$\SHLTL$ uses relativized versions of the temporal modalities with
respect to finite sets $\Gamma$ of $\LTL$ formulas.
Intuitively, these modalities are evaluated by a lockstepwise
traversal of the sub-traces of the given traces which are obtained by
removing ``redundant'' positions with respect to the pointwise
evaluation of the $\LTL$ formulas in $\Gamma$.
$\CHLTL$ extends $\HLTL$ by unary modalities $\tpl{C}$ parameterized by
a non-empty subset $C$ of trace variables---called the
\emph{context}---which restrict the evaluation of the temporal
modalities to the traces associated with the variables in $C$.
Both $\SHLTL$ and $\CHLTL$ are subsumed by $\HU$ and still have an
undecidable model-checking problem, and fragments of the two logics
with a decidable model-checking have been
investigated~\cite{BozzelliPS21}.\vspace{0.1cm} 

\noindent \textbf{Our contribution.} In this paper, we  study
expressiveness and decidability of 
asynchronous extensions of
$\HLTL$~\cite{GutsfeldOO21,BaumeisterCBFS21,BozzelliPS21}.
Our main goal is to compare the expressive power of these logics
together with the known logics for linear-time hyperproperties based
on the equal-level predicate whose most powerful representative is
$\MSOE$.
The first-order fragment $\FOE$ of $\MSOE$ is already strictly more
expressive than $\HLTL$~\cite{Finkbeiner017} and, unlike $\MSOE$, its
model-checking problem is decidable~\cite{CoenenFHH19}.
We obtain an almost complete expressiveness picture, summarized in
Figure~\ref{fig:Expressiveness}, where novel results are annotated in
red.
In particular, for $\AHLTL$, we show that although $\HLTL$ and
$\AHLTL$ are expressively incomparable,  $\HLTL$ can be embedded
into $\AHLTL$  using a natural encoding.
We also establish that 
$\AHLTL$ is
strictly less expressive than $\HU$ and its $\AAWA$
counterpart.
For the relative expressiveness of $\AHLTL$, $\SHLTL$, and $\CHLTL$,
we prove that $\AHLTL$ and $\SHLTL$ are expressively incomparable, and
that $\CHLTL$ is not subsumed by $\AHLTL$ or by $\SHLTL$.
The question of whether $\AHLTL$ and $\SHLTL$ are subsumed or not by
$\CHLTL$ remains open.
Additionally, we show that each of these logics is not subsumed by
$\MSOE$.
This last result solves a recent open
question~\cite{GutsfeldOO21,BaumeisterCBFS21}.

Since hyperproperties are a generalization of trace properties, we
also investigate the expressive power of the considered asynchronous
extensions of $\HLTL$ when interpreted on singleton sets of
traces.
For $\HLTL$ and its more expressive extension $\SHLTL$, singleton
models are just the ones whose traces are $\LTL$-definable and
checking the existence of such a model (\emph{single-trace
  satisfiability}) is decidable and \PSPACE-complete.
On the other hand, we show that for both $\AHLTL$ and $\CHLTL$,
single-trace satisfiability is highly undecidable being
$\Sigma_1^{1}$-hard.
Moreover, for $\CHLTL$ extended with past temporal modalities, we provide a nice
characterization of the singleton models which generalizes the
well-known equivalence of $\LTL$ and first-order logic $\FO$ over
traces established by Kamp's theorem.
We show that over singleton models, $\CHLTL$ with past corresponds to
the extension $\FOINPLUS$ of $\FO$ with addition over variables.

Finally, we investigate the decidability frontier for model-checking
$\CHLTL$ by enforcing the undecidability result of~\cite{BozzelliPS21}
and by identifying a maximal fragment of $\CHLTL$ for which model
checking is decidable.
This fragment subsumes $\HLTL$ and can be translated into $\FOE$. Due
to lack of space, many proofs are omitted and included in the
Appendix.

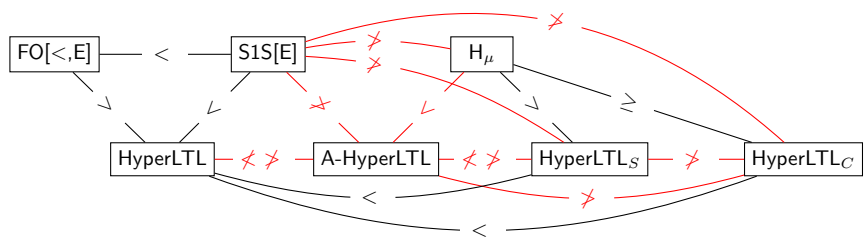
\begin{figure}[t]
    \centering
  \begin{center}
		{\begin{tikzpicture}
			[scale=1.0, bend angle = 15,  every node/.style={scale=0.8}]
	

			\node [cnode, node distance = 0em]
			(HU)
			{\normalsize  \,\,\,$\HU$\,\,\,};

            \node [cnode, draw = none, node distance = 7em]
		 	(Extra1)
		 	[below left of = HU]
		 	{ };	

            \node [cnode, node distance = 10em]
			(HLTL)
			[left of = Extra1]
			{\normalsize $\HLTL$};

            \node [cnode, node distance = 10em]
			(AHLTL)
			[right of = HLTL]
			{\normalsize $\AHLTL$};

            \node [cnode, node distance = 10em]
			(SHLTL)
			[right of = AHLTL]
			{\normalsize$\SHLTL$};

            \node [cnode, node distance = 10em]
			(CHLTL)
			[right of = SHLTL]
			{\normalsize$\CHLTL$};

           \node [cnode, node distance = 10em]
			(MSOE)
			[left of = HU]
			{\normalsize$\MSOE$};

            \node [cnode, node distance = 10em]
			(FOE)
			[left of = MSOE]
			{\normalsize$\FOE$};

           \path[red, thin,-] (HLTL) edge [] node[fill=white, anchor=center,  align=center, midway, sloped, font=\normalsize]
                { $\not < \,\, \not >$ }(AHLTL);

           \path[red, thin,-] (AHLTL) edge [] node[fill=white, anchor=center,  align=center, midway, sloped, font=\normalsize]
                { $\not < \,\, \not >$ }(SHLTL);

             \path[red, thin,-] (SHLTL) edge [] node[fill=white, anchor=center,  align=center, midway, sloped, font=\normalsize]
                { $ \,\not > \,$ }(CHLTL);

            \path[red, thin,-] (AHLTL) edge [] node[fill=white, anchor=center,  align=center, midway, sloped, font=\normalsize]
                { $ \,< \,$ }(HU);

            \path[black, thin,-] (SHLTL) edge [] node[fill=white, anchor=center,  align=center, midway, sloped, font=\normalsize]
                { $ \,> \,$ }(HU);

            \path[black, thin,-] (CHLTL) edge [] node[fill=white, anchor=center,  align=center, midway, sloped, font=\normalsize]
                { $ \,\geq \,$ }(HU);

            \path[black, thin,-] (HLTL) edge [] node[fill=white, anchor=center,  align=center, midway, sloped, font=\normalsize]
                { $ \,< \,$ }(MSOE);

            \path[red, thin,-] (AHLTL) edge [] node[fill=white, anchor=center,  align=center, midway, sloped, font=\normalsize]
                { $ \,\not > \,$ }(MSOE);

                \path[red, thin,-] (SHLTL) edge [bend right = 15] node[fill=white, anchor=center,  align=center, near end, sloped, font=\normalsize]
                { $ \,\not > \,$ }(MSOE);

            \path[red, thin,-] (MSOE) edge [bend left = 10] node[fill=white, anchor=center,  align=center, midway, sloped, font=\normalsize]
                { $ \,\not > \,$ }(HU);

             \path[red, thin,-] (AHLTL) edge [bend right = 15] node[fill=white, anchor=center,  align=center, midway, sloped, font=\normalsize]
                { $ \,\not > \,$ }(CHLTL);

             \path[black, thin,-] (HLTL) edge [bend right = 20] node[fill=white, anchor=center,  align=center, midway, sloped, font=\normalsize]
                { $ \,< \,$ }(CHLTL);

             \path[red, thin,-] (MSOE) edge [bend left = 30] node[fill=white, anchor=center,  align=center, midway, sloped, font=\normalsize]
                { $ \,\not > \,$ }(CHLTL);

               \path[black, thin,-] (HLTL) edge [bend right = 15] node[fill=white, anchor=center,  align=center, midway, sloped, font=\normalsize]
                { $ \,< \,$ }(SHLTL);

                 \path[black, thin,-] (HLTL) edge  node[fill=white, anchor=center,  align=center, midway, sloped, font=\normalsize]
                { $ \,> \,$ }(FOE);

                \path[black, thin,-] (MSOE) edge  node[fill=white, anchor=center,  align=center, midway, sloped, font=\normalsize]
                { $ \,< \,$ }(FOE);

			\end{tikzpicture} }
		\end{center}
\vspace{-0.4cm}
    \caption{Expressiveness comparisons between linear-time hyper logics}
    \label{fig:Expressiveness}
\end{figure}

\section{Preliminaries}
\label{sec-prelim}

Let $\N$ be the set of natural numbers.
For all $i,j\in\N$, $[i,j]$ denotes the set of natural numbers $h$
such that $i\leq h\leq j$.
Given a word $w$ over some alphabet $\Sigma$, $|w|$ is the length of $w$
($|w|=\infty$ if $w$ is infinite). For each $0\leq i<|w|$, $w(i)$ is
the $(i+1)^{th}$ symbol of $w$, and $w^{i}$ is the suffix of $w$ from
position $i$, i.e., the word $w(i)w(i+1)\ldots$.

We fix a finite  set $\AP$ of atomic propositions.
A \emph{trace} is an infinite word over $2^{\AP}$.
A \emph{pointed trace} is a pair $(\pi,i)$ consisting of a trace $\pi$
and a position $i\in\N$.
%
%
Two traces $\pi$ and $\pi'$ are \emph{stuttering equivalent} if there
are two infinite sequences of positions $0=i_0<i_1 \ldots $ and
$0=i'_0<i'_1\ldots$ s.t.~for all $k\geq 0$ and for all
$\ell\in [i_k,i_{k+1}-1]$ and $\ell'\in
[i'_k,i'_{k+1}-1]$, 
$\pi(\ell)=\pi'(\ell')$. The trace $\pi'$ is a \emph{stuttering
  expansion} of the trace $\pi$ if there is an infinite sequence of
positions $0=i_0<i_1 \ldots $ such that for all $k\geq 0$ and
\mbox{for all $\ell\in [i_k,i_{k+1}-1]$, $\pi'(\ell) =\pi(k)$.}
\vspace{0.1cm}

\noindent \textbf{Kripke structures.} A \emph{Kripke
  structure $($over $\AP$$)$} is a tuple $\Ku=\tpl{S,S_0,E,V}$, where
$S$ is a finite set of states, $S_0\subseteq S$ is the set of initial states, $E\subseteq
S\times S$ is a transition relation 
and $V:S \rightarrow 2^{\AP}$ is an \emph{$\AP$-valuation} of the set of states. 
A \emph{path} of $\Ku$ is an infinite sequence of states $t_0,t_1,\ldots$ such that $t_0\in S_0$
and $(t_{i},t_{i+1})\in E$ for all $i\geq 0$. The Kripke structure $\Ku$ induces the set $\Lang(\Ku)$ of traces
of the form $V(t_0),V(t_1),\ldots$ such that $t_0,t_1,\ldots$ is a path of $\Ku$.\vspace{0.1cm}

\noindent \textbf{Relative Expressiveness.} In Sections~\ref{section:AsynchronousHLTLResults}-\ref{sec:OverallExpressiveness}, we compare the expressiveness of various logics for linear-time
hyperproperties. Let $\M$ be a set of models (in our case, a model is a set of traces),
and $\Logic$ and $\Logic'$ be two logical languages interpreted over models in $\M$.
Given two formulas $\varphi\in\Logic$ and $\varphi'\in\Logic'$, we say that
$\varphi$ and $\varphi'$ are \emph{equivalent} if for each model $M\in\M$, $M$
satisfies $\varphi$ iff $M$ satisfies $\varphi'$.
The language $\Logic$ \emph{is subsumed by} $\Logic'$, denoted $\Logic\leq \Logic'$,
if each formula in $\Logic$ has an equivalent formula in $\Logic'$.
The language $\Logic$ \emph{is strictly less expressive than}
$\Logic$ (written $\Logic<\Logic')$ 
if $\Logic\leq \Logic'$ and there is a
$\Logic'$-formula which has no equivalent in $\Logic$.
Finally, two logics $\Logic$ and
$\Logic'$ \emph{are expressively incomparable} if both $\Logic\not\leq
\Logic'$ and $\Logic'\not\leq \Logic$.\vspace{0.1cm}

\noindent\textbf{Linear-time hyper specifications.}
We consider an abstract notion of linear-time hyper specifications which are interpreted over sets of traces.
We fix a finite set $\Var$ of trace variables.
A \emph{pointed-trace assignment $\Pi$} is a partial mapping  over $\Var$  assigning
to each trace variable $x$ in its domain
$\Dom(\Pi)$ a pointed trace.
The assignment $\Pi$ is \emph{initial} if for each $x\in \Dom(\Pi)$, $\Pi(x)$ is of the form $(\pi,0)$ for some trace $\pi$.
 For a  variable $x\in \Var$  and a pointed trace $(\pi,i)$, we denote by $\Pi[x\mapsto (\pi,i)]$ the pointed-trace  assignment having domain $\Dom(\Pi)\cup \{x\}$ defined as: $\Pi[x\mapsto (\pi,i)](x)=(\pi,i)$ and $\Pi[x\mapsto (\pi,i)](y)=\Pi(y)$  if $y\neq x$.

A \emph{multi-trace specification} $\S(x_1,\ldots,x_n)$ is a specification (in some formalism) parameterized by a subset $\{x_1,\ldots,x_n\}$ of $\Var$
whose semantics is given by a set $\Upsilon$ of \emph{initial} pointed-trace assignments with domain $\{x_1,\ldots,x_n\}$:  we write $\Pi\models S(x_1,\ldots,x_n)$ for the trace assignments $\Pi$ in $\Upsilon$.
Given a class $\C$ of multi-trace specifications,  \emph{linear-time hyper expressions} $\xi$ over $\C$ are defined as:
 $
\xi \DefinedAs    \exists x.  \xi \ | \ \forall x.  \xi \ | \ S(x_1,\ldots,x_n)
$,
 where $x,x_1,\ldots,x_n\in \Var$, $S(x_1,\ldots,x_n)\in\C$, and  $\exists x$ (resp., $\forall x$) is the
\emph{hyper} existential (resp., universal) trace quantifier for variable $x$.
An expression $\xi$ is a \emph{sentence} if every variable $x_i$ in the multi-trace specification $S(x_1,\ldots,x_n)$ of $\xi$ is \emph{not free} (i.e., $x_i$ is in the scope of a quantifier for  variable $x_i$).
The \emph{quantifier alternation depth} of $\xi$ is the number of switches between $\exists$ and $\forall$ quantifiers in the quantifier prefix of $\xi$. For a set $\Lang$ of traces  and an initial pointed-trace assignment $\Pi$ such that $\Dom(\Pi)$ contains the free variables of $\xi$ and the traces referenced by $\Pi$ are in $\Lang$, the satisfaction relation $(\Lang,\Pi)\models \xi$ is inductively defined as follows:
  \[ \begin{array}{ll}
  (\Lang,\Pi) \models  \exists x. \xi  &  \Leftrightarrow \text{ for some trace } \pi\in\Lang:\,  (\Lang,\Pi[x\mapsto (\pi,0)]) \models  \xi \\
    (\Lang,\Pi) \models  \forall x. \xi  &  \Leftrightarrow \text{ for each trace } \pi\in\Lang:\,  (\Lang,\Pi[x\mapsto (\pi,0)]) \models  \xi \\
(\Lang,\Pi) \models S(x_1,\ldots,x_n)  &  \Leftrightarrow \Pi\models S(x_1,\ldots,x_n)
\end{array}
\]
For a sentence $\xi$, we write $\Lang\models \xi$ to mean that
$(\Lang,\Pi_\emptyset)\models \xi$, where $\Pi_\emptyset$ is the empty
assignment.
If $\Lang\models \xi$ we say that $\Lang$ is a \emph{model} of
$\xi$.
If, additionally, $\Lang$ is a singleton we call  it a
\emph{single-trace model}.
By restricting our attention to the single-trace models, a
linear-time hyper sentence $\xi$ denotes a trace property consisting
of the traces $\pi$ such that $\{\pi\}\models \xi$. For a class $\C$ of
multi-trace specifications, we consider the following decision
problems:
\begin{compactitem}
\item the \emph{satisfiability} (resp., \emph{single-trace
    satisfiability}) problem is checking for a linear-time hyper
  sentence $\xi$ over $\C$, whether $\xi$ has a model (resp., a
  single-trace model), and
\item the model checking problem is checking for a Kripke structure
  $\Ku$ and a linear-time hyper sentence $\xi$ over $\C$, whether
  $\Lang(\Ku)\models \xi$.
\end{compactitem}

For instance, $\HLTL$ formulas are linear-time hyper sentences over
the class of multi-trace specifications, called \emph{$\HLTL$
  quantifier-free formulas}, obtained by standard $\LTL$
formulas~\cite{Pnueli77} by replacing atomic propositions $p$ with
relativized versions $p[x]$, where $x\in\Var$.
Intuitively, $p[x]$ asserts that $p$ holds at the current position of
the trace assigned to $x$. Given an $\HLTL$ quantifier-free  formula
$\psi(x_1,\ldots,x_n)$, an initial pointed trace assignment $\Pi$ such
that $\Dom(\Pi)\supseteq \{x_1,\ldots,x_n\}$, and a position
$i\geq 0$, the satisfaction relation $(\Pi,i)\models \psi$ is defined
as a natural extension of the satisfaction relation
$(\pi,i)\models \theta$ for $\LTL$ formulas $\theta$ and traces
$\pi$.
In particular,
\begin{compactitem}
\item $(\Pi,i)\models p[x_k]$ if  $p\in \Pi(x_k)(i)$,
\item $(\Pi,i)\models \Next\psi$ if $(\Pi,i+1)\models \psi$, and 
\item $(\Pi,i)\models \psi_1\Until\psi_2$ if there is $j\geq i$ such
  that $(\Pi,j)\models \psi_2$ and $(\Pi,k)\models \psi_1$ for all
  $k\in [i,j-1]$.
\end{compactitem}
\vspace{0.2cm}

\noindent \textbf{Asynchronous Word Automata and the Fixpoint Calculus $\HU$.} We shortly recall  the framework of parity alternating asynchronous word automata (parity $\AAWA$)~\cite{GutsfeldOO21}, a class of finite-state automata for the asynchronous traversal of multiple infinite words.
Intuitively, given $n\geq 1$, an $\AAWA$ with $n$ tapes ($\NAAWA$) has access to $n$ infinite words over the input alphabet  $\Sigma$ and at each step, activates multiple copies where for each of them, there is exactly one input word whose current input symbol is consumed (i.e., the reading head of such word moves one position to the right).
 In particular, the target of a move   of $\Au$ is encoded by a pair $(q,i)$, where $q$ indicates the target state while the direction $i\in[1,n]$ indicates on which input word to progress. Details on the syntax and semantics of $\AAWA$ are given in Appendix~\ref{APP:AAWA}. 
     We denote by \emph{Hyper
$\AAWA$} the class of linear-time hyper sentences over the multi-trace specifications given by parity $\AAWA$.
 We also consider the fixpoint calculus $\HU$ introduced in~\cite{GutsfeldOO21} that provides a logical characterization
  of  Hyper $\AAWA$.

\section{Advances in Asynchronous Extensions of $\HLTL$}\label{section:AsynchronousHLTLResults}

In this section, we investigate expressiveness and decidability issues
on known asynchronous extensions of $\HLTL$, namely,
\emph{Asynchronous $\HLTL$}~\cite{BaumeisterCBFS21}, \emph{Stuttering $\HLTL$}~\cite{BozzelliPS21}, and
\emph{Context $\HLTL$}~\cite{BozzelliPS21}.

\subsection{Results for Asynchronous $\HLTL$ (\AHLTL)}
\label{sec:AsynchronousHLTL}

We first recall $\AHLTL$~\cite{BaumeisterCBFS21}, a syntactical extension of $\HLTL$ which
allows to express pure asynchronous hyperproperties.
Then, we show that although $\AHLTL$ does not subsume $\HLTL$, $\HLTL$
can be embedded into $\AHLTL$ by means of an additional
proposition.
Second, we establish that $\AHLTL$ is subsumed by Hyper $\AAWA$ and
the fixpoint calculus $\HU$.
Finally, we show that unlike $\HLTL$, single-trace satisfiability of
$\AHLTL$ is undecidable.

\noindent The logic $\AHLTL$ models the asynchronous passage of time
between computation traces using the notion of a trajectory.
Given a non-empty subset $V\subseteq \Var$, a
\emph{trajectory over $V$} is an infinite sequence $t$ of non-empty
subsets of $V$.
Intuitively, the positions $i\geq 0$ along $t$ model
the global time flow and for each position $i\geq 0$, $t(i)$
determines the trace variables in $V$ whose associated traces make
progress at time $i$.  The trajectory $t$ is \emph{fair} if for each
$x\in V$, there are infinitely many positions $i$ such that
$x\in t(i)$.
%

$\AHLTL$ formulas are linear-time hyper sentences over multi-trace
specifications $\psi$, called \emph{$\AHLTL$ quantifier-free
  formulas}, where $\psi$ is of the form $\ET \theta$ or $\AT \theta$
and $\theta$ is a $\HLTL$ quantifier-free formula: $\ET$ is the
existential trajectory modality and $\AT$ is the universal trajectory
modality.
Given a pointed trace assignment $\Pi$ and a trajectory $t$ over
$\Dom(\Pi)$, the \emph{successor of $(\Pi,t)$}, denoted by
$(\Pi,t) + 1$, is defined as $(\Pi',t')$, where:
\begin{inparaenum}[(1)]
\item   $t'$ is the trajectory $t^{1}$ (the suffix of $t$ from position $1$), and
\item $\Dom(\Pi')=\Dom(\Pi)$ and for each $x\in\Dom(\Pi)$ with $\Pi(x)=(\pi,i)$,
  $\Pi'(x)=(\pi,i+1)$ if $x\in t(0)$, and $\Pi'(x)=\Pi(x)$ otherwise.
 \end{inparaenum}

For each $k\geq 1$, we write $(\Pi,t)+k$ for denoting the pair $(\Pi'',t'')$ obtained by $k$-applications of the
successor function starting from $(\Pi,t)$.
Given a $\HLTL$ quantifier-free formula $\theta$ such that $\Dom(\Pi)$ contains the  variables
occurring in $\theta$, the satisfaction relations $\Pi   \models \ET\theta$, $\Pi   \models \AT\theta$, and   $(\Pi,t)\models \theta$ are defined as follows (we omit the semantics of the Boolean
connectives): 
  \[ \begin{array}{r@{\;\;}c@{\;\;}ll}
    \Pi   &\models& \ET\theta  &  \Leftrightarrow   \text{ for some fair trajectory $t$ over $\Dom(\Pi)$, }  (\Pi,t)\models \theta\\
      \Pi   &\models& \AT\theta  &  \Leftrightarrow   \text{ for all fair trajectories $t$ over $\Dom(\Pi)$, }  (\Pi,t)\models \theta\\
 (\Pi,t)  &\models& \Rel{p}{x}  &  \Leftrightarrow  \Pi(x)=(\pi,i) \text{ and }p\in \pi(i)\\
   (\Pi,t)  &\models&  \Next\theta &  \Leftrightarrow  (\Pi,t)+1 \models  \theta\\
   (\Pi,t)  &\models&  \theta_1\Until \theta_2  &
  \Leftrightarrow  \text{for some }i\geq 0:\,   (\Pi,t)+i \models  \theta_2 \text{ and }  (\Pi,t)+k \models  \theta_1 \text{ for all } 0\leq k<i
\end{array}
\]
%
We also exploit an alternative characterization of the semantics of quantifier-free $\AHLTL$ formulas which easily follows from
the definition of trajectories.

\begin{proposition}
  \label{prop:CharacterizationAHLTLSemantics}
  Let $\theta$ be a quantifier-free $\HLTL$ formula over trace
  variables $x_1,\ldots,x_k$, and let $\pi_1,\ldots,\pi_k$ be $k$
  traces. Then:
\begin{compactitem}
\item
  $\{x_1 \leftarrow (\pi_1,0),\ldots,x_k \leftarrow (\pi_k,0)\}\models
  \ET\theta$ iff for all $i\in [1,k]$, there is a stuttering expansion
  $\pi'_i$ of $\pi_i$ such that
  $\{x_1 \leftarrow (\pi'_1,0),\ldots,x_k \leftarrow
  (\pi'_k,0)\}\models \theta$.
\item
  $\{x_1 \leftarrow (\pi_1,0),\ldots,x_k \leftarrow (\pi_k,0)\}\models
  \AT\theta$ iff for all $i\in [1,k]$ and for all stuttering
  expansions $\pi'_i$ of $\pi_i$, it holds that
  $\{x_1 \leftarrow (\pi'_1,0),\ldots,x_k \leftarrow
  (\pi'_k,0)\}\models \theta$.
\end{compactitem}
\end{proposition}

\noindent \textbf{$\AHLTL$ versus  $\HLTL$.}
We now show that, unlike other temporal logics for asynchronous
hyperproperties (see Sections~\ref{sec:StutteringHLTL}
and~\ref{sec:ContextHyper}), $\AHLTL$ does not subsume $\HLTL$.
Given an atomic proposition $p$, we consider the following linear-time
hyperproperty.

\noindent\fbox{
\parbox{0.97\textwidth}{
  \textbf{$p$-synchronicity:} a set $\Lang$ of traces  satisfies the  $p$-synchronicity hyperproproperty if for all traces   $\pi,\pi'\in\Lang$ and positions $i\geq 0$, $p\in \pi(i)$ iff $p\in \pi'(i)$.}
}\vspace{0.1em}

This hyperproperty can be expressed in $\HLTL$ as follows:
$\forall x_1.\,\forall x_2.\, \Always (p[x_1] \leftrightarrow
p[x_2])$.
However, it cannot be expressed in $\AHLTL$ (for details, see
Appendix~\ref{APP:AHLTLinexpressibilitypSynchronicity}).

\begin{theorem}\label{theo:AHLTLinexpressibilitypSynchronicity}
  $\AHLTL$ cannot express $p$-synchronicity. Hence, $\AHLTL$ does not
  subsume $\HLTL$.
\end{theorem}

\noindent Though $\AHLTL$ does not subsume $\HLTL$, we can embed $\HLTL$
into $\AHLTL$ by using an additional proposition $\#\notin \AP$ as
follows.
We can ensure that along a trajectory, traces progress at each global
instant by requiring that proposition $\#$ holds exactly at the even
positions.
Formally, given a trace $\pi$ over $\AP$, we denote by $enc_\#(\pi)$
the trace over $\AP\cup\{\#\}$ defined as:
$enc_\#(\pi)(2i)=\pi(2i)\cup \{\#\}$ and $enc_\#(\pi)(2i+1)=\pi(2i+1)$
for all $i\geq 0$. We extend the encoding $enc_\#$ to sets of traces
$\Lang$ and assignments $\Pi$ over $\AP$ in the obvious
way.
For each $x\in\Var$, let $\theta_\#(x)$ be the following one-variable
quantifier-free $\HLTL$ formula:
$ \#[x]\wedge \Always (\#[x] \leftrightarrow \neg \Next \#[x])$.
It is easy to see that for a trace $\rho$ over $\AP\cup \{\#\}$, a
stuttering expansion $\rho'$ of $\rho$ satisfies $\theta_\#(x)$ with
$x$ bound to $\rho'$ iff $\rho'=\rho$ and $\rho$ is the $\#$-encoding
of some trace over $\AP$.
It follows that satisfiability of an $\HLTL$ formula $\varphi$ can be
reduced in linear-time to the satisfiability of the $\AHLTL$ formula
$\varphi_\#$ obtained from $\varphi$ by replacing the quantifier-free
part $\psi(x_1,\ldots,x_k)$ of $\varphi$ with the quantifier-free
$\AHLTL$ formula 
$\ET.\,(\psi(x_1,\ldots,x_k)\wedge \bigwedge_{i\in
  [1,k]}\theta_\#(x_i))$.
For model checking, given a Kripke structure $\Ku=\tpl{S,S_0,E,V}$, we
construct in linear-time a Kripke structure $\Ku_\#$ over
$\AP\cup \{\#\}$ such that
$\Lang(\Ku_\#)=enc_\#(\Lang(\Ku))$. Formally,
$\Ku_\#=\tpl{S\times \{0,1\},$ $S_0\times \{1\},E',V'}$ where
$E'=\{((s,b),(s',1-b))\mid (s,s')\in E \text{ and }b=0,1\}$,
$V'((s,1))=V(s)\cup\{\#\}$ and $V'((s,0))=V(s)$ for all $s\in
S$. Thus, we obtain the following
result. 

\begin{theorem} Satisfiability (resp., model checking) of $\HLTL$ can
  be reduced in linear-time to satisfiability (resp., model checking)
  of $\AHLTL$.
\end{theorem}

\noindent \textbf{$\AHLTL$ versus Hyper $\AAWA$ and $\HU$.}
We show that $\AHLTL$ is subsumed by Hyper $\AAWA$ and $\HU$.
To this purpose, we exhibit an exponential-time translation of
quantifier-free $\AHLTL$ formulas into equivalent parity \AAWA.

\begin{theorem}\label{theo:FromAHLTLtoAAWA}
  Given an $\AHLTL$ quantifier-free formula $\psi$ with trace
  variables $x_1,\ldots,x_n$, one can build in singly exponential time
  a parity $\NAAWA$ $\Au_{\psi}$ over $2^{\AP}$ accepting the set of
  $n$-tuples $(\pi_1,\ldots,\pi_n)$ of traces such that
  $(\{x_1\leftarrow (\pi_1,0),\ldots,x_n\leftarrow (\pi_n,0)\})\models
  \psi$.
\end{theorem}

\begin{proof}[Proof sketch] We first assume that $\psi$ is of the form
  $\ET\theta$ for some $\HLTL$ quantifier-free formula $\theta$.
  By an adaptation of the standard automata theoretic approach for
  \LTL~\cite{VardiW94}, we construct a \emph{nondeterministic}
  $\NAAWA$ ($\NNAWA$) $\Au_{\ET\theta}$ equipped with standard
  generalized B\"{u}chi acceptance conditions which accepts a
  $n$-tuple $(\pi_1,\ldots,\pi_n)$ of traces iff there is a fair
  trajectory $t$ such that
  $(\{x_1\leftarrow (\pi_1,0),\ldots,x_n\leftarrow
  (\pi_n,0)\}),t\models \theta$.
  By standard arguments, a generalized B\"{u}chi $\NNAWA$ can be
  converted in quadratic time into an equivalent parity $\NNAWA$.
  The behaviour of the automaton is subdivided into phases where each
  phase intuitively corresponds to a global timestamp.
  During a phase, $\Au_{\ET\theta}$ keeps tracks in its state of the
  guessed set of subformulas of $\theta$ that hold at the current
  global instant and guesses which traces progress at the next global
  instant by moving along a non-empty guessed set of directions in
  $\{1,\ldots,n\}$ in turns. In particular, after a movement along
  direction $i$, the automaton keeps track in its state of the
  previous chosen direction $i$ and \emph{either} moves to the next
  phase, \emph{or} remains in the current phase by choosing a
  direction $j>i$.  The transition function in moving from the end of
  a phase to the beginning of the next phase captures the semantics of
  the next modalities and the `local' fixpoint characterization of the
  until modalities.  Moreover, the generalized B\"{u}chi acceptance
  condition is used for ensuring the fulfillment of the liveness
  requirements $\theta_2$ in the until sub-formulas
  $\theta_1 \Until \theta_2$, and for guaranteeing that the guessed
  trajectory is fair (i.e., for each direction $i\in [1,n]$, the
  automaton moves along $i$ infinitely often). Details of the
  construction are given in Appendix~\ref{APP:FromExistentialAHLTL}.

Now, let  us consider a quantifier-free $\AHLTL$ formula of the form $\AT\theta$
with trace variables $x_1,\ldots,x_n$, and let $\Au_{\ET\neg\theta}$ be the parity $\NAAWA$ associated with the formula
$\ET\neg\theta$. By~\cite{GutsfeldOO21}, one can construct in linear-time (in the size of $\Au_{\ET\neg\theta}$), a parity
$\NAAWA$ $\Au_{\AT\theta}$ accepting the complement of the language of $n$-tuples of traces accepted by $\Au_{\ET\neg\theta}$.
\end{proof}

\noindent Thus, being $\HU$ and Hyper $\AAWA$  expressively equivalent, we obtain the following result.

 \begin{corollary}\label{cor:FromAHLTLToHU}
Hyper $\AAWA$  subsumes $\AHLTL$. $\HU$ also subsumes $\AHLTL$.
\end{corollary}

\noindent \textbf{Undecidability of single-trace satisfiability for $\AHLTL$.}
It is easy to see that for $\HLTL$, single-trace satisfiability
corresponds to $\LTL$ satisfiability (hence, it is \PSPACE-complete).
We show now that for $\AHLTL$, the problem is highly undecidable being at least 
$\Sigma_{1}^{1}$-hard.
The crucial observation is that we can enforce alignment
requirements on distinct stuttering expansions of the same trace which
allow to encode recurrent computations of \emph{non-deterministic
  $2$-counter machines}~\cite{Harel86}. Recall that such a machine is
a tuple $M = \tpl{Q,\Delta,\delta_\init,\delta_\rec}$, where $Q$ is a
finite set of (control) locations,
$\Delta \subseteq Q\times \Inst \times Q$ is a transition relation
over the instruction set $\Inst= \{\inc,\dec,\zero\}\times \{1,2\}$,
and $\delta_\init\in \Delta$ and $\delta_\rec\in \Delta$ are two
designated transitions, the initial and the recurrent one.

An $M$-configuration is a pair $(\delta,\nu)$ consisting of a transition $\delta\in \Delta$ and a counter valuation $\nu: \{1,2\}\to \N$. A  computation of $M$ is an \emph{infinite} sequence of configurations of the form $((q_0,(\instr_0,c_0),q_1),\nu_0),((q_1,(\instr_1,c_1),q_2),\nu_1),\ldots$  such that for each $i\geq 0$:
 %
\begin{compactitem}
\item    $\nu_{i+1}(3-c_i)= \nu_i(3-c_i)$ ;
\item  $\nu_{i+1}(c_i)= \nu_i(c_i) +1$ if $\instr_i=\inc$, and $\nu_{i+1}(c_i)= \nu_i(c_i) -1$ if $\instr_i=\dec$;
\item  $\nu_{i+1}(c_i)= \nu_i(c_i)=0$ if $\instr_i= \zero$.
\end{compactitem}

The recurrence problem is to decide whether for a given machine
$M$, there is a computation starting at the initial configuration
$(\delta_\init,\nu_0)$, where $\nu_0(c)=0$ for each $c\in \{1,2\}$,
which visits $\delta_\rec$ infinitely often.
This problem is known to be $\Sigma_{1}^{1}$-complete~\cite{Harel86}.

\begin{theorem}\label{theo:UndecidabilitySIngleTraceAHLTL}
The single-trace satisfiability problem of $\AHLTL$ is at least $\Sigma_{1}^{1}$-hard.
\end{theorem}
\begin{proof}[Proof sketch]
  Let $M = \tpl{Q,\Delta,\delta_\init,\delta_\rec}$ be a counter
  machine.
  We construct a two-variable $\AHLTL$ formula $\varphi_M$ such that
  $M$ is a positive instance of the recurrence problem if and only if
  $\varphi_M$ has a single-trace model.
  %
  %
  The set of atomic propositions is
  $\AP\DefinedAs \Delta \cup \{c_1,c_2,\#,\Beg,\Pad\}$.
  Intuitively, propositions $c_1$ and $c_2$ are used to encode the
  values of the two counters in $M$ and $\#$ is used to ensure that
  the values of the counters are not modified in the stuttering
  expansions of a trace encoding a computation of $M$.
  Proposition $\Beg$ marks the beginning of a configuration
  code, and proposition $\Pad$ is exploited for encoding a padding
  word at the end of a configuration code: formula $\varphi_M$ will
  ensure that only these words can be ``expanded'' in the stuttering
  expansions of a trace. Formally, an $M$-configuration $(\delta,\nu)$
  is encoded by the finite words over $2^{\AP}$ (called
  \emph{segments}) of the form
  $\{\Beg,\delta\} P_1\ldots P_m \{\Pad\}^{k}$, where $k\geq 1$,
  $m= \max(\nu(1),\nu(2))$, and for all $i\in [1,m]$,
  %
  \begin{compactitem}
  \item $\emptyset\neq P_i\subseteq \{\#,c_1,c_2\}$,
  \item $\#\in P_i$ iff $i$ is odd, and
  \item for all $\ell\in\{1,2\}$, $c_\ell\in P_i$ iff
    $i\leq \nu(\ell)$.
  \end{compactitem}
  A computation $\rho$ of $M$ is then encoded by the traces obtained
  by concatenating the codes of the configurations along $\rho$
  starting from the first one.
  The $\AHLTL$ formula $\varphi_M$ is given by
  $ \exists x_1\exists x_2.\,\ET\, \psi$, where the quantifier-free
  $\HLTL$ formula $\psi$ guarantees that for the two stuttering
  expansions $\pi_1$ and $\pi_2$ of the given trace $\pi$, the
  following holds:
\begin{compactitem}
  \item both $\pi_1$ and $\pi_2$ are infinite concatenations of segments;
  \item the first segment of $\pi_1$ encodes the initial configuration $(\delta_\init,\nu_0)$ of $M$ and the second segment of $\pi_1$
  encodes a configuration which is a successor of  $(\delta_\init,\nu_0)$ in $M$;
  \item $\delta_\rec$ occurs infinitely often along $\pi_1$;
  \item for each $i\geq 2$, the $(i+1)^{th}$ segment $s_2$ of $\pi_2$ starts at the same position as the
  $i^{th}$ segment $s_1$ of $\pi_1$. Moreover, $s_1$ and $s_2$ have the same length and the configuration encoded by $s_2$ is a successor in $M$ of the configuration encoded by $s_1$.
\end{compactitem}
Now, since $\pi_1$ and $\pi_2$ are stuttering expansions of the same
trace $\pi$, the alternation requirement for proposition $\#$ in the
encoding of an $M$-configurations ensures that $\pi_1$ and $\pi_2$
encode the same infinite sequence of $M$-configurations. Hence,
$\varphi_M$ has a single-trace model if and only if $M$ is a positive
instance of the recurrence problem.  Details appear in
Appendix~\ref{APP:UndecidabilitySIngleTraceAHLTL}.
\end{proof} 
\subsection{Results for Stuttering $\HLTL$ (\SHLTL)}%
\label{sec:StutteringHLTL}

Stuttering $\HLTL$ (\SHLTL)~\cite{BozzelliPS21} is an asynchronous extension of $\HLTL$
obtained by using \emph{stutter-relativized} versions of the temporal
modalities w.r.t.~finite sets $\Gamma$ of $\LTL$ formulas.
In this section, we show that $\AHLTL$ and $\SHLTL$ are expressively incomparable.

In $\SHLTL$, the notion of successor of a position $i$ along a trace
$\pi$ is relativized using a finite set $\Gamma$ of $\LTL$
formulas.
If in the interval $[i,\infty[$, the truth value of each
formula in $\Gamma$ does not change along $\pi$ (i.e., for each
$j\geq i$ and for each $\theta\in\Gamma$, $(\pi,i)\models\theta$ iff
$(\pi,j)\models\theta$), then the \emph{$\Gamma$-successor of $i$ in
  $\pi$} coincides with the local successor $i+1$.
Otherwise, the \emph{$\Gamma$-successor of $i$ in $\pi$} is the
smallest position $j>i$ such that the truth value of some formula
$\theta$ in $\Gamma$ changes in moving from $i$ to $j$ (i.e., for some
$\theta\in\Gamma$, $(\pi,i)\models\theta$ iff
$(\pi,j)\not\models\theta$).
The $\Gamma$-successor induces a trace, called \emph{$\Gamma$-stutter
  trace of $\pi$} and denoted by $\stfr_{\Gamma}(\pi)$, obtained from
$\pi$ by repeatedly applying the $\Gamma$-successor starting from
position $0$, i.e.
$\stfr_{\Gamma}(\pi)\DefinedAs \pi(i_0)\pi(i_1)\ldots$, where $i_0=0$
and $i_{k+1}$ is the $\Gamma$-successor of $i_k$ in $\pi$ for all
$k\geq 0$.
Note that $\stfr_{\Gamma}(\pi)=\pi$ if $\Gamma=\emptyset$.
Given a pointed-trace assignment $\Pi$, the \emph{$\Gamma$-successor
  $\SUCC_\Gamma(\Pi)$ of $\Pi$} is the  pointed
trace-assignment with domain $\Dom(\Pi)$ defined as follows for each $x\in\Dom(\Pi)$:
 if $\Pi(x)=(\pi,i)$, then
$\SUCC_\Gamma(\Pi)=(\pi,\ell)$ where $\ell$ is the $\Gamma$-successor
of $i$ in $\pi$.  For each $j\in \N$, we use $\SUCC^{\,j}_\Gamma$ for
the function obtained by $j$ applications of the function
$\SUCC_\Gamma$.

$\SHLTL$ formulas are linear-time hyper sentences over multi-trace
specifications $\psi$, called \emph{$\SHLTL$ quantifier-free
  formulas}, where the syntax of $\psi$ is as
follows: 
 %
\[
  \psi ::=    \top \ | \  \Rel{p}{x}  \ | \ \neg \psi \ | \ \psi \wedge \psi \ | \ \Next_\Gamma \psi  \ | \ \psi \Until_\Gamma \psi
\]
where $p\in \AP$, $x\in \Var$, $\Gamma$ is a finite set of $\LTL$
formulas over $\AP$, and $\Next_\Gamma$ and $\Until_\Gamma$ are the
stutter-relativized versions of the $\LTL$ temporal modalities.
Informally,  the relativized temporal modalities
$\Next_\Gamma$ and $\Until_\Gamma$ are evaluated by a lockstepwise
traversal of the $\Gamma$-stutter traces associated with the currently
quantified traces.
Standard $\HLTL$ corresponds to the fragment of
$\SHLTL$ where the subscript of each temporal modality is the empty set $\emptyset$.

Given a $\SHLTL$ quantifier-free formula $\psi$ and a pointed trace
assignment $\Pi$ such that $\Dom(\Pi)$ contains the trace variables
occurring in $\psi$, the satisfaction relation $\Pi\models \psi$ is
inductively defined as follows (we omit the semantics of the Boolean
connectives): 
  \[ \begin{array}{ll}
 \Pi  \models \Rel{p}{x}  &  \Leftrightarrow  \Pi(x)=(\pi,i) \text{ and }p\in \pi(i)\\
   \Pi  \models  \Next_\Gamma\psi &  \Leftrightarrow  \SUCC_\Gamma(\Pi) \models  \psi\\
   \Pi  \models  \psi_1\Until_\Gamma \psi_2  &
  \Leftrightarrow  \text{for some }i\geq 0:\,   \SUCC^{\,i}_\Gamma(\Pi) \models  \psi_2 \text{ and }  \SUCC^{\,k}_\Gamma(\Pi) \models  \psi_1 \text{ for all } 0\leq k<i
\end{array} \]

\emph{Stuttering $\LTL$} formulas, corresponding to one-variable
$\SHLTL$ quantifier-free formulas, can be translated in polynomial
time into equivalent $\LTL$ formulas (see~\cite{BozzelliPS21}).
Thus, since $\LTL$ satisfiability is \PSPACE-complete, the following result 
holds.
%

\begin{proposition}\label{prop:SingleTraceSHLTL} The trace properties
  definable by $\SHLTL$ formulas are $\LTL$ definable, and
  single-trace satisfiability of $\SHLTL$
is \PSPACE-complete.
\end{proposition}


\noindent\textbf{$\SHLTL$ versus $\AHLTL$.}
We show that $\SHLTL$ and $\AHLTL$ are expressively incomparable even
over singleton sets of atomic propositions.
 By Theorem~\ref{theo:AHLTLinexpressibilitypSynchronicity}, unlike $\SHLTL$,
 $\AHLTL$ does not subsume $\HLTL$ even when $|\AP|=1$. Hence, $\SHLTL$ is not subsumed by $\AHLTL$ even when $|\AP|=1$.
 We show now that the converse holds as well. Intuitively,  $\AHLTL$ can encode counting mechanisms which
 cannot be expressed in $\SHLTL$.
 Let $\AP=\{p\}$.
 We exhibit two families $\{\Lang_n\}_{n\geq 1}$ and $\{\Lang'_n\}_{n\geq 1}$
 of trace sets and an $\AHLTL$ formula $\varphi_A$ such that
 \begin{compactitem}
 \item $\varphi_A$ can distinguish the traces set $\Lang_n$ and
   $\Lang'_n$ for each $n\geq 1$, but
 \item for each $\SHLTL$ formula $\psi$, there is $n$ such that
   $\psi$ does not distinguish $\Lang_n$ and $\Lang'_n$.
 \end{compactitem}
 For each $n\geq 1$, let $\pi_n$, $\rho_n$, and $\rho'_n$ be the traces defined as:\vspace{0.1cm}

$
 \pi_n \DefinedAs ( \emptyset\cdot p )^{n}\cdot \emptyset^{\omega} \hspace{5em}
 \rho_n \DefinedAs  ( \emptyset\cdot p )^{2n}\cdot \emptyset^{\omega} \hspace{5em}
  \rho'_n \DefinedAs  ( \emptyset\cdot p )^{2n+1}\cdot \emptyset^{\omega}
$\vspace{0.1cm}

\noindent For each $n\geq 1$, define $\Lang_n\DefinedAs\{\pi_n,\rho_n\}$ and   $\Lang'_n\DefinedAs\{\pi_n,\rho'_n\}$.
Let $\psi_1(x)$ and $\psi_2(x)$ be two one-variable quantifier-free
$\HLTL$ formulas capturing the following requirements:
\begin{compactitem}
\item $\psi_1(x)$ captures traces of the form
  $( \emptyset\cdot p )^{k}\cdot \emptyset^{\omega}$ for some
  $k\geq 1$,
\item $\psi_2(x)$ captures traces of the form
  $( \emptyset^{2}\cdot p^{2} )^{k}\cdot \emptyset^{\omega}$ for some
  $k\geq 1$.
\end{compactitem}
Let $\psi(x,y)$ be the two-variable quantifier-free $\HLTL$ formula
defined as follows:
\[
\psi(x,y)\DefinedAs \Eventually(p[x]\wedge p[y] \wedge \Next\Always (\neg p[x]\wedge \neg p[y]))
\]
\noindent Intuitively, if $x$ is bound to a trace $\nu_1$ satisfying
$\psi_1 $ and $y$ is bound to a trace $\nu_2$ satisfying $\psi_2$,
then $\psi(x,y)$ holds iff $\nu_1$ is of the form
$ ( \emptyset\cdot p )^{2k}$ and $\nu_2$ is of the form
$ ( \emptyset^{2}\cdot p^{2} )^{k}$ for some $k\geq 1$. The $\AHLTL$
formula $\varphi_A$ is then defined as follows:
\[
 \varphi_A\DefinedAs \forall x_1.\forall x_2.\ET \,\bigl([\psi(x_1,x_2) \wedge \psi_1(x_1)\wedge \psi_1(x_2) ]\,\vee\,
    [\psi(x_1,x_2)\wedge \displaystyle{\bigvee_{i\in \{1,2\}}}(\psi_1(x_i)\wedge \psi_2(x_{3-i}))]\bigr)
    \]
Let $\pi$ a trace of the form $\pi =( \emptyset\cdot p )^{k}\cdot \emptyset^{\omega}$ for some $k\geq 1$.
We  observe that  the unique stuttering expansion $\nu_1$ of $\pi$ such that $\nu_1$ satisfies $\psi_1(x)$ is $\pi$ itself.
Similarly,  there is a unique stuttering expansion $\nu_2$ of $\pi$ such that $\nu_2$ satisfies $\psi_2(x)$, and such a trace
$\nu_2$ is given by $( \emptyset^{2}\cdot p^{2} )^{k}\cdot \emptyset^{\omega}$.
Fix $n\geq 1$. Let us consider the trace set $\Lang_n=\{\pi_n,\rho_n\}$. By construction, if both variables $x_1$ and $x_2$
in the definition of $\varphi_A$ are bound to the same trace $\pi$ in $\Lang_n$, then the first disjunct in the definition
of $\varphi_A$ is fulfilled by taking $\pi$ itself as an expansion of $\pi$. On the other hand, assume that
variable $x_1$ (resp., $x_2$) is bound to trace $\pi_n$ and variable $x_2$ (resp., $x_1$) is bound to trace $\rho_n$. In this case, by taking as stuttering expansion of $\pi_n$ the trace  $( \emptyset^{2}\cdot p^{2} )^{n}\cdot \emptyset^{\omega}$ and
as stuttering expansion of $\rho_n =( \emptyset\cdot p )^{2n}\cdot \emptyset^{\omega}$ the trace $\rho_n$ itself, the second disjunct in the definition of $\varphi_A$ is fulfilled. Hence, $\Lang_n$ is a model of $\varphi_A$.

Now, we show that $\Lang'_n=\{\pi_n,\rho'_n\}$ does not satisfy $\varphi_A$. Let us consider the mapping assigning to variable
$x_1$ the trace $\pi_n$ and to variable $x_2$ the trace $\rho'_n$. With this mapping, the quantifier-free part of $\varphi_A$ cannot be fulfilled. This because the unique stuttering expansion of $\pi_n= ( \emptyset\cdot p )^{n}\cdot \emptyset^{\omega}$ (resp.,
  $\rho'_n= ( \emptyset\cdot p )^{2n+1}\cdot \emptyset^{\omega}$) satisfying $\psi_1$ is $\pi_n$ (resp., $\rho'_n$) itself. Moreover,
  the unique stuttering expansion of $\pi_n$ (resp., $\rho'_n$) satisfying $\psi_2$ is $  ( \emptyset^{2}\cdot p^{2} )^{n}\cdot \emptyset^{\omega}$ (resp.,
  $ ( \emptyset^{2}\cdot p^{2} )^{2n+1}\cdot \emptyset^{\omega}$). Hence, for all $n\geq 1$, $\Lang_n\models \varphi_A$
and $\Lang'_n\not\models \varphi_A$. On the other hand, one can show that the following holds (for details, see Appendix~\ref{APP:InexpressivityCountingsHLTL}).


\newcounter{prop-InexpressivityCountingsHLTL}
\setcounter{prop-InexpressivityCountingsHLTL}{\value{proposition}}
\newcounter{sec-InexpressivityCountingsHLTL}
\setcounter{sec-InexpressivityCountingsHLTL}{\value{section}}

\begin{proposition}\label{prop:InexpressivityCountingsHLTL} For each $\SHLTL$ formula $\psi$, there is $n\geq 1$ s.t.
$\Lang_n\models \psi$ iff $\Lang'_n\models \psi$.
\end{proposition}

Thus, since $\AHLTL$ does not subsume
$\SHLTL$, we obtain the following result. 

\begin{corollary}\label{cor:IncomparabilityAHLTL_SHLTL} $\AHLTL$ and $\SHLTL$ are expressively incomparable.
\end{corollary}

\subsection{Results for Context $\HLTL$ (\CHLTL)}\label{sec:ContextHyper}

Context $\HLTL$ ($\CHLTL$)~\cite{BozzelliPS21} extends $\HLTL$ by
unary modalities $\tpl{C}$ parameterized by a non-empty subset $C$ of
trace variables---called the \emph{context}---which restrict the
evaluation of the temporal modalities to the traces associated with
the variables in $C$.
We show that $\CHLTL$ is not subsumed by $\AHLTL$ or
$\SHLTL$, and single-trace satisfiability of $\CHLTL$ is undecidable.
Moreover, we provide a characterization of the \emph{finite trace properties}
denoted by $\CHLTL$ formulas in terms of the extension $\FOPLUS$ of standard
first-order logic $\FOF$ over finite words with
\emph{addition}. We also establish that the variant of $\FOPLUS$ over infinite words characterizes the 
trace properties expressible in the extension of $\CHLTL$ with past temporal modalities.
Finally, we identify a fragment of $\CHLTL$ which subsumes $\HLTL$ and
for which model checking is decidable.

$\CHLTL$ formulas are linear-time hyper sentences over multi-trace
specifications $\psi$, called \emph{$\CHLTL$ quantifier-free
  formulas}, where the syntax of $\psi$ 
is as follows:
\[
   \psi ::=  \top  \ | \   \Rel{p}{x}  \ | \ \neg \psi \ | \ \psi \wedge \psi \ | \ \Next\psi \ | \   \psi \Until \psi \ | \ \tpl{C} \psi
\]
where $p\in \AP$, $x\in \Var$, and $\tpl{C}$ is the context modality
with $\emptyset \neq C\subseteq \Var$.
A context $C$ is \emph{global for a formula $\varphi$} if $C$ contains
all the trace variables occurring in $\varphi$.  Let $\Pi$ be a
pointed-trace assignment.
Given a context $C$ and an offset $i\geq 0$, we denote by $\Pi +_C i$
the pointed-trace assignment with domain $\Dom(\Pi)$ defined as
follows. For each $x\in \Dom(\Pi)$, where $\Pi(x)=(\pi,h)$:
$[\Pi +_C i](x)=(\pi,h+i)$ if $x\in C$, and $[\Pi +_C i](x)=\Pi(x)$
otherwise.
Intuitively, the positions of the pointed traces associated with the
variables in $C$ advance of the offset $i$, while the positions of the
other pointed traces remain unchanged.
Let $\psi$ be a $\CHLTL$ quantifier-free formula such that $\Dom(\Pi)$
contains the variables occurring in $\psi$.
The satisfaction relation $(\Pi,C)\models \psi$ is defined as follows
(we omit the semantics of the Boolean connectives):
   %
  \[ \begin{array}{ll}
 (\Pi,C) \models \Rel{p}{x} & \Leftrightarrow  \Pi(x)=(\pi,i) \text{ and }p\in \pi(i)\\
  (\Pi,C) \models  \Next\psi & \Leftrightarrow (\Pi +_C 1,C)\models  \psi\\
  (\Pi,C) \models  \psi_1\Until \psi_2 & \Leftrightarrow  \text{for some }i\geq 0:\,  (\Pi +_C i,C)\models  \psi_2   \text{ and }(\Pi +_C k,C) \models  \psi_1 \text{ for all } k<i\\
 (\Pi,C) \models \tpl{C'}\psi&\Leftrightarrow (\Pi,C') \models  \psi
\end{array} \]
We write $\Pi\models \psi$ to mean that $(\Pi,\Var)\models\psi$.\vspace{0.2cm}

\noindent \textbf{$\CHLTL$ versus $\AHLTL$ and $\SHLTL$.}
We show that $\CHLTL$ is able to capture powerful non-regular trace
properties which cannot be expressed in $\AHLTL$ or in $\SHLTL$.
In particular, we consider the following trace property over
$\AP=\{p\}$:\vspace{0.3em}

\noindent\fbox{
\parbox{0.97\textwidth}{
%
  \emph{Suffix Property:} a trace $\pi$ satisfies the suffix property if $\pi$ has a proper suffix
  $\pi^{i}$ for some $i>0$ such that $\pi^{i}=\pi$.
}
}\vspace{0.3em}

\noindent{}This property can be expressed in $\CHLTL$ by the following formula
\vspace{0.1cm}

$
\varphi_{\textit{suff}} \DefinedAs \forall x_1.\,\forall x_2.\, \bigwedge_{p\in \AP}\Always(p[x_1] \leftrightarrow p[x_2])
\wedge \{x_2\}\Eventually\Next \{x_1,x_2\} \bigwedge_{p\in \AP}\Always(p[x_1] \leftrightarrow p[x_2])
$\vspace{0.1cm}

\noindent We show that no $\AHLTL$ and no $\SHLTL$ formula is
equivalent to $\varphi_{\textit{suff}}$.

\begin{theorem}\label{theo:A_S_HLTLinexpressibilitypSuffix}
  $\AHLTL$ and $\SHLTL$ cannot express the suffix property.  Hence,
  $\CHLTL$ is not subsumed by $\AHLTL$ or by $\SHLTL$.
\end{theorem}

\begin{proof}[Proof sketch]
  By Proposition~\ref{prop:SingleTraceSHLTL}, the set of single-trace
  models of a $\SHLTL$ formula is regular. Thus, since the suffix
  trace property is not regular, the result for $\SHLTL$ follows.

  Consider now $\AHLTL$.
  For each $n\geq 1$, let
  $\pi_n \DefinedAs (p^{n}\cdot \emptyset )^{\omega}$ and
  $\pi'_n \DefinedAs p^{n+1}\cdot \emptyset \cdot (p^{n}\cdot
  \emptyset)^{\omega}$.
By construction $\pi_n$ satisfies  the suffix property but $\pi'_n$ not.
Hence, for each $n\geq 1$, the $\CHLTL$ formula $\varphi_{\textit{suff}}$ distinguishes the singleton sets
$\{\pi_n\}$ and $\{\pi'_n\}$.
 On the other hand, we can show the following result, hence, Theorem~\ref{theo:A_S_HLTLinexpressibilitypSuffix}
 directly follows (a proof of the following claim is given in Appendix~\ref{APP:A_S_HLTLinexpressibilitypSuffix}).  \vspace{0.1cm}

\noindent \textbf{Claim.} For each $\AHLTL$ formula $\psi$, there is $n\geq 1$
such that $\{\pi_n\}\models \psi $ \emph{iff} $\{\pi'_n\}\models \psi$.
 \end{proof}

\noindent \textbf{Single-trace satisfiability and characterization of $\CHLTL$ finite-trace properties.}
Like $\AHLTL$, and unlike $\HLTL$ and $\SHLTL$, single-trace satisfiability of $\CHLTL$ turns out to be undecidable. In particular, by a straightforward adaptation of the undecidability proof in~\cite{BozzelliPS21} for model checking $\CHLTL$, one can reduce the recurrence problem in Minsky counter machines~\cite{Harel86} to single-trace satisfiability of $\CHLTL$.

\begin{theorem}
  The single-trace satisfiability problem for $\AHLTL$ is
  $\Sigma_{1}^{1}$-hard.
\end{theorem}

A finite trace (over $\AP$) is a finite non-empty word over $2^{\AP}$.
By adding a fresh proposition $\#\notin \AP$, a finite trace can
be encoded by the trace $enc(w)$ over $\AP\cup\{\#\}$ given by
$w\cdot \{\#\}^{\omega}$.
Given a $\CHLTL$ formula $\varphi$ over $\AP\cup \{\#\}$, the
\emph{finite-trace property denoted by $\varphi$}  is the language
$\Lang_f(\varphi)$ of finite traces $w$ over $\AP$ such that the
single-trace model $\{enc(w)\}$ satisfies $\varphi$.
We provide now a characterization of the finite-trace properties
denoted by $\CHLTL$ formulas over $\AP\cup \{\#\}$ in terms of the
extension $\FOPLUS$ of standard first-order logic $\FOF$ over finite
words on $2^{\AP}$ with
addition. 
Formally, $\FOPLUS$ is a first-order logic with equality over the
signature $\{<,+\}\cup \{P_a\mid a\in\AP \}$, where the atomic
formulas $\psi$ have the following syntax with $x,y$ and $z$ being
first-order variables:
$\psi\DefinedAs x= y\mid x<y \mid z=y+x \mid P_a(x)$. A $\FOPLUS$
sentence (i.e., a $\FOPLUS$ formula with no free variables) is
interpreted over finite traces $w$, where:
\begin{inparaenum}[(i)]
\item variables ranges over
  the set $\{0,\ldots,|w|-1\}$ of positions of $w$,
\item the binary predicate $<$ is the natural ordering on
  $\{0,\ldots,|w|-1\}$, and
\item the predicates $z=y+x$ and $P_a(x)$ are interpreted in the
  obvious way.
\end{inparaenum}
We establish the following result.

\newcounter{theorem-CharacterizationCHLTLFIniteTraces}
\setcounter{theorem-CharacterizationCHLTLFIniteTraces}{\value{theorem}}

\newcounter{sec-CharacterizationCHLTLFIniteTraces}
\setcounter{sec-CharacterizationCHLTLFIniteTraces}{\value{section}}

\begin{theorem}\label{theo:CharacterizationCHLTLFIniteTraces} Given a $\FOPLUS$ sentence $\varphi$ over $\AP$, one can construct in polynomial time a $\CHLTL$ formula $\psi$ over $\AP\cup \{\#\}$ such that $\Lang_f(\psi)$ is the set of models of $\varphi$.
Vice versa, given a $\CHLTL$ formula $\psi$ over $\AP\cup \{\#\}$, one can construct in single exponential time a $\FOPLUS$ sentence $\varphi$ whose set of models is $\Lang_f(\psi)$.
\end{theorem}

Intuitively, when a $\CHLTL$ formula $\psi$ over $\AP\cup\{\#\}$ is
interpreted over singleton models $\{enc(w)\}$ for a given finite
trace $w$, the trace variables in the quantifier-free part of $\psi$
and the temporal modalities evaluated in different contexts can
simulate quantification over positions in $w$ and the atomic formulas
of $\FOPLUS$.
In particular, the addition predicate $z=x+y$ can be simulated by
requiring that two segments of $w$ whose endpoints are referenced by
trace variables have the same length: this is done by shifting with
the eventually modality in a suitable context the left segment of a
non-negative offset, and by checking that the endpoints of the
resulting segments coincide. Note that for trace variables $x$ and $y$
which refer to positions $i$ and $j$ of $w$, one can require that $i$
and $j$ coincide by the $\CHLTL$ formula
$\{x,y\}\Eventually (\neg \#[x]\wedge \neg \#[y]\wedge \Next(
\#[x]\wedge \#[y]))$. A detailed proof of
Theorem~\ref{theo:CharacterizationCHLTLFIniteTraces} is given in
Appendix~\ref{APP:CharacterizationCHLTLFIniteTraces}. Similarly, if we
consider the  extension of $\CHLTL$ with the past counterparts of
the temporal modalities, then the trace properties denoted by past
$\CHLTL$ formulas correspond to the ones denoted by sentences in the
variant $\FOINPLUS$ of $\FOPLUS$ over infinite words on $2^{\AP}$
(traces).  For arbitrary traces, past temporal modalities are crucial for enforcing that two   variables 
refer to the same position (for details see Appendix~\ref{APP:CharacterizationCHLTLInFIniteTraces}).  

\begin{theorem}\label{theo:CharacterizationCHLTLInFIniteTraces} Past $\CHLTL$ and $\FOINPLUS$ capture the same class
  of trace properties.
\end{theorem}

\noindent \textbf{Results about model checking $\CHLTL$.}
Model checking $\CHLTL$ is known to be undecidable even for formulas
where the unique temporal modality occurring in the scope of a
non-global context is $\Eventually$.
For the fragment where the unique temporal modality occurring in a
non-global context is $\Next$, then the problem is decidable.
This fragment has the same expressiveness as $\HLTL$ but it is
exponentially more succinct than $\HLTL$~\cite{BozzelliPS21}.
We provide now new insights on model checking $\CHLTL$.
On the negative side, we show that model checking is undecidable even
for the fragment $\U$ consisting of two-variable quantifier
alternation-free formulas of the form
$\exists x_1.\exists x_2.\, \psi_0\wedge \{x_2\}\Eventually
\{x_1,x_2\}\psi$, where $\psi_0$ and $\psi$ are quantifier-free
$\HLTL$ formulas. A proof of Theorem~\ref{theo:UndecidabilityMC_CHLTL} is in Appendix~\ref{APP:UndecidabilityMC_CHLTL}.

\newcounter{theorem-UndecidabilityMC_CHLTL}
\setcounter{theorem-UndecidabilityMC_CHLTL}{\value{theorem}}

\newcounter{sec-UndecidabilityMC_CHLTL}
\setcounter{sec-UndecidabilityMC_CHLTL}{\value{section}}

\begin{theorem}\label{theo:UndecidabilityMC_CHLTL}
  Model-checking against the fragment $\U$ of $\CHLTL$ is $\Sigma_0^{1}$-hard.
\end{theorem}

By Theorem~\ref{theo:UndecidabilityMC_CHLTL}, $\CHLTL$ model checking
becomes undecidable whenever in a formula a non-singleton context $C$
occurs within a distinct context $C'\neq C$.
Thus, we consider the fragment of $\CHLTL$, called \emph{simple
  $\CHLTL$}, where each context $C$ which occurs in a distinct context
$C'\neq C$ is a singleton. Note that simple $\CHLTL$ subsumes $\HLTL$.

\begin{theorem}\label{theo:DecidabilityMCSimpleCHLTL} The model
  checking problem of simple $\CHLTL$ is decidable.
\end{theorem}

Theorem~\ref{theo:DecidabilityMCSimpleCHLTL} is proved by a polynomial time translation of simple $\CHLTL$ formulas
into equivalent sentences of \emph{first-order logic $\FOE$ with the equal-level predicate $E$}~\cite{Finkbeiner017} whose model checking problem is known to be decidable~\cite{CoenenFHH19}.
This logic is interpreted over sets $\Lang$ of traces, and first-order
variables refer to pointed traces over $\Lang$. In simple $\CHLTL$,
the evaluation of temporal modalities is subdivided in two phases. In
the first phase, modalities are evaluated by a synchronous traversal
of the traces bound to the variables in a non-singleton context. In
the second phase, the temporal modalities are evaluated along a single
trace and singleton contexts allows to switch from a trace to another
one by enforcing a weak form of mutual temporal relation. This
behaviour can be encoded in $\FOE$ (for details, see
Appendix~\ref{APP:DecidabilityMCSimpleCHLTL}).

 \subsection{$\HU$ versus $\AHLTL$, $\SHLTL$, and  $\CHLTL$}\label{sec:HUVersusOthers}

 Both $\SHLTL$ and $\CHLTL$ are 
 subsumed by $\HU$ and
 Hyper $\AAWA$~\cite{BozzelliPS21}.
 In particular, quantifier-free formulas of $\SHLTL$ and $\CHLTL$ can
 be translated in polynomial time into equivalent B\"{u}chi
 \AAWA.
 Corollary~\ref{cor:FromAHLTLToHU} shows that $\AHLTL$ is subsumed by
 $\HU$ as well.
 Thus, since $\AHLTL$ and $\SHLTL$ are expressively incomparable (by
 Corollary~\ref{cor:IncomparabilityAHLTL_SHLTL}), there is an $\HU$
 formula which cannot be expressed in $\AHLTL$ (resp., $\SHLTL$). 
 Therefore, we obtain the following corollary.

 \begin{corollary}\label{cor:ExpressivenessHU}  $\HU$ is
   strictly more expressive than $\AHLTL$ and $\SHLTL$, and subsumes $\CHLTL$.
\end{corollary} 
\section{Asynchronous vs Synchronous Extensions of
  $\HLTL$}\label{sec:OverallExpressiveness}

We compare now the expressiveness of the asynchronous extensions of
$\HLTL$ against $\MSOE$~\cite{CoenenFHH19}.
$\MSOE$ is a monadic second-order logic with equality over the
signature $\{<,E\}\cup \{P_a \mid a\in \AP\}$ which syntactically
extends the monadic second-order logic of one successor $\MSO$ with
the equal-level binary predicate
$E$.
While $\MSO$ is interpreted over traces, $\MSOE$ is interpreted over
sets of traces.
A set $\Lang$ of traces induces the relational structure with domain
$\Lang\times \N$ (i.e., the set of pointed traces associated with
$\Lang$), where
%
\begin{compactitem}
\item the binary predicate $<$ is interpreted as the set of pairs of
  pointed traces in $\Lang\times \N$ of the form
  $((\pi,i_1),(\pi,i_2))$ such that $i_1<i_2$, and
\item the equal-level predicate $E$ is interpreted as the set of pairs
  of pointed traces in $\Lang\times \N$ of the form
  $((\pi_1,i),(\pi_2,i))$.
\end{compactitem}
Hence, $<$ allows to compare distinct timestamps along the same trace,
while the equal-level predicate allows to compare distinct traces at
the same timestamp. For a formal definition of the syntax and
semantics of $\MSOE$, see Appendix
\ref{APP:LogicsWIthEqualLevelPredicate}.

We  show that the considered asynchronous extensions of $\HLTL$
are not subsumed by $\MSOE$.
Intuitively, for some $k\geq 1$, the logic $\AHLTL$ (resp., $\SHLTL$, resp., $\CHLTL$) can
express hyperproperties whose set of models having cardinality $k$
(\emph{$k$-models}) can be encoded by a non-regular set of traces.
On the other hand, we show that the encoding of the $k$-models of a
$\MSOE$ formula always leads to a regular language.

Let $k\geq 1$. We consider the set of atomic propositions given by
$\AP\times [1,k]$ for encoding sets $\Lang$ of traces (over $\AP$)
having cardinality $k$ by traces over $\AP\times [1,k]$.

A trace $\nu$ over $\AP\times [1,k]$ is \emph{well-formed} if for all
$\ell,\ell'\in [1,k]$ with $\ell\neq \ell'$, there is $i\in\N$ and
$p\in\AP$ so that $(p,\ell)\in \nu(i)$ iff $(p,\ell')\notin \nu(i)$.
A well-formed trace $\nu$ over $\AP\times [1,k]$ encodes the set
$\Lang(\nu)$ of the traces $\pi$ (over $\AP$) such that there is
$\ell\in[1,k]$ where $\pi$ corresponds to the projection of $\nu$ over
$\AP\times \{\ell\}$, i.e.~for each $i\geq 0$,
$\pi(i)=\{p\in \AP\mid (p,\ell)\in \nu(i)\}$. Since $\nu$ is
well--formed, $|\Lang(\nu)|= k$ and we say that $\nu$ is a
\emph{$k$-code} of $\Lang(\nu)$. Note that for a set $\Lang$ of traces
(over $\AP$) of cardinality $k$, each ordering of the traces in
$\Lang$ induces a distinct $k$-code.

  Given an hyperproperty specification $\xi$ over $\AP$, a \emph{$k$-model of $\xi$} is a set of traces satisfying
  $\xi$ having cardinality $k$. The \emph{$k$-language of $\xi$} is the set of $k$-codes  associated with the $k$-models
  of $\xi$. The specification $\xi$ is \emph{$k$-regular} if its $k$-language is a regular language over $\AP\times [1,k]$.

  We first show that for each $k\geq 1$, $\MSOE$ sentences are $k$-regular.

\newcounter{lemma-MSOEkregular}
\setcounter{lemma-MSOEkregular}{\value{lemma}}
\newcounter{sec-MSOEkregular}
\setcounter{sec-MSOEkregular}{\value{section}}

\begin{lemma}\label{lemma:MSOEkregular} Let $k\geq 1$ and  $\varphi$ be a $\MSOE$
  sentence over $\AP$. Then, one can construct a $\MSO$ sentence $\varphi'$
  over $\AP\times [1,k]$ whose set of models is the $k$-language of
  $\varphi$.
\end{lemma}

A proof of Lemma~\ref{lemma:MSOEkregular} is in
Appendix~\ref{APP:MSOEkregular}.  Since $\MSO$ sentences capture only
regular languages of traces, by Lemma~\ref{lemma:MSOEkregular}, we
obtain the following result.

\begin{proposition}\label{prop:MSOEkregular} Let $k\geq 1$ and
  $\varphi$ be a  $\MSOE$ sentence over $\AP$. Then, $\varphi$ is
  $k$-regular.
\end{proposition}

We now show that given one of the considered asynchronous extensions
$\Logic$ of $\HLTL$, there is $k\geq 1$ and a $\Logic$ formula
$\varphi$ such that $\varphi$ is \emph{not} $k$-regular.

\begin{proposition}\label{prop:AsynchronousNOtkRegular}
  There is a $\CHLTL$ formula over $\{p\}$ which is not $1$-regular,
  and there are $\AHLTL$ and $\SHLTL$ formulas  over $\{p\}$
    which   are not $2$-regular.
\end{proposition}

\begin{proof}

  Let $\AP=\{p\}$. The $\CHLTL$ formula $\varphi_{\textit{suff}}$
  defined in Section~\ref{sec:ContextHyper} whose models consist of
  the singletons $\{\pi\}$ such that $\pi$ satisfies the suffix
  property is not   $1$-regular.

  Consider now $\AHLTL$ and $\SHLTL$. For all $k,n\geq 1$, let
  $\pi_{k,n}\DefinedAs\emptyset^{k} \cdot (\{p\}\cdot \emptyset
  )^{n}\cdot \emptyset^{\omega}$ and
  $\Lang_{k,n}\DefinedAs\{\pi_{1,n},\pi_{k,n}\}$. We denote by
  $\Lang_2$ the set of traces over $\AP\times [1,2]$ which are
  $2$-codes of the sets $\Lang_{k,n}$ for $k>1$ and $n\geq 1$.
  Clearly, $\Lang_2$ is not regular. Let $\theta(x)$ be a one-variable
  quantifier-free $\HLTL$ formula   capturing the traces
  $\pi_{k,n}$ for $k,n\geq 1$.  We define a $\SHLTL$ formula $\psi_S$
  and an $\AHLTL$ formula $\psi_A$ whose $2$-language is
  $\Lang_2$: \vspace{0.1cm}

$\psi_S\DefinedAs\exists x_1.\,\forall x_2.\, \Next p[x_1]\wedge \theta(x_1)\wedge \theta(x_2)\wedge \Always_{\{p\}}(p[x_1] \leftrightarrow p[x_2]);
 $

 $\psi_A\DefinedAs \exists x_1.\,\forall x_2.\,\forall x_3.\,\ET.\, \Next p[x_1]\wedge   \theta(x_2) \wedge \theta(x_3) \wedge \Always(p[x_2] \leftrightarrow p[x_3])$.
\end{proof}

\noindent{}By Propositions~\ref{prop:MSOEkregular}
and~\ref{prop:AsynchronousNOtkRegular}, and since $\AHLTL$, $\SHLTL$,
and $\CHLTL$ are  subsumed by $\HU$
(Corollary~\ref{cor:ExpressivenessHU}), we obtain the following
result.

\begin{corollary}\label{cor:InExpressivenessMSOE} $\AHLTL$, $\SHLTL$,
  $\CHLTL$, and $\HU$ are not subsumed by $\MSOE$.
\end{corollary}

%
%

\section{Conclusions}\label{Sec:Concl}

Two interesting questions are left open.
The first concerns the expressiveness of $\CHLTL$ versus $\AHLTL$ and
$\SHLTL$. We have shown that $\CHLTL$ is not subsumed by $\AHLTL$ or
$\SHLTL$. We conjecture that the converse holds too. The motivation is
that (unlike $\CHLTL$) $\AHLTL$ and $\SHLTL$ implicitly allow a
restricted form of monadic second-order quantification.
In particular, we conjecture that the hyperproperty characterizing
the sets consisting of stuttering-equivalent traces, which can be
easily expressed both in $\AHLTL$ and $\SHLTL$, cannot be captured by
$\CHLTL$.

The second question is whether $\MSOE$ is subsumed or not by $\HU$.
It is known that contrary to $\MSOE$ and $\FOE$, $\HLTL$ cannot
express requirements which relate at some point an unbounded number of
traces~\cite{BozzelliMP15}.
The main reason is that---differently from $\MSOE$ and
$\FOE$---quantifiers in $\HLTL$ only refer to the initial positions of
the traces.
Since in $\HU$ the semantics of quantifiers is the same as $\HLTL$, we
conjecture that the inexpressiveness result for $\HLTL$
in~\cite{BozzelliMP15} can be extended to $\HU$ as well.
This would imply together with the results of
Corollary~\ref{cor:InExpressivenessMSOE} that $\MSOE$ and $\HU$ are
expressively incomparable and that so are $\FOE$ and $\HU$.




\bibliographystyle{plainurl}
\bibliography{biblio}

\newpage

\makeatletter
\edef\thetheorem{\expandafter\noexpand\thesection\@thmcountersep\@thmcounter{theorem}}
\makeatother

\begin{LARGE}
  \noindent\textbf{Appendix}
\end{LARGE}

\newcounter{aux}
\newcounter{auxSec}

\section{Syntax and semantics of $\AAWA$}\label{APP:AAWA}

We recall the syntax and semantics of parity $\AAWA$~\cite{GutsfeldOO21}. We need additional definitions.

For a  set $X$, $\B^{+}(X)$   denotes the set
of \emph{positive} Boolean formulas over $X$, i.e. Boolean formulas  built from elements in $X$
using $\vee$ and $\wedge$ (we also allow the formulas $\true$ and
$\false$). A subset $Y$ of $X$  \emph{satisfies}
$\theta\in\B^{+}(X)$ iff the truth assignment that assigns $\true$
to the elements in $Y$ and $\false$ to the elements of $X\setminus
Y$ satisfies $\theta$.\vspace{0.2cm}

\noindent \textbf{Labeled trees.} A tree $T$ is a   prefix closed subset of $\N^{*}$.  Elements of $T$ are called nodes and $\varepsilon$ (the empty word) is the root of $T$. For $x\in T$, the set of children of $x$ in $T$ is the set of
nodes of the form $x\cdot n$ for some $n\in \N$.  A  path of $T$ is a maximal sequence $\pi$ of nodes such that $\pi(0)=\varepsilon$ and $\pi(i)$ is a child in $T$ of $\pi(i-1)$ for all $0<i<|\pi|$.   For an alphabet $\Sigma$, a $\Sigma$-labeled   tree is a pair $\tpl{T, \Lab}$ consisting of a  tree
 and a labelling $\Lab:T \mapsto \Sigma$  assigning to  each node in $T$ a symbol in $\Sigma$. \vspace{0.2cm}

\noindent \textbf{Syntax and Semantics of parity $\AAWA$.} Here, we consider a equivalent and slight variant of the automata in~\cite{GutsfeldOO21}.  Let $n\geq 1$. A parity $\AAWA$ with $n$ tapes (parity $\NAAWA$ for short)  over a finite alphabet $\Sigma$ is a tuple  $\Au=\tpl{\Sigma, Q,Q_0,\rho,\Omega}$, where
 $Q$ is a finite set of (control) states, $Q_0\subseteq Q$ is a set of initial states,
  $\rho:Q\times \Sigma^{n}\rightarrow \B^{+}(Q\times [1,n])$ is the transition function, and $\Omega:Q \mapsto \N$ is a parity acceptance condition assigning to each state a natural number (\emph{color}).
  If $\rho(q,\sigma)$ only consists  of disjunctions for all $q\in Q$ and $\sigma\in \Sigma^{n}$, we call an
  $\AAWA$ a Nondeterministic Asynchronous Word Automaton (\NAWA\ for short).

Let $n\geq 1$ and $\Au=\tpl{\Sigma, Q,Q_0,\rho,\Omega}$ be a parity $\AAWA$ with $n$ tapes (parity $\NAAWA$).  A run of $\Au$ over an $n$-tuple $\overrightarrow{w}=(w_1,\ldots,w_n)$ of infinite words over $\Sigma$ is a $(Q\times \N^{n})$-labeled  tree
$r=\tpl{T_r,\Lab_r}$, where each node  of $T_r$ labelled by $(q,\wp)$ with $\wp=(i_1,\ldots,i_n)$ describes a copy of the automaton that is in  state $q$ and reads the $(i_h+1)^{th}$ symbol of the input word $w_h$ for each $h\in [1,n]$. Moreover, we require that $r(\varepsilon)=(q_0,(0,\ldots,0))$ (initially, the automaton is in state $q_0$ reading the first position of each input word), and for each  $\tau\in T_r$ with $\Lab_r(\tau)=(q,\wp)$ and $\wp=(i_1,\ldots,i_n)$, there
 is a (possibly empty)  set
  $\{(q_1,d_1),\ldots,(q_k,d_k)\}\subseteq Q\times [1,n]$ for some $k\geq 0$  satisfying $\delta(q,(w_1(i_1),\ldots,w_n(i_n)))$ such that $\tau$ has $k$ children $\tau_1,\ldots,\tau_k$ and  $\Lab_r(\tau_j) = (q_j,(i_1,\ldots, i_{d_j}+1,\ldots,i_n))$ for all $1\leq j\leq k$.
   The run $r$ is accepting if for each infinite path  $\nu$ of $r$, the smallest
color of the states in $Q$ that occur infinitely often along $\nu$ is even.
    An $n$-tuple  $\overrightarrow{w}$ of infinite words over $\Sigma$ is accepted by
    $\Au$ if there exists
  an accepting run of  $\Au$ over $\overrightarrow{w}$.

We also consider B\"{u}chi $\NAAWA$ and generalized B\"{u}chi $\NAAWA$. In a B\"{u}chi $\NAAWA$, the acceptance condition is given by a set $F$ of accepting states (encodable as a parity condition with two colors): in this case, a run is accepting if each infinite
path  visits infinitely often some accepting state in $F$. In a generalized B\"{u}chi $\NAAWA$, the acceptance condition is given by a family $\F$ of sets of accepting states (B\"{u}chi components): a run is accepting if every infinite path visits a state in every set $F\in\F$ infinitely
often. By standard arguments (see e.g.~\cite{Demri2016}), a generalized B\"{u}chi $\NAAWA$ can be translated into an equivalent B\"{u}chi
$\NAAWA$ with quadratic blowup by exploiting a counter modulo the number of B\"{u}chi components.

\newpage

\section{Proofs from Section~\ref{sec:AsynchronousHLTL}}\label{APP:AsynchronousHLTL}

\subsection{Proof of Theorem~\ref{theo:AHLTLinexpressibilitypSynchronicity}}\label{APP:AHLTLinexpressibilitypSynchronicity}

In this section, we provide a proof of Theorem~\ref{theo:AHLTLinexpressibilitypSynchronicity} by showing that
$\AHLTL$ cannot express the $p$-synchronicity requirement. At this purpose, we need additional definitions.

 Let $\AP=\{p\}$ and for each $n\geq 1$, let $\pi_n$ and $\pi'_n$ be the traces defined as:
 \[
 \pi_n \DefinedAs p^{n}\cdot \emptyset^{\omega} \text{ and } \pi'_n \DefinedAs  p^{n+1}\cdot \emptyset^{\omega}
 \]
Moreover, let $\Lang_n\DefinedAs\{\pi_n\}$ and $\Lang'_n\DefinedAs\{\pi_n,\pi'_n\}$.
By construction $\Lang_n$ satisfies  the $p$-synchronicity requirement
while $\Lang'_n$  not. Thus, in order to show that $\AHLTL$ cannot express $p$-synchronicity, it suffices to prove the following result.

\begin{proposition}\label{prop:AHLTLinexpressibilitypSynchronicity} For each $\AHLTL$ formula $\psi$, there is $n\geq 1$
such that $\Lang_n\models \psi $ if and only if $\Lang'_n\models \psi $.
\end{proposition}

 In order to prove Proposition~\ref{prop:AHLTLinexpressibilitypSynchronicity}, we
first establish  the following preliminary result.

\begin{lemma}\label{lemma:pSynchronicityAHLTL} Let $n\geq 1$ and $\psi$ be an \HLTL\  quantifier-free formula over $\{p\}$ with trace variables $x_1,\ldots,x_k$, and whose nesting depth of next modality is at most $n-1$. Then, for all initial pointed-trace assignments
$\Pi$ over $\Lang_n$ and $\Pi'$ over $\Lang'_n$ such that $\Dom(\Pi)=\Dom(\Pi')=\{x_1,\ldots,x_k\}$, it holds that
$\Pi\models \ET\psi$ \text{ iff } $\Pi'\models \ET\psi$.
\end{lemma}
\begin{proof}
By the definition of the traces $\pi_n$ and $\pi'_n$, each stuttering expansion of $\pi'_n$ is also a stuttering expansion
of $\pi_n$. Recall that $\Lang_n=\{\pi_n\}$ and $\Lang'_n=\{\pi_n,\pi'_n\}$  Hence, by Proposition~\ref{prop:CharacterizationAHLTLSemantics}, we have that
 $\Pi'\models \ET\psi$ entails that $\Pi\models \ET\psi$.

Now assume that $\Pi\models \ET\psi$. By
Proposition~\ref{prop:CharacterizationAHLTLSemantics}, for all $i\in [1,k]$, there is a stuttering
expansion $\rho_i$ of $\pi_n=p^{n}\cdot \emptyset^{\omega}$ such that
\[
 \{x_1 \leftarrow (\rho_1,0),\ldots,x_k \leftarrow (\rho_k,0)\}\models \psi
\]
For all $i\in [1,k]$, let $\rho'_i$ be the trace given by $p\cdot \rho_i$. By definition of $\pi_n$ and $\pi'_n$,
it holds that $\rho'_i$ is a stuttering expansion of $\pi'_n$ for all $i\in [1,k]$.
We prove that
\[
 \{x_1 \leftarrow (\rho'_1,0),\ldots,x_k \leftarrow (\rho'_k,0)\}\models \psi
\]
Hence, by Proposition~\ref{prop:CharacterizationAHLTLSemantics}, we obtain that $\Pi'\models \ET\psi$ and the result follows. We first deduce the following preliminary result which can be proved by a straightforward induction on
$n-i-1$.\vspace{0.1cm}

\noindent \textbf{Claim.} Let $i\in [0,n-1]$ and $\psi'$ be an $\HLTL$ quantifier-free formula over $\{p\}$ with  variables in $\{x_1,\ldots,x_k\}$ and whose nesting depth of next modality is at most
$n-i-1$. Then:
\[
 (\{x_1 \leftarrow (\rho'_1,0),\ldots,x_k \leftarrow (\rho'_k,0)\},i)\models \psi' \text{ iff }
  (\{x_1 \leftarrow (\rho'_1,0),\ldots,x_k \leftarrow (\rho'_k,0)\},i+1)\models \psi'
\]
 By construction, we have that for each $i\in [1,k]$, $(\rho'_i)^{1}= \rho_i$. Thus, by the previous claim, we obtain that for each \HLTL\  quantifier-free formula over $\{p\}$ with trace variables $x_1,\ldots,x_k$, and whose nesting depth of next modality is at most $n-1$,
\[
 \{x_1 \leftarrow (\rho_1,0),\ldots,x_k \leftarrow (\rho_k,0)\}\models \psi \text{ iff }
  \{x_1 \leftarrow (\rho'_1,0),\ldots,x_k \leftarrow (\rho'_k,0)\}\models \psi
\]
and we are done.
\end{proof}

\noindent \textbf{Proof of Proposition~\ref{prop:AHLTLinexpressibilitypSynchronicity}.}  Let $\psi= Q_1 x_1.\ldots Q_k x_k.\psi' $ be a $\AHLTL$ formula where  $\psi'$ is quantifier free and $Q_i\in \{\exists,\forall\}$
  for all $i\in [1,k]$.
  If $\psi'$ is of the form $\ET\psi''$ for some $\HLTL$ quantifier-free formula $\psi''$, then the result
  directly follows from Lemma~\ref{lemma:pSynchronicityAHLTL} by taking $n>|\psi''|$.

  Now, assume that $\psi'$ is of the form $\AT\psi''$ for some $\HLTL$ quantifier-free formula $\psi''$.
  We observe that by the semantics of $\AHLTL$, for each trace set $\Lang$,
   $\Lang\not \models   Q_1 x_1.\ldots Q_k x_k.\AT\psi'' $ if and only if
    $\Lang  \models \overline{Q}_1 x_1.\ldots \overline{Q}_k x_k.\ET\neg \psi'' $, where for each $i\in [1,k]$,
    $\overline{Q}_i= \exists$ if $Q_i=\forall$, and $\overline{Q}_i= \forall$ otherwise. Hence, by Lemma~\ref{lemma:pSynchronicityAHLTL}, by taking  $n>|\psi''|$,   the result holds in this case as well.\qed

\subsection{Construction of the generalized B\"{u}chi $\NAAWA$ in the  proof of Theorem~\ref{theo:FromAHLTLtoAAWA}}\label{APP:FromExistentialAHLTL}

In this section, we provide a detailed proof of the following result.

\begin{proposition}\label{prop:FromExistentialAHLTLtoNAWA}Given an  existential $\AHLTL$ quantifier-free formula $\ET\psi$ with trace variables $x_1,\ldots,x_n$,
one can build in singly exponential time a generalized B\"{u}chi  $\NNAWA$ $\Au_{\ET\psi}$ over $2^{\AP}$   accepting the set of $n$-tuples $(\pi_1,\ldots,\pi_n)$ of traces such  that $(\{x_1\mapsto (\pi_1,0),\ldots,x_1\mapsto (\pi_n,0)\})\models \ET\psi$.
\end{proposition}
\begin{proof}
Without loss
of generality, we assume that $\AP$ is the set of atomic propositions occurring in the
$\HLTL$ quantifier-free formula $\psi$. The \emph{closure $\cl(\psi)$} of $\psi$ is the set
consisting (i) of the relativized propositions $p[x_i]$ for each $p\in \AP $ and $i\in [1,n]$ and of their negations,
and (2) of the subformulas $\theta$ of $\psi$  and of their negations $\neg \theta$ (we identify $\neg\neg\theta$ with $\theta$).
An \emph{atom $A$ of $\psi$} is a maximal propositionally consistent subset of $\cl(\psi)$, i.e.~such that
the following holds:
\begin{compactitem}
  \item for each $\theta\in\cl(\psi)$, $\theta\in A$ iff $\neg\theta \notin A$,
  \item for each $\theta_1\wedge\theta_2\in \cl(\psi)$, $\theta_1\wedge\theta_2\in  A$ iff $\{\theta_1,\theta_2\}\subseteq A$.
\end{compactitem}
We denote by $\Atoms(\psi)$ the set of atoms of $\psi$. Note that $|\Atoms(\psi)|= 2^{O(n\cdot|\psi|)}$.
For two $\psi$-atoms $A$ and $A'$, we say that $A'$ is a successor of $A$ if:
\begin{compactitem}
  \item for each $\Next\theta\in\cl(\psi)$, $\Next\theta\in A$ iff $\theta\in A'$,
  \item for each $\theta_1\Until\theta_2\in \cl(\psi)$, $\theta_1\Until\theta_2\in  A$ iff \emph{either} $\theta_2\in A$ \emph{or}
  $\theta_1\in A$ and $\theta_1\Until\theta_2\in A'$.
\end{compactitem}
Intuitively, if an atom $A$ describes the set of subformulas of $\psi$ that hold at the current phase, then the
set of subformulas associated to the next phase must be a successor of $A$.

An atom $A$ is consistent with an $n$-tuple $(\sigma_1,\ldots,\sigma_n)\in (2^{\AP})^{n}$ of input symbols if for each
$p\in \AP$ and $i\in [1,n]$, $p[x_i]\in A$ iff $p\in \sigma_i$.

The generalized  B\"{u}chi $\NNAWA$
 $\Au_{\ET\psi} =\tpl{2^{\AP},Q,Q_0,\rho,\F}$ is defined as follows. The set $Q$ of states is given by
 $(\Atoms(\psi)\times [1,n])    \cup (\Atoms(\psi)\times [1,n] \times \{\Beg\})$. States in $ \Atoms(\psi)\times [1,n] \times \{\Beg\} $ are associated with the beginning of a phase, and the second component $d\in [1,n]$ in a state represents the previous chosen direction.
 The set $Q_0$ of initial states consists of the states of the form $(A,1,\Beg)$ such that $\psi\in A$, and the transition function
 $\rho: Q\times (2^{\AP})^{n}\rightarrow 2^{Q\times [1,n]}$ is defined as follows for each $(\sigma_1,\ldots,\sigma_n)\in (2^{\AP})^{n}$:
 \begin{compactitem}
  \item for each $(A,i,\Beg)\in \Atoms(\psi)\times[1,n]\times \{\Beg\}$, $\rho((A,i,\Beg),(\sigma_1,\ldots,\sigma_n))$ is empty if
  $A$ is not consistent with $(\sigma_1,\ldots,\sigma_n)$. Otherwise,
  $\rho((A,i,\Beg),(\sigma_1,\ldots,\sigma_n))$
   consists  (i) of the target moves of the
  form $((A,d),d)$ for  $d\in [1,n]$, and  (ii) of the target moves $((A',d,\Beg),d)$ such that
  $A'$ is a successor of $A$ and $d\in [1,n]$.
 \item for each $(A,i)\in \Atoms(\psi)\times[1,n]$, $\rho((A,i),(\sigma_1,\ldots,\sigma_n))$ consists (i) of the target moves of the form
 $((A,d),d)$ such that $d>i$, and (ii) of the target moves $((A',d,\Beg),d)$ such that $d>i$ and
  $A'$ is a successor of $A$.
\end{compactitem}
Finally, the family $\F$ of acceptance components is defined as follows:
\begin{compactitem}
  \item for each until formula $\theta_1\Until\theta_2$, $\F$ has a component consisting of the states $(A,d)$ and $(A,d,\Beg)$   such that either
  $\theta_2\in A$ or $\theta_1\Until\theta_2 \notin A$ (\emph{fulfillment of the eventuality modalities}),
  \item for each direction $d\in [1,n]$, $\F$ has a component consisting of states of the form $(A,d)$ or $(A,d,\Beg)$
  (\emph{fulfillment of the fairness condition for trajectories}).
\end{compactitem}\vspace{0.2cm}

Correctness of the construction directly follows from the following claim. \vspace{0.2cm}

\noindent \textbf{Claim.} Let  $(\pi_1,\ldots,\pi_n)$ be an $n$-tuple of traces. Then there is an accepting run
of $\Au_{\ET\psi}$ over $(\pi_1,\ldots,\pi_n)$ \emph{if and only if} $(\{x_1\mapsto (\pi_1,0),\ldots,x_1\mapsto (\pi_n,0)\})\models \ET\psi$.\vspace{0.2cm}

\noindent \textbf{Proof of the claim.} An $n$-pointer is an $n$-tuple $(i_1,\ldots,i_n)$ of natural numbers.
Given two $n$-pointers $\wp=(i_1,\ldots,i_n)$ and $\wp'=(i'_1,\ldots,i'_n)$, we say that
$\wp'$ is a successor of $\wp$, if (i) $i'_d\in\{i_d, i_d+1\}$ for each $d\in [1,n]$, and (ii)
$i'_d = i_d+1$ for some $d\in [1,n]$. An infinite sequence $\nu = \wp_0,\wp_1,\ldots$
of $n$-pointers is \emph{well-formed} if $\wp_0 = (0,\ldots,0)$  and for each $i\geq 0$,
 $\wp_{i+1}$ is a successor of $\wp_i$. Moreover, $\nu$ is \emph{fair} if for each $d\in [1,n]$ and $m>0$, there is
 $\ell\geq 0$ such that the $d^{\text{th}}$-component of $\wp_\ell$ is greater than $m$.
 Evidently, there is a bijection between (fair) well-formed infinite sequences $\nu= \wp_0,\wp_1,\ldots$ of $n$-pointers and the (fair) trajectories over the set $\{x_1,\ldots,x_n\}$ of trace variables. In particular, the trajectory $t(\nu)$ associated with
 $\nu$ has as $k^{\text{th}}$ element the non-empty set of trace variables $x_d$ such that the $d^{\text{th}}$-component of $\wp_{k+1}$ is greater than the  $d^{\text{th}}$-component of $\wp_k$.

Let $\AP_{[x_1,\ldots,x_n]}$ be the finite set consisting of the relativized propositions $p[x_i]$ for each $p\in \AP$
and $i\in [1,n]$. We interpret $\AP_{[x_1,\ldots,x_n]}$ as a set of atomic propositions and $\psi$ as an
$\LTL$ formula over $\AP_{[x_1,\ldots,x_n]}$.  Then, given an
$n$-pointer $\wp=(j_1,\ldots,j_n)$, we denote by $(\pi_1,\ldots,\pi_n)[\wp]$ the subset of $\AP_{[x_1,\ldots,x_n]}$
given by
\[
\displaystyle{\bigcup_{i\in [1,n]}\bigcup_{p\in \pi_i(j_i)}} \{p[x_i]\}
\]
By the semantics of $\AHLTL$, it easily follows that $(\{x_1\mapsto (\pi_1,0),\ldots,x_1\mapsto (\pi_n,0)\})\models \ET\psi$
if and only if there is a fair well-formed infinite sequence $\nu=\wp_0,\wp_1,\ldots$ of $n$-pointers such that the
infinite word $\xi$ over $2^{\AP_{[x_1,\ldots,x_n]}}$ given by
\[
 (\pi_1,\ldots,\pi_n)[\wp_0],(\pi_1,\ldots,\pi_n)[\wp_1],\ldots
\]
is a model of the $\LTL$ formula $\psi$ over $\AP_{[x_1,\ldots,x_n]}$. By the standard automata-theoretic construction for $\LTL$,
it follows that $(\{x_1\mapsto (\pi_1,0),\ldots,x_1\mapsto (\pi_n,0)\})\models \ET\psi$ if and only if
there is a fair well-formed infinite sequence $\nu=\wp_0,\wp_1,\ldots$ of $n$-pointers and
an infinite sequence of $\psi$-atoms $A_0,A_1,\ldots$ such that
 \begin{compactitem}
  \item  $\psi\in A_0$ and $A_{i+1}$ is a successor of $A_i$ for each $i\geq 0$,
  \item for each $i\geq 0$ with $\wp_i=(j_1,\ldots,j_n)$, $A_i$ is consistent with $(\pi_1(j_1),\ldots,\pi_n(j_n))$,
  \item and
  for each until formula $\theta_1\Until\theta_2$, there are infinitely many $i$ such that either
   $\theta_2\in A_i$ or $\theta_1\Until\theta_2 \notin A_i$.
\end{compactitem}
Thus by construction of the automaton $\Au_{\ET\psi}$, the result easily follows. This concludes the proof of the claim and
Proposition~\ref{prop:FromExistentialAHLTLtoNAWA} as well.
\end{proof}

\subsection{Construction of the $\AHLTL$ formula $\varphi_M$ in the proof of Theorem~\ref{theo:UndecidabilitySIngleTraceAHLTL}}\label{APP:UndecidabilitySIngleTraceAHLTL}

In this section, we provide details about the construction of the
$\AHLTL$ formula $\varphi_M$ in the proof of Theorem~\ref{theo:UndecidabilitySIngleTraceAHLTL}.
Fix a counter machine $M = \tpl{Q,\Delta,\delta_\init,\delta_\rec}$ and let  $\AP\DefinedAs \Delta \cup \{c_1,c_2,\#,\Beg,\Pad\}$. Recall that  an $M$-configuration $(\delta,\nu)$ is encoded by the finite words over $2^{\AP}$ (called \emph{segments}) of the form:
\[
\{\Beg,\delta\} P_1\ldots P_m \{\Pad\}^{k}
\]
where $k\geq 1$, $m= \max(\nu(1),\nu(2))$, and for all $i\in [1,m]$, (i) $\emptyset\neq  P_i\subseteq \{\#,c_1,c_2\}$,
(ii) $\#\in P_i$ iff $i$ is odd, and (iii) for all $\ell\in\{1,2\}$, $c_\ell\in P_i$ iff $i\leq \nu(\ell)$. A computation $\rho$ of $M$ is then encoded by the traces obtained by concatenating the codes of the configurations along $\rho$ starting from the first one. In order to define the $\AHLTL$ formula $\varphi_M$, we exploit the following auxiliary formulas for each trace variable $x$:
\[
  \begin{array}{l}
    \theta_{\Beg}(x) \DefinedAs \Beg[x] \wedge \displaystyle{\bigvee_{\delta\in \Delta}(\delta[x]\wedge \bigwedge_{p\in\AP\setminus\{\Beg,\delta\}}\neg p[x])}\\
    \theta_{c}(x) \DefinedAs (c_1[x]\vee c_2[x]) \wedge \displaystyle{
    \bigwedge_{p\in\AP\setminus\{c_1,c_2,\#\}}\neg p[x]}\\
        \theta_{\Pad}(x) \DefinedAs \Pad[x]  \wedge \displaystyle{
    \bigwedge_{p\in\AP\setminus\{\Pad\}}\neg p[x]}\\
\end{array}
\]
Intuitively, $\theta_{\Beg}(x)$ characterizes the first position of a segment, and $\theta_{c}(x)$ and $\theta_{\Pad}(x)$ characterize the counter and the pad positions
of a segment, respectively.
Then, the $\AHLTL$ formula $\varphi_M$ is given by $\varphi_M\DefinedAs \exists x_1\exists x_2.\,\ET\, \psi$, where the quantifier-free $\HLTL$ formula $\psi$ guarantees that for the two stuttering expansions $\pi_1$ and $\pi_2$ of the given trace $\pi$, the following holds:
\begin{itemize}
  \item both $\pi_1$ and $\pi_2$ are infinite concatenations of segments. This is ensured by the following conjunct for each $i=1,2$:
\[
  \begin{array}{l}
    \theta_{\Beg}(x_i) \,\wedge\, \Always\Eventually\, \theta_{\Beg}(x_i) \,\wedge\, \Always \Bigl( \theta_{\Beg}(x_i) \rightarrow \Next[(\theta_{c}(x_i)\wedge \#[x_i])\vee \theta_{\Pad}(x_i)] \,\wedge\, \vspace{0.2cm} \\
    \theta_{\Pad}(x_i) \rightarrow \Next [\theta_{\Pad}(x_i) \vee \theta_{\Beg}(x_i)] \,\wedge \,
    (\theta_{c}(x_i) \wedge \#[x_i]) \rightarrow \Next[(\theta_{c}(x_i)\wedge \neg\#[x_i])\vee \theta_{\Pad}(x_i)] \,\wedge \, \vspace{0.2cm}\\
     \,[\theta_{c}(x_i)\wedge \neg\#[x_i]]  \rightarrow  \Next[(\theta_{c}(x_i)\wedge \#[x_i])\vee \theta_{\Pad}(x_i)] \,\wedge
   \, \displaystyle{\bigwedge_{\ell =1,2}} [\theta_{c}(x_i)\wedge \neg c_\ell[x_i]] \rightarrow \neg\Next\, c_\ell \Bigr)
\end{array}
\]
  \item The first segment of $\pi_1$ encodes the initial configuration $(\delta_\init,\nu_0)$ of $M$ and the second segment of $\pi_1$
  encodes a configuration which is a successor of  $(\delta_\init,\nu_0)$ in $M$. Without loss of generality we assume that
   $\delta_\init$ corresponds to an increment instruction for the first counter. Then, the previous requirement can be expressed
   as follows:
   \[
      \delta_\init[x_1] \,\wedge \Next\Bigl(\Pad[x_1]\,\wedge\,\bigl(\neg \theta_{\Beg}(x_1)\,\Until\,[\theta_{\Beg}(x_1)\,\wedge\,
      \Next (c_1[x_1]\wedge \neg c_2[x_1]\wedge \Next \Pad[x_1]) ]\,\bigr)\Bigr)
   \]
  \item $\delta_\rec$ occurs infinitely often along $\pi_1$:
  \[
    \Always\Eventually \delta_\rec
   \]
  \item For each $i\geq 2$, the $(i+1)^{th}$ segment $s_2$ of $\pi_2$ starts at the same position as the
  $i^{th}$ segment $s_1$ of $\pi_1$. Moreover, $s_1$ and $s_2$ have the same length.
   \[
    \Next\Bigl( (\neg\theta_{\Beg}(x_1)\wedge \neg\theta_{\Beg}(x_2))\,\Until\, \bigl( \neg\theta_{\Beg}(x_1)\wedge \theta_{\Beg}(x_2)
     \wedge\Next\Always (\theta_{\Beg}(x_1) \leftrightarrow \theta_{\Beg}(x_2)) \,\bigr)\Bigr)
   \]
   \item For each segment $s_2$ of $\pi_2$ distinct from the first one whose start position coincides with the start position
   of a segment $s_1$ of $\pi_1$, we have that the configuration encoded by $s_2$ is a successor in $M$ of the configuration encoded by $s_1$.
\[
 \displaystyle{\bigwedge_{\delta=(q,(\instr,\ell),q')\in\Delta}}
 \Next\Always \Bigl( \bigl(\theta_{\Beg}(x_1)\wedge \theta_{\Beg}(x_2) \wedge \delta[x_1]\bigr) \,\longrightarrow\, \Bigl(\displaystyle{\bigvee_{\delta'=(q',(\instr',\ell'),q'')\in\Delta}}\delta'[x_2] \,\wedge\,
\]
\[
   \hspace{2.5cm} \underbrace{\Next\bigl((c_{3-\ell}[x_1] \wedge c_{3-\ell}[x_2])  \Until (\neg c_{3-\ell}[x_1] \wedge \neg c_{3-\ell}[x_2])}_{\text{the value of counter $3-\ell$ does not change}}\bigr)\,\wedge\,
\]
\[
     \hspace{0.5cm} \underbrace{\instr =\inc \rightarrow \Next\bigl((c_{\ell}[x_1] \wedge c_{\ell}[x_2])  \Until (\neg  c_{\ell}[x_1] \wedge c_{\ell}[x_2]\wedge \neg \Next c_{\ell}[x_2])}_{\text{increment of counter $\ell$}}\bigr)\,\wedge\,
\]
\[
     \hspace{0.5cm} \underbrace{\instr =\dec \rightarrow \Next\bigl((c_{\ell}[x_1] \wedge c_{\ell}[x_2])  \Until (  c_{\ell}[x_1] \wedge \neg c_{\ell}[x_2]\wedge \neg \Next c_{\ell}[x_1])}_{\text{decrement of counter $\ell$}}\bigr)\,\wedge\,
\]
\[
     \hspace{2.5cm} \underbrace{\instr =\zero \rightarrow \Next (\neg c_{\ell}[x_1] \wedge\neg  c_{\ell}[x_2]) }_{\text{zero test of counter $\ell$}}\,\,\Bigr)\Bigr)
\]
\end{itemize}
Now, we crucially observe that since $\pi_1$ and $\pi_2$ are stuttering expansions of the same trace $\pi$, the alternation requirement for proposition $\#$ in the encoding of an $M$-configurations ensures that $\pi_1$ and $\pi_2$ encode the same infinite sequence of $M$-configurations. Hence, $\varphi_M$ has a single-trace model \emph{if and only if} $M$ has an infinite computation starting from the initial configuration and visiting the designated transition  $\delta_\rec$ infinitely often.

\section{Proofs from Section~\ref{sec:StutteringHLTL}}\label{APP:StutteringHLTL}

\subsection{Proof of Proposition~\ref{prop:InexpressivityCountingsHLTL}}\label{APP:InexpressivityCountingsHLTL}

In this section, we provide a proof of Proposition~\ref{prop:InexpressivityCountingsHLTL}.

\setcounter{aux}{\value{proposition}}
\setcounter{auxSec}{\value{section}}
\setcounter{section}{\value{sec-InexpressivityCountingsHLTL}}
\setcounter{proposition}{\value{prop-InexpressivityCountingsHLTL}}

\begin{proposition}For each $\SHLTL$ formula $\psi$, there is $n\geq 1$ s.t.
$\Lang_n\models \psi$ iff $\Lang'_n\models \psi$.
\end{proposition}

 \setcounter{proposition}{\value{aux}}
 \setcounter{section}{\value{auxSec}}

Recall that  for each $n\geq 1$, $\Lang_n\DefinedAs\{\pi_n,\rho_n\}$  and   $\Lang'_n\DefinedAs\{\pi_n,\rho'_n\}$, where
$\pi_n$, $\rho_n$, and $\rho'_n$ are the traces over $\AP=\{p\}$ defined as follows:
\[
 \pi_n \DefinedAs ( \emptyset\cdot p )^{n}\cdot \emptyset^{\omega} \quad\quad
 \rho_n \DefinedAs  ( \emptyset\cdot p )^{2n}\cdot \emptyset^{\omega} \quad\quad
  \rho'_n \DefinedAs  ( \emptyset\cdot p )^{2n+1}\cdot \emptyset^{\omega}
\]
For a pointed trace $(\pi,i)$ and a finite set of $\LTL$ formulas $\Gamma$, we write
$\SUCC_\Gamma(\pi,i)$ to mean the pointed trace $(\pi,\ell)$ where $\ell$ is the $\Gamma$-successor
of position $i$ in $\pi$.
In order to prove Proposition~\ref{prop:InexpressivityCountingsHLTL}, we need two preliminary technical results.

\begin{lemma}\label{lemma:stutterTraceInexpressivity} The following holds for all $n\geq 1$ and $k\geq 1$;
 \begin{enumerate}
   \item Let us consider the trace $\pi=(\emptyset \cdot p)^{k+n}\cdot\emptyset^{\omega}$. Then for all $\LTL$ formulas $\theta$ such that
   $|\theta|<n$ and $i\in [0,k-1]$, it holds that (i) $(\pi,2i)\models \theta$ iff $(\pi,2i+2)\models \theta$, and (ii)
   $(\pi,2i+1)\models \theta$ iff $(\pi,2i+3)\models \theta$.
   \item  Let $\Gamma$ be a finite set of $\LTL$ formulas $\theta$ such that $|\theta|<n$. Then, there is a trace $\rho$
   coinciding with the $\Gamma$-stutter trace of some  suffix  of $(\emptyset \cdot p)^{n-1}\cdot\emptyset^{\omega}$ such that
    \begin{itemize}
   \item \emph{either}   $\stfr_{\Gamma}(\rho_n)= (\emptyset \cdot p)^{n+1}\cdot \rho$, $\stfr_{\Gamma}(\rho'_n)= (\emptyset \cdot p)^{n+2}\cdot \rho$, and $\stfr_{\Gamma}(\pi_n)= (\emptyset \cdot p)\cdot \rho$,
   \item \emph{or} $\stfr_{\Gamma}(\rho_n)= \stfr_{\Gamma}(\rho'_n)= \stfr_{\Gamma}(\pi_n) = \emptyset\cdot \rho$.
 \end{itemize}
 \item Let us consider the trace $\pi=(\emptyset \cdot p)^{k+n}\cdot\emptyset^{\omega}$. Then for all \emph{stuttering} $\LTL$ formulas $\theta_S$ such that
   $|\theta_S|<n$ and $i\in [0,k-1]$, it holds that (i) $(\pi,2i)\models \theta_S$ iff $(\pi,2i+2)\models \theta_S$, and (2)
   $(\pi,2i+1)\models \theta_S$ iff $(\pi,2i+3)\models \theta_S$.
 \end{enumerate}
\end{lemma}
\begin{proof}
Property~(1) of Lemma~\ref{lemma:stutterTraceInexpressivity} can be proved by a straightforward double induction on
$k$ and the structure of the given $\LTL$ formula $\theta$. Now, let us consider Property~(2). Let $\Gamma$ be a finite set of $\LTL$ formulas $\theta$ such that $|\theta|<n$. Recall that
$\pi_n = ( \emptyset\cdot p )^{n}\cdot \emptyset^{\omega}$,  $\rho_n = ( \emptyset\cdot p )^{2n}\cdot \emptyset^{\omega}$,
and $\rho'_n = ( \emptyset\cdot p )^{2n+1}\cdot \emptyset^{\omega}$. Thus, the result easily follows from
  Property~(1)  and the definition of the function $\stfr_\Gamma$.

  Finally, let us consider Property~(3). The proof is by a  double induction on
$n$ and the structure of the given stuttering $\LTL$ formula $\theta_S$ where $|\theta_S|<n$. We focus on the case where $k=1$ and consider the positions $1$ and $3$ of $\pi=(\emptyset \cdot p)^{1+n}\cdot\emptyset^{\omega}$ (the proof for the other cases being similar). We proceed on the structure of the formula $\theta_S$. The base case (where $\theta_S$ is an atomic proposition) and the case of Boolean connectives directly follow  from the induction hypothesis.  It remains to consider the cases where
 either $\theta_S=\Next_\Gamma\theta_1$ or $\theta_S= \theta_1\Until_\Gamma\theta_2$ for some formulas $\theta_1$ and $\theta_2$, and finite set $\Gamma$ of $\LTL$ formulas. Since $|\theta_S|<n$, it holds that $n-1\geq 1$, $|\theta_1|<n-1$,  $|\theta_2|<n-1$, and $\Gamma$ consists of $\LTL$ formulas with size smaller than $n-1$. Note that $\pi$ can be written in the form
  $\pi=(\emptyset \cdot p)^{k+n-1}\cdot\emptyset^{\omega}$ where $k=2$. Thus,
  by the induction hypothesis on the length of the formula, we have that for all formulas $\theta\in \Gamma\cup \{\theta_1,\theta_2\}$, it holds that (i) $(\pi,1)\models \theta$ iff $(\pi,3)\models \theta$, and (ii) $(\pi,0)\models \theta$ iff $(\pi,2)\models \theta$ iff $(\pi,4)\models \theta$. Hence, the result easily follows.
\end{proof}

\begin{lemma}\label{lemma:InexpressivityCountingsHLTL} Let $n\geq 1$ and $\psi$ be a $\SHLTL$ quantifier-free formulas over $\{p\}$ with two trace variables $x_1$ and $x_2$ such that $|\psi|<n$. Then, for all $i,j\in [0,2n-1]$ such that $j\leq i$:
\[
\{x_1\leftarrow (\pi_n,i),x_2\leftarrow (\rho_n,j)\}\models \psi \text{ iff } \{x_1\leftarrow (\pi_n,i),x_2\leftarrow (\rho'_n,j)\}\models \psi
\]
\end{lemma}
\begin{proof}
Recall that
$\pi_n = ( \emptyset\cdot p )^{n}\cdot \emptyset^{\omega}$,  $\rho_n = ( \emptyset\cdot p )^{2n}\cdot \emptyset^{\omega}$,
and $\rho'_n = ( \emptyset\cdot p )^{2n+1}\cdot \emptyset^{\omega}$.

Fix $i,j\in [0,2n-1]$ such that $j\leq i$. The proof is by a double induction on $2n-1-i$ and the structure of $\psi$.
The case where $\psi=p[x_i]$ for some $i\in \{1,2\}$ directly follows from the fact that $\rho_n(j)=\rho'_n(j)$ for all
$j\in [0,2n-1]$. The cases of Boolean connectives where either $\psi=\neg\psi_1$ or $\psi=\psi_1\wedge\psi_2$ for some formulas $\psi_1$ and $\psi_2$ directly follow  from the induction hypothesis. It remains to consider the cases where
 either $\psi=\Next_\Gamma\psi_1$ or $\psi= \psi_1\Until_\Gamma\psi_2$ for some formulas $\psi_1$ and $\psi_2$, and finite set $\Gamma$ of $\LTL$ formulas. Since $|\psi|<n$, all the \LTL\ formulas $\theta$ in $\Gamma$ satisfy $|\theta|<n$. By Property~(2) of Lemma~\ref{lemma:stutterTraceInexpressivity}, there is a suffix $\rho_S$ of $(\emptyset \cdot p)^{n-1}\cdot\emptyset^{\omega}$ such that one of the following holds:
    \begin{compactitem}
   \item \emph{either}   $\stfr_{\Gamma}(\rho_n)= (\emptyset \cdot p)^{n+1}\cdot \stfr_{\Gamma}(\rho_S)$, $\stfr_{\Gamma}(\rho'_n)= (\emptyset \cdot p)^{n+2}\cdot \stfr_{\Gamma}(\rho_S)$, and $\stfr_{\Gamma}(\pi_n)= (\emptyset \cdot p)\cdot \stfr_{\Gamma}(\rho_S)$,
   \item \emph{or} $\stfr_{\Gamma}(\rho_n)= \stfr_{\Gamma}(\rho'_n)= \stfr_{\Gamma}(\pi_n) = \emptyset\cdot \stfr_{\Gamma}(\rho_S)$.
   Note that in this case, being $\rho_n = ( \emptyset\cdot p )^{2n}\cdot \emptyset^{\omega}$
and $\rho'_n = ( \emptyset\cdot p )\cdot \rho_n$, it holds that for all $k\geq 1$,  $\SUCC^{k}_\Gamma(\rho'_n,j)= \SUCC^{k}_\Gamma(\rho_n,j)+2$.
 \end{compactitem}

 \bigskip

\noindent We first consider the case of next modality where $\psi = \Next_\Gamma \psi_1$. We
 distinguish two cases:
     \begin{itemize}
   \item $\stfr_{\Gamma}(\rho_n)= \stfr_{\Gamma}(\rho'_n)= \stfr_{\Gamma}(\pi_n) = \emptyset\cdot \stfr_{\Gamma}(\rho_S)$.
  Hence, $\SUCC_\Gamma(\rho'_n,j)= \SUCC_\Gamma(\rho_n,j)+2$, i.e.~ there is
    $k>j$ such that  $\SUCC_\Gamma(\rho_n,j)= k$ and $\SUCC_\Gamma(\rho'_n,j)= k+2$.
   Thus, since $(\rho_n)^{k}=(\rho'_n)^{k+2}$, the result trivially follows.
    \item Now assume that  $\stfr_{\Gamma}(\rho_n)= (\emptyset \cdot p)^{n+1}\cdot \stfr_{\Gamma}(\rho_S)$, $\stfr_{\Gamma}(\rho'_n)= (\emptyset \cdot p)^{n+2}\cdot \stfr_{\Gamma}(\rho_S)$, and $\stfr_{\Gamma}(\pi_n)= (\emptyset \cdot p)\cdot \stfr_{\Gamma}(\rho_S)$.  Hence, $\SUCC_\Gamma(\rho_n,j)= j+1$ and $\SUCC_\Gamma(\rho'_n,j)= j+1$. We distinguish  two cases:
     \begin{compactitem}
   \item $i=2n-1$. Since $\pi_n= (\emptyset \cdot p)^{n}\cdot\emptyset^{\omega}$, we have that
   $\SUCC_\Gamma(\pi_n,2n-1)= 2n$ and $\pi_n(k)=\emptyset$ for all $k\geq 2n$. Let
   $\theta_S$ be the stuttering $\LTL$ formula obtained from $\psi_1$ by replacing each occurrence of $p[x_2]$ with $p$
   and each occurrence of $p[x_1]$ with $\neg\top$. Evidently, (i) $\{x_1\leftarrow (\pi_n,2n),x_2\leftarrow (\rho_n,j+1)\}\models \psi_1$ \text{ iff } $(\rho_n,j+1)\models \theta_S$, and (ii) $\{x_1\leftarrow (\pi_n,2n),x_2\leftarrow (\rho'_n,j+1)\}\models \psi_1$ \text{ iff } $(\rho'_n,j+1)\models \theta_S$.
   Since $\rho'_n$ can be written in form $\rho'_n = ( \emptyset\cdot p )^{(n+1)+n}\cdot \emptyset^{\omega}$,
    $|\theta_S|<n$ and $j+1\in [0,2n]$, by applying Property~(3) of Lemma~\ref{lemma:stutterTraceInexpressivity} with $k=n+1$, it holds that $(\rho'_n,j+1)\models \theta_S$ iff $(\rho'_n,j+3)\models \theta_S$. Hence,
    being  $(\rho_n)^{j+1}=(\rho'_n)^{j+3}$, we obtain that
   $(\rho_n,j+1)\models \theta_S$ iff $(\rho'_n,j+1)\models \theta_S$. Thus, since $\SUCC_\Gamma(\rho_n,j)= j+1$ and $\SUCC_\Gamma(\rho'_n,j)= j+1$, the result follows.
    \item $i<2n-1$. If $\SUCC_\Gamma(\pi_n,i)\geq 2n$, we proceed as in the previous case. Otherwise, $\SUCC_\Gamma(\pi_n,i)\leq 2n-1$ and $j+1\leq \SUCC_\Gamma(\pi_n,i)$ (recall that $j\leq i$). Thus, since  $\SUCC_\Gamma(\rho_n,j)= j+1$,   $\SUCC_\Gamma(\rho'_n,j)= j+1$, and $2n-1-\SUCC_\Gamma(\pi_n,i)<2n-1-i$, the result directly follows from the induction hypothesis.
 \end{compactitem}
 \end{itemize}

 \bigskip

  \noindent We now consider the case  where $\psi = \psi_1\Until_\Gamma \psi_2$. We
 distinguish two cases:
     \begin{compactitem}
   \item $\stfr_{\Gamma}(\rho_n)= \stfr_{\Gamma}(\rho'_n)= \stfr_{\Gamma}(\pi_n) = \emptyset \cdot \stfr_{\Gamma}(\rho_S)$.
   It follows that for all $k\geq 1$,  $\SUCC^{k}_\Gamma(\rho'_n,j)= \SUCC^{k}_\Gamma(\rho_n,j)+2$. Moreover, by the induction hypothesis, we can assume that the result holds for the formulas $\psi_1$ and $\psi_2$.
   Thus, since $(\rho_n)^{\ell}=(\rho'_n)^{\ell+2}$ for all $\ell\geq 0$, the result trivially follows.
    \item Now assume that  $\stfr_{\Gamma}(\rho_n)= (\emptyset \cdot p)^{n+1}\cdot \stfr_{\Gamma}(\rho_S)$, $\stfr_{\Gamma}(\rho'_n)= (\emptyset \cdot p)^{n+2}\cdot \stfr_{\Gamma}(\rho_S)$, and $\stfr_{\Gamma}(\pi_n)= (\emptyset \cdot p)\cdot \stfr_{\Gamma}(\rho_S)$.  Then, since $\pi_n(i')=\emptyset$ for all $i'\geq 2n$, by the induction hypothesis and by reasoning as for case of the next modalities, we deduce that:
     \begin{compactitem}
   \item for all  $i',j'\in [0,2n-1]$ with $j'\leq i'$ and $h\in\{1,2\}$,
\[
\{x_1\leftarrow (\pi_n,i'),x_2\leftarrow (\rho_n,j')\}\models \psi_h \text{ iff } \{x_1\leftarrow (\pi_n,i'),x_2\leftarrow (\rho'_n,j')\}\models \psi_h
\]
\item for all $i'\geq 2n$, $j'\in \{2n,2n+2\}$  and    $h\in\{1,2\}$,
\[
\{x_1\leftarrow (\pi_n,i'),x_2\leftarrow (\rho_n,2n)\}\models \psi_h \text{ iff } \{x_1\leftarrow (\pi_n,i'),x_2\leftarrow (\rho'_n,j')\}\models \psi_h
\]
\item for all $i'\geq 2n$, $j'\in \{2n+1,2n+3\}$  and    $h\in\{1,2\}$,
\[
\{x_1\leftarrow (\pi_n,i'),x_2\leftarrow (\rho_n,2n+1)\}\models \psi_h \text{ iff } \{x_1\leftarrow (\pi_n,i'),x_2\leftarrow (\rho'_n,j')\}\models \psi_h
\]
  \end{compactitem}
  Hence, the result easily follows.  This concludes the proof of Lemma~\ref{lemma:InexpressivityCountingsHLTL}.
 \end{compactitem}
\end{proof}

\noindent \textbf{Proof of Proposition~\ref{prop:InexpressivityCountingsHLTL}.} Let $n\geq 1$. Recall that $\Lang_n=\{\pi_n,\rho_n\}$ and $\Lang'_n=\{\pi_n,\rho'_n\}$. Given two initial pointed-trace assignments
$\Pi$ over $\Lang_n$ and $\Pi'$ over $\Lang'_n$, we say that $\Pi$ and $\Pi'$ are \emph{similar} if $\Dom(\Pi)=\Dom(\Pi')$ and for all
$x\in \Dom(\Pi)$, either (i) $\Pi(x)=\Pi'(x)=(\pi_n,0)$, or (ii) $\Pi(x)=(\rho_n,0)$ and $\Pi'(x)=(\rho'_n,0)$. Thus, if $\Pi$ and $\Pi'$ are similar, by  Lemma~\ref{lemma:InexpressivityCountingsHLTL}, for each $\SHLTL$ quantifier-free formula $\psi$ such that $|\psi|<n$ it holds that $\Pi\models \psi$ iff $\Pi'\models \psi$. It follows that for each $\SHLTL$ formula $\varphi$ such that
$|\varphi|<n$, $\Lang_n\models \varphi$ iff $\Lang'_n\models \varphi$. Hence, Proposition~\ref{prop:InexpressivityCountingsHLTL} directly follows.

\section{Proofs from Section~\ref{sec:ContextHyper}}\label{APP:ContextHyper}

\subsection{Proof of the claim in the proof of Theorem~\ref{theo:A_S_HLTLinexpressibilitypSuffix}}\label{APP:A_S_HLTLinexpressibilitypSuffix}

Let $\AP=\{p\}$ and for each $n\geq 1$, let $\pi_n$ and $\pi'_n$ be the traces defined as:
 \[
 \pi_n \DefinedAs (p^{n}\cdot \emptyset )^{\omega} \text{ and } \pi'_n \DefinedAs  p^{n+1}\cdot \emptyset \cdot (p^{n}\cdot \emptyset)^{\omega}
 \]
The claim in the proof of
Theorem~\ref{theo:A_S_HLTLinexpressibilitypSuffix} corresponds to the following proposition.

\begin{proposition}\label{prop:A_S_HLTLinexpressibilitySuffix} For each $\AHLTL$  formula $\psi$, there is $n\geq 1$ such that $\{\pi_n\}\models \psi $ if and only if $\{\pi'_n\}\models \psi $.
\end{proposition}

In order to prove Proposition~\ref{prop:A_S_HLTLinexpressibilitySuffix}, we
need the following preliminary result.

 \begin{lemma}\label{lemma:SuffixPropertyAHLTL} Let $n\geq 1$ and $\psi$ be an \HLTL\  quantifier-free formula over $\{p\}$ with trace variables $x_1,\ldots,x_k$, and whose nesting depth of next modality is at most $n-1$. Then,
\[\{x_1 \leftarrow (\pi_n,0),\ldots,x_k \leftarrow (\pi_n,0)\}\models \ET\psi  \text{ iff }
\{x_1 \leftarrow (\pi'_n,0),\ldots,x_k \leftarrow (\pi'_n,0)\}\models \ET\psi
\]
\end{lemma}
\begin{proof}
By the definition of the traces $\pi_n$ and $\pi'_n$, each stuttering expansion of $\pi'_n$ is also a stuttering expansion
of $\pi_n$. Hence, by Proposition~\ref{prop:CharacterizationAHLTLSemantics}, we have that
 $\{x_1 \leftarrow (\pi'_n,0),\ldots,x_k \leftarrow (\pi'_n,0)\}\models \ET\psi$ entails that
$\{x_1 \leftarrow (\pi_n,0),\ldots,x_k \leftarrow (\pi_n,0)\}\models \ET\psi$.

Now assume that $\{x_1 \leftarrow (\pi_n,0),\ldots,x_k \leftarrow (\pi_n,0)\}\models \ET\psi$. By
Proposition~\ref{prop:CharacterizationAHLTLSemantics}, for all $i\in [1,k]$, there is a stuttering
expansion $\rho_i$ of $\pi_n$ such that
\[
 \{x_1 \leftarrow (\rho_1,0),\ldots,x_k \leftarrow (\rho_k,0)\}\models \psi
\]
For all $i\in [1,k]$, let $\rho'_i$ be the trace given by $p\cdot \rho_i$. By definition of $\pi_n$ and $\pi'_n$,
it holds that $\rho'_i$ is a stuttering expansion of $\pi'_n$ for all $i\in [1,k]$.
We prove that
\[
 \{x_1 \leftarrow (\rho'_1,0),\ldots,x_k \leftarrow (\rho'_k,0)\}\models \psi
\]
Hence, by Proposition~\ref{prop:CharacterizationAHLTLSemantics} the result follows. We first deduce the following preliminary result which can be proved by a straightforward induction on $n-i-1\in [0,n-1]$.\vspace{0.1cm}

\noindent \textbf{Claim.} Let $i\in [0,n-1]$ and $\psi'$ be an $\HLTL$ quantifier-free formula 
with variables in $\{x_1,\ldots,x_k\}$ and whose nesting depth of next modality is at most
$n-i-1$. Then:
\[
 (\{x_1 \leftarrow (\rho'_1,0),\ldots,x_k \leftarrow (\rho'_k,0)\},i)\models \psi' \text{ iff }
  (\{x_1 \leftarrow (\rho'_1,0),\ldots,x_k \leftarrow (\rho'_k,0)\},i+1)\models \psi'
\]
 By construction, we have that for each $i\in [1,k]$, $(\rho'_i)^{1}= \rho_i$. Thus, by the previous claim, we obtain that for each \HLTL\  quantifier-free formula 
 with trace variables $x_1,\ldots,x_k$, and whose nesting depth of next modality is at most $n-1$,
\[
 \{x_1 \leftarrow (\rho_1,0),\ldots,x_k \leftarrow (\rho_k,0)\}\models \psi \text{ iff }
  \{x_1 \leftarrow (\rho'_1,0),\ldots,x_k \leftarrow (\rho'_k,0)\}\models \psi
\]
and we are done.
\end{proof}

\noindent \textbf{Proof of Proposition~\ref{prop:A_S_HLTLinexpressibilitySuffix}.}
Let $\psi$ be an $\AHLTL$ formula with trace variables $x_1,\ldots,x_k$ and $\psi'$  be the quantifier-free part of $\psi$.
  First, assume that $\psi'$ is of the form $\ET\psi''$ for some $\HLTL$ quantifier-free formula $\psi''$.
  Note that for each trace variable $\pi$,
   $\{\pi\}\models \psi$ iff $\{x_1 \leftarrow (\pi,0),\ldots,x_k \leftarrow (\pi,0)\}\models \ET\psi''$. Thus, by applying
   Lemma~\ref{lemma:SuffixPropertyAHLTL}, it follows that for each $n>|\psi|$, $\{\pi_n\}\models \psi $ iff $\{\pi'_n\}\models \psi $, and the result holds.

  Now, assume that $\psi'$ is of the form $\AT\psi''$ for some $\HLTL$ quantifier-free formula $\psi''$.
  Evidently, for each trace variable $\pi$,    $\{\pi\}\not\models \psi$ iff $\{x_1 \leftarrow (\pi,0),\ldots,x_k \leftarrow (\pi,0)\}\models \ET\neg\psi''$. Hence, by Lemma~\ref{lemma:SuffixPropertyAHLTL}, for each  $n>|\psi|$, $\{\pi_n\}\models \psi $ iff $\{\pi'_n\}\models \psi $, and the result holds in this case as well.

\subsection{Proof of Theorem~\ref{theo:CharacterizationCHLTLFIniteTraces}}\label{APP:CharacterizationCHLTLFIniteTraces}

In this section, we provide a proof of Theorem~\ref{theo:CharacterizationCHLTLFIniteTraces}.

\setcounter{aux}{\value{theorem}}
\setcounter{theorem}{\value{theorem-CharacterizationCHLTLFIniteTraces}}
\setcounter{auxSec}{\value{section}}
\setcounter{section}{\value{sec-CharacterizationCHLTLFIniteTraces}}

\begin{theorem} Given a $\FOPLUS$ sentence $\varphi$ over $\AP$, one can construct in polynomial time a $\CHLTL$ formula $\psi$ over $\AP\cup \{\#\}$ such that $\Lang_f(\psi)$ is the set of models of $\varphi$.
Vice versa, given a $\CHLTL$ formula $\psi$ over $\AP\cup \{\#\}$, one can construct in single exponential time a $\FOPLUS$ sentence $\varphi$ whose set of models is $\Lang_f(\psi)$.
\end{theorem}
\setcounter{theorem}{\value{aux}}
\setcounter{section}{\value{auxSec}}

In order to prove Theorem~\ref{theo:CharacterizationCHLTLFIniteTraces}, we first give a formal definition of the syntax and semantics of $\FOPLUS$. Formulas $\varphi$ of $\FOPLUS$  over $\AP$ and $\Var$ are defined by the following syntax:
\[
\varphi \DefinedAs \top  \ | \  P_a(x)  \ | \  x=y  \ | \ x<y  \ | \ z=x+y  \ | \ \neg \varphi \ | \ \varphi \wedge \varphi \ | \ \exists x. \varphi
\]
where $a\in \AP$ and $x,y,z\in \Var$. Given a finite trace $w$, a $w$-valuation is a partial mapping $g$ over $\Var$ assigning to each variable in its domain $\Dom(g)$ a position $0\leq i<|w|$. For a $\FOPLUS$ formula $\varphi$ and a $w$-valuation $g$ whose domain contains the free variables (i.e. the variables $x$ which do not occur in the scope of a quantifier $\exists x$) of $\varphi$, the satisfaction relation $(w,g)\models \varphi$ is defined as follows (we omit the semantics of the Boolean connectives which is standard):
\[ \begin{array}{ll}
(w,g) \models   P_a(x)  &  \Leftrightarrow  a\in w(g(x))\\
(w,g) \models   x=y    &  \Leftrightarrow  g(x)=g(y)\\
(w,g) \models   x<y    &  \Leftrightarrow  g(x)<g(y)\\
(w,g) \models   z=x+y    &  \Leftrightarrow  g(z)=g(x)+g(y)\\
(w,g) \models  \exists x. \varphi  &  \Leftrightarrow   (w,g[x \mapsto i])\models \varphi \text{ for some }0\leq i<|w|
\end{array} \]
where $g[x \mapsto i](x)=i$ and $g[x \mapsto i](y)=g(y)$ for each $y\in\Dom(g)\setminus\{x\}$. Note that if $\varphi$ is a sentence (i.e., $\varphi$ has no free variables), the satisfaction relation $(w,g)\models \varphi$ is independent
of $g$. We say that $w$ is a model of a sentence $\varphi$, denoted  $w\models \varphi$, if $(w,g_\emptyset)\models \varphi$, where $g_\emptyset$ is the empty $w$-valuation.

Recall that for a $\CHLTL$ formula $\psi$ over $\AP\cup \{\#\}$, where $\#\notin \AP$, $\Lang_f(\psi)$ is the set of
traces (over $\AP\cup \{\#\}$) of the form $enc(w)=w\cdot \{\#\}^{\omega}$ for some finite trace $w$ (over $\AP$) such that
$\{enc(w)\}$ satisfies $\psi$.
For a finite trace $w$, a \emph{$w$-pointed-trace assignment} is a pointed-trace assignment $\Pi$ associating to each variable
$x\in \Dom(\Pi)$, a pointed trace of the form $(enc(w),i)$. For a quantifier-free $\CHLTL$ formula $\psi$ over $\AP\cup \{\#\}$ and a finite trace $w$, we write $enc(w)\models \psi$ to mean that the initial $w$-pointed-trace assignment, having as domain the set of variables occurring in $\psi$ and assigning to each variable the pointed trace $(enc(w),0)$, satisfies $\psi$. Now, we prove Theorem~\ref{theo:CharacterizationCHLTLFIniteTraces}.
The first part of Theorem~\ref{theo:CharacterizationCHLTLFIniteTraces} directly follows from the following proposition.

\begin{proposition}\label{prop:FromFOPLUSToCHLTL} Let $\varphi$ be a $\FOPLUS$ sentence over $\AP$. Then, one can construct in polynomial time a $\CHLTL$ quantifier-free formula $f(\varphi)$ over $\AP\cup\{\#\}$ such that for each finite trace $w$, it holds that
$w\models \varphi$ \emph{if and only if} $enc(w)\models f(\varphi)$.
\end{proposition}
\begin{proof}
Fix a $\FOPLUS$ sentence $\varphi$ over $\AP$. Without loss of generality, we assume that each quantifier in $\varphi$ introduces
a different variable. For each atomic subformula $\theta$ of $\varphi$ and for each variable $x$ occurring in $\theta$,
we introduce a fresh copy $x_{\theta,S}$ of $x$ we call \emph{synchronization variable}. Moreover, if $\theta$ is of the form
$z=x+y$, we introduce two fresh copies $x_{\theta,I}$ and $y_{\theta,I}$ of $x$ and $y$, respectively, we call
\emph{initialization variables}. For a variable $x$ occurring in $\varphi$, we denote by $C_x$ the set (context) consisting of the variable $x$ plus the synchronization copies of $x$. Note that $C_x$ does not include the initialization copies of $x$.

Given two variables $x$ and $y$, we exploit in the translation  a $\CHLTL$ quantifier-free formula $\psi_=(x,y)$ over $\AP\cup\{\#\}$ and $\{x,y\}$,
which intuitively asserts that along the encoding $enc(w)$ of a finite trace $w$,  $x$ and $y$ refer to the same position of $w$ (recall that $\#\notin \AP$ and $enc(w)=w\cdot \{\#\}^{\omega}$):
     \[
      \psi_=(x,y) \DefinedAs \{x,y\}\Eventually(\neg \#[x]\wedge \neg \#[y]
       \wedge \Next (\#[x]\wedge  \#[y]))
     \]

Let $\theta$ be a subformula of $\varphi$. We associate to $\theta$ a $\CHLTL$ quantifier-free formula $f(\theta)$ over $\AP\cup\{\#\}$. The mapping $f$
is inductively defined as follows:
\begin{itemize}
  \item $ f(P_a(x)) \DefinedAs a[x]$ for each $a\in\AP$.
  \item $\theta$ is of the form $x=y$: $f(\theta) \DefinedAs \psi_=(x_{\theta,S},y_{\theta,S})$.
  \item $\theta$ is of the form $x<y$: $f(\theta) \DefinedAs \{x_{\theta,S},y_{\theta,S}\}\Eventually(\neg \#[x_{\theta,S}]\wedge  \#[y_{\theta,S}])$.
  \item $\theta$ is an atomic formula of the form $z=x+y$:
     \[
     \begin{array}{ll}
       f(\theta) \DefinedAs & \{x_{\theta,I},x_{\theta,S}\}\Eventually(\psi_=\bigl(x_{\theta,I},y_{\theta,S})\wedge \psi_=(x_{\theta,S},z_{\theta,S})\bigr)\,\vee\,\vspace{0.1cm}\\
      & \{y_{\theta,I},y_{\theta,S}\}\Eventually(\psi_=\bigl(y_{\theta,I},x_{\theta,S})\wedge \psi_=(y_{\theta,S},z_{\theta,S})\bigr)
       \end{array}
     \]
     The first disjunct handles the case where $x\leq y$, while the second disjunct
     considers the case where $x>y$. For the first case (the second case being similar), we exploit the synchronization
     variables $x_{\theta,S}$, $y_{\theta,S}$, and $z_{\theta,S}$, and the initialization variable $x_{\theta,I}$. Intuitively, in $\CHLTL$, the variables $x_{\theta,I}$ and $x_{\theta,S}$  are moved forwardly and synchronously of a non-negative offset by the eventually modality. Moreover, assuming that $x_{\theta,I}$ points to the first position of the word, we check
     (by means of the auxiliary formula $\psi_=$) that   this offset corresponds to $y$ and at the same time $z=x+y$ by checking that after the shift the variables   $x_{\theta,I}$ and $y_{\theta,S}$ (resp., $x_{\theta,S}$ and $z_{\theta,S}$) refer to the same position.
  \item $f(\neg\theta)\DefinedAs \neg f(\theta)$.
  \item $f(\theta_1\wedge\theta_2)\DefinedAs f(\theta_1)\wedge f(\theta_2)$.
  \item $f(\exists x.\theta)\DefinedAs  C_x \Eventually(\neg\#[x]\wedge f(\theta))$.
\end{itemize}\vspace{0.1cm}

We now prove that the construction is correct. Let $\Var$ be the set of variables occurring in $\varphi$ and $\widehat{\Var}$
be the set of variables consisting of the variables in $\Var$ together with their synchronization and initialization copies.

Let $w$ be a finite trace, $\theta$ be a subformula of $\varphi$, $V\subseteq\Var$ the set of variables occurring
in $\theta$, and $g$ a $w$-valuation  with domain $V$. We denote by $\Pi_{w,g}$ the $w$-pointed-trace assignment with domain $\widehat{\Var}$ defined as follows for all $x\in \Var$ and for all synchronized copies $x_{\theta,S}$ and
initialization copies $x_{\theta,I}$ of $x$:
\begin{itemize}
  \item $\Pi_{w,g}(x_{\theta,I})\DefinedAs(enc(w),0)$,
  \item if $x\in V$, then $\Pi_{w,g}(x)\DefinedAs (enc(w),g(x))$ and $\Pi_{w,g}(x_{\theta,S})\DefinedAs(enc(w),g(x))$,
  \item if $x\notin V$, then $\Pi_{w,g}(x)\DefinedAs (enc(w),0)$ and $\Pi_{w,g}(x_{\theta,S})\DefinedAs(enc(w),0)$.
\end{itemize}

By  a straightforward induction on the structure of $\theta$, the following holds.\vspace{0.1cm}

\noindent \textbf{Claim.} $(w,g)\models \theta$ if and only if $\Pi_{w,g}\models f(\theta)$.  \vspace{0.1cm}

Since $\varphi$ is a sentence, by the previous claim, we obtain that $(w,g_\emptyset)\models \theta$ iff $\Pi_{w,g_\emptyset}\models f(\varphi)$, where $g_\emptyset$ is the empty $w$-valuation. Note that $\Pi_{w,g_\emptyset}$
associates to each variable in $\widehat{\Var}$ the pointed-trace $(enc(w),0)$. Hence
$w\models \varphi$ iff $enc(w)\models f(\varphi)$, and we are done.
\end{proof}

Finally, the second part of Theorem~\ref{theo:CharacterizationCHLTLFIniteTraces} directly follows from the following proposition.

\begin{proposition}\label{FromSingleTRaceCHLTLToFOPLUS} Let $\varphi$ be a $\CHLTL$ quantifier-free formula over $\AP\cup\{\#\}$.  Then, one can construct in exponential time a  $\FOPLUS$ sentence $f(\varphi)$ over $\AP$ such that for each finite trace $w$, it holds that $w\models f(\varphi)$ \emph{if and only if} $enc(w)\models  \varphi$.
\end{proposition}
\begin{proof}
Fix a $\CHLTL$ quantifier-free formula $\varphi$ and let $\Var=\{x_1,\ldots,x_n\}$ be the set of variables occurring in $\varphi$. Given a finite trace $w$, we define an equivalence relation $\sim_w$ between $w$-pointed-trace assignments $\Pi_1$ and $\Pi_2$ having domain $\Var$:
$\Pi_1\sim_w \Pi_2$ if for each $x_i\in\Var$ with $\Pi_1(x_i)=(enc(w),h_1)$ and $\Pi_2(x_i)=(enc(w),h_2)$, it holds that (i) $h_1< |w|$ iff $h_2<|w|$, and (ii) $h_1=h_2$ iff $h_1<|w|$. Thus, $\Pi_1$ and $\Pi_2$ are equivalent if  they coincide on all the variables which refer to positions of the finite trace $w$. In the translation from $\CHLTL$ into $\FOPLUS$, we will exploit the following result which can be proved by a straightforward induction on the structure of $\varphi$.\vspace{0.2cm}

\noindent \textbf{Claim 1.} Let $w$ be a finite trace and $\Pi_1$ and $\Pi_2$ be two $w$-pointed-trace assignments with domain $\Var$ such that $\Pi_1 \sim_w \Pi_2$. Then, $\Pi_1\models \varphi$ iff $\Pi_2\models \varphi$.\vspace{0.2cm}

For each variable $x_i\in \Var$ and $\ell\in \N$, let $y_i^{\ell}$ be a fresh copy of $x_i$. For a subformula $\theta$ of $\varphi$,
a context $\emptyset\neq C\subseteq \Var$, $\ell\in\N$, and a subset $Ex\subseteq\Var$, we associate to the tuple $(\theta,C,\ell,Ex)$ a  $\FOPLUS$ formula $f(\theta,C,\ell,Ex)$ over $\AP$ with free variables $y_1^{\ell},\ldots,y_n^{\ell}$. Intuitively, $C$ represents the current context and for the given finite trace $w$, $Ex$ is the set of variables in $\Var$ whose copies refer to positions in $enc(w)=w\cdot \{\#\}^{\omega}$ which go beyond $w$. The $\FOPLUS$ formula $f(\theta,C,\ell,Ex)$ is inductively defined as follows:
\begin{itemize}
  \item $\theta$ is a relativized atomic proposition $p[x_i]$:
  \[
   f(p[x_i],C,\ell,Ex) \DefinedAs
          \left\{\begin{array}{ll}
                  P_p[y_i^{\ell}]             & \text{ if }p\neq \# \text{ and }x_i\notin Ex \\
                  \top             & \text{ if }p= \# \text{ and }x_i\in Ex \\
                 \neg\top            &  \text{ otherwise}
                \end{array}\right.
  \]
  \item $f(\neg\theta,C,\ell,Ex)\DefinedAs \neg f(\theta,C,\ell,Ex)$.
  \item $f(\theta_1\wedge\theta_2,C,\ell,Ex)\DefinedAs  f(\theta_1,C,\ell,Ex)\wedge f(\theta_2,C,\ell,Ex)$.
  \item $f(C'\theta,C,\ell,Ex)\DefinedAs  f(\theta,C',\ell,Ex)$.
  \item $\theta$ is of the form $\theta=\Next\theta_1$.

  \noindent For the variables $x_i\in C\setminus Ex$, we need to distinguish between those whose copies
  $y_i^{\ell}$ refer to the last position of the given finite trace $w$ from those whose copies
  $y_i^{\ell}$ refer to positions along $w$ distinct from the last one. In the first case, when we apply the next modality
  $\Next$ in the context $C$, then the copies $y_i^{\ell}$ of the corresponding variables $x_i$ will refer to positions
  which does not belong to $w$. This is encoded by adding the set $V$ of these variables $x_i$ to the set $Ex$. Thus,
   $f(\Next\theta_1,C,\ell,Ex)$ is defined as follows, where $y_0$ is a fresh variable:
   \[
   \begin{array}{l}
   f(\Next\theta_1,C,\ell,Ex)\DefinedAs  \displaystyle{\bigvee_{V\subseteq C\setminus Ex}}\exists y_1^{\ell+1}\ldots
   \exists y_n^{\ell+1}.\,\Bigl(\displaystyle{\bigwedge_{x_i\in C\setminus Ex}}(x_i\in V \leftrightarrow \forall y_0.\, y_0\leq y_i^{\ell})\,\wedge\,\vspace{0.2cm}\\
      \displaystyle{\bigwedge_{x_i\in (\Var\setminus C)\cup Ex\cup V}}y_i^{\ell+1}= y_i^{\ell}\,\wedge\,
     \displaystyle{\bigwedge_{x_i\in C\setminus (Ex\cup V)}}y_i^{\ell+1}= y_i^{\ell}+1 \,\wedge\,
     f(\theta_1,C,\ell+1,Ex\cup V)
   \end{array}
   \]
   \item $\theta$ is of the form $\theta=\theta_1\Until\theta_2$.

   \noindent This is the case more involved. Intuitively, in the evaluation of $\theta_1\Until\theta_2$
    within the context $C$, the variables $x_i\in C\setminus Ex$ move forwardly and synchronously of a non-negative offset
    $z\geq 0$. After this shift, $\theta_2$ must be true and for all intermediate offsets $z'\leq z$, $\theta_1$ must be true. However, during these $z'$-shifts some variables can move beyond the given finite trace $w$ in $enc(w)=w\cdot \{\#\}^{\omega}$.
    We can handle this situation as follows. For a subset $V\subseteq \Var$, we denote by $Part(V)$ the set of tuples
    $(P_1,\ldots,P_k)$ such that $P_1,\ldots,P_k$ form a partition of $V$ (i.e., $P_1,\ldots,P_k$ are pairwise disjunct and their union is $V$) and $P_1,\ldots,P_{k-1}$ are not empty. Note that $k\leq n$. Then, for the previous offset $z$, there must be a partition $(P_1,\ldots,P_k)\in Part(C\setminus Ex)$ and intermediate offsets $0\leq z_1<\ldots <z_k$ such that
    $z_k=z$ and the following holds:
    \begin{itemize}
      \item for each $i\in [1,k-1]$, $P_i$ is the set of all and only variables in $C\setminus Ex$ that after the $z_i$-shift move exactly at the
      last position of $w$,
      \item $P_k$ is the (possibly empty) set of all and only the variables in $C\setminus Ex$ that after the $z_k$-shift, where $z_k=z$, refer to positions belonging to $w$.
    \end{itemize}
    Then the $\FOPLUS$ formula $f(\theta_1\Until\theta_2,C,\ell,Ex)$ is defined as follows:
   \[
   \begin{array}{l}
   f(\theta_1\Until\theta_2,C,\ell,Ex)\DefinedAs \displaystyle{\bigvee_{(V_1,\ldots,V_k)\in Part(C\setminus Ex)}}\exists z_0\exists z_1\ldots\exists z_k.\, z_1<z_2\ldots<z_k\,\wedge \, \vspace{0.2cm}\\
   z_0=0\,\wedge \, (z_k>0 \rightarrow f(\theta_1,C,\ell,Ex))\,\wedge \, \displaystyle{\bigwedge_{i\in [1,k-1]} g(z_{i-1},z_i,V_i,\theta_1, C, \ell,Ex \cup \bigcup_{j\in [1,i-1]}V_i ) } \,\wedge\, \vspace{0.2cm}\\
  \displaystyle{h(z_{k-1},z_k,\theta_1,\theta_2,C,\ell, Ex\cup \bigcup_{j\in [1,k-1]}V_j )}
   \end{array}
   \]
   where for two  variables $y$ and $z$, a subset  $V\subseteq C$, two subformulas
   $\psi_1$ and $\psi_2$ of $\varphi$, $\ell\in\N$, and a subset $Ex\subseteq \Var$, the $\FOPLUS$ formulas
   $g(y,z,V,\psi_1, \ell, C, Ex)$ and $h(y,z, \psi_1,\psi_2,C, \ell,  Ex)$ are defined as follows:
\[
  \begin{array}{l}
  g(y,z,V,\psi_1, C,\ell,Ex)\DefinedAs \exists y_1^{\ell+1}.\,\ldots \exists y_n^{\ell+1}.
  \Bigl( \displaystyle{\bigwedge_{x_i\in (\Var\setminus C)\cup Ex}}  y_i^{\ell+1}= y_i^{\ell}\,\, \wedge \vspace{0.2cm}\\
  \displaystyle{ \bigwedge_{x_i\in C\setminus Ex}}(y_i^{\ell+1}= y_i^{\ell}+z \wedge (x_i\in V \leftrightarrow \forall z_0.\, z_0\leq y_i^{\ell+1}))\, \wedge\,  \forall z_0\forall y_1^{\ell+2}.\,\ldots \forall y_n^{\ell+2}.\,\vspace{0.2cm}\\
  \bigl[y<z_0\leq z\,\wedge\,\displaystyle{\bigwedge_{x_i\in (\Var\setminus C)\cup Ex}} y_i^{\ell+2}= y_i^{\ell}  \, \wedge\,   \displaystyle{\bigwedge_{x_i\in C\setminus Ex}}  y_i^{\ell+2}= y_i^{\ell}+z_0  \bigr] \longrightarrow
 f(\psi_1,C,\ell+2,Ex) \Bigr)
\end{array}
\]
\[
  \begin{array}{l}
  h(y,z,\psi_1,\psi_2, C,\ell,Ex)\DefinedAs \exists y_1^{\ell+1}.\,\ldots \exists y_n^{\ell+1}.
  \Bigl( \displaystyle{\bigwedge_{x_i\in (\Var\setminus C)\cup Ex}}  y_i^{\ell+1}= y_i^{\ell} \,\wedge
   \vspace{0.2cm}\\
  \displaystyle{\bigwedge_{x_i\in C\setminus Ex}}y_i^{\ell+1}= y_i^{\ell}+z \,\wedge\,f(\psi_2,C,\ell+1,Ex) \,\wedge\,\forall z_0\forall y_1^{\ell+2}.\,\ldots \forall y_n^{\ell+2}.\,\vspace{0.2cm}\\
  \bigl[y<z_0<z\,\wedge\,\displaystyle{\bigwedge_{x_i\in (\Var\setminus C)\cup Ex}} y_i^{\ell+2}= y_i^{\ell}  \, \wedge\,   \displaystyle{\bigwedge_{x_i\in C\setminus Ex}}  y_i^{\ell+2}= y_i^{\ell}+z_0  \bigr] \longrightarrow
 f(\psi_1,C,\ell+2,Ex) \Bigr)
\end{array}
\]
\end{itemize}\vspace{0.2cm}

Now, we prove that the construction is correct. Let $w$ be a finite trace. Given a $w$-pointed trace assignment
$\Pi$ with domain $\Var$ and $\ell\in\N$, we associate to $\Pi$ and $\ell$ a $w$-valuation  $g(\Pi,\ell)$ with domain $\{y_1^{\ell},\ldots,y_n^{\ell}\}$ defined as follows for each variable $y_i^{\ell}$, where $\Pi(x_i)=(enc(w),h)$:
$g(\Pi,\ell)(y_i^{\ell})\DefinedAs h$ if $h<|w|$, and $g(\Pi,\ell)(y_i^{\ell})\DefinedAs |w|-1$ otherwise.

A  $w$-pointed trace assignment $\Pi$ is \emph{consistent with a subset $Ex\subseteq \Var$} if for each $x_i\in\Var$
with $\Pi(x_i)=(enc(w),h)$, it holds that $h<|w|$ iff $x_i\notin Ex$. By construction, Claim~1,  and a straightforward induction
on the structure of the given subformula $\theta$ of $\varphi$, we obtain the following result.
\vspace{0.2cm}

\noindent \textbf{Claim 2.} Let $w$ be a finite trace, $\theta$ a subformula of $\varphi$, $\emptyset\neq C\subseteq \Var$ a context, $\ell\in \N$, and $Ex\subseteq \Var$. Then, for each   $w$-pointed-trace assignment $\Pi$ which is consistent with
$Ex$, it holds that
\[
\Pi,C\models \theta\,\, \text{ iff }\,\, w,g(\Pi,\ell)\models f(\theta,C,\ell,Ex)
\]
Let $f(\varphi)$ be the $\FOPLUS$ sentence defined as follows:
\[
f(\varphi)\DefinedAs \exists x_1\ldots \exists x_n.\,(f(\varphi,\Var,1,\emptyset)\wedge \displaystyle{\bigwedge_{i\in [1,n]}}
(y_i^{1}=x_i \wedge x_i=0))
\]
By Claim~1, it follows that for each finite trace, $enc(w)\models \varphi$ if and only if $w\models f(\varphi)$.
This concludes the proof of Proposition~\ref{FromSingleTRaceCHLTLToFOPLUS}.
 \end{proof}

\subsection{Syntax and semantics of past $\CHLTL$ and proof of Theorem~\ref{theo:CharacterizationCHLTLInFIniteTraces}}\label{APP:CharacterizationCHLTLInFIniteTraces}

In this section, we first define the syntax and semantics of past $\CHLTL$, the extension of $\CHLTL$ with past temporal modalities. Then, we provide a proof of
Theorem~\ref{theo:CharacterizationCHLTLInFIniteTraces}. 

 The syntax of past $\CHLTL$ quantifier-free
  formulas $\psi$ over $\AP$ and $\Var$ is as follows:
\[
   \psi ::=  \top  \ | \   \Rel{p}{x}  \ | \ \neg \psi \ | \ \psi \wedge \psi \ | \ \Next\psi \ | \ \PNext\psi \ | \   \psi \Until \psi 
   \ | \   \psi \PUntil \psi\ | \ \tpl{C} \psi
\]
where $\PNext$ (``previous") and $\PUntil$ (``since") are the standard past counterparts of the temporal modalities
$\Next$ and $\Until$, respectively.
 Let $C$ be a context, $\Pi$  a  pointed-trace assignment, and $i\geq 0$ be an offset.
 The \emph{$(C,i)$-predecessor of $\Pi$}, denoted by $\Pi -_C i$, is defined as follows:
 \begin{itemize}
   \item if there is $x\in \Dom(\Pi)\cap C$ such that $\Pi(x)=(\pi,h)$ and $h<i$, then 
   $\Pi -_C i= \und$ (intuitively, $\und$ denotes the \emph{undefined} value);
   \item otherwise, $\Pi -_C i$ is the pointed-trace assignment with domain $\Dom(\Pi)$ defined as
follows. For each $x\in \Dom(\Pi)$, where $\Pi(x)=(\pi,h)$:
$[\Pi -_C i](x)=(\pi,h-i)$ if $x\in C$, and $[\Pi -_C i](x)=\Pi(x)$
otherwise.
 \end{itemize}
 
Intuitively, if the pointed-trace assignment $\Pi -_C i$ is defined, then the positions of the pointed traces associated with the
variables in $C$ move back of  $i$ steps, while the positions of the
other pointed traces remain unchanged. The semantics of the past temporal modalities is as follows:
  \[ \begin{array}{ll}
  (\Pi,C) \models  \PNext\psi & \Leftrightarrow  \Pi -_C 1 \neq \und \text{ and }  (\Pi -_C 1,C)\models  \psi\\
  (\Pi,C) \models  \psi_1\PUntil \psi_2 & \Leftrightarrow  \text{for some }i\geq 0 \text{ such that }\Pi -_C i \neq \und:\,   (\Pi -_C i,C)\models  \psi_2   \text{ and }\\
  & \,\, (\Pi -_C k,C) \models  \psi_1 \text{ for all } 0\leq k<i
\end{array} \]
 We write $\Pi\models \psi$ to mean that $(\Pi,\Var)\models\psi$. Moreover, for a trace $w$,
 we write $w\models \psi$ to mean that $\Pi_w\models \psi$, where $\Pi_w$ is the initial pointed-trace assignment
 associating to each trace variable occurring in $\psi$ the pointed trace $(w,0)$.
 
 Theorem~\ref{theo:CharacterizationCHLTLInFIniteTraces} directly follows from the following two propositions.
 
 \begin{proposition} Let $\varphi$ be a $\FOINPLUS$ sentence over $\AP$. Then, one can construct in polynomial time a $\CHLTL$ quantifier-free formula $f(\varphi)$ over $\AP$ such that for each trace $w$, it holds that
$w\models \varphi$ \emph{if and only if} $w\models f(\varphi)$.
\end{proposition}
\begin{proof} The proof is similar to that of Proposition~\ref{prop:FromFOPLUSToCHLTL} but we use the past temporal modalities for 
encoding the atomic formulas of the given $\FOINPLUS$ sentence $\varphi$. In particular, given two variables $x$ and $y$, we exploit the following past $\CHLTL$ quantifier-free formula $\psi_=(x,y)$ for expressing that $x$ and $y$ refer to the same position:
  \[ \begin{array}{ll}
  \psi_=(x,y) & \DefinedAs   (\init(x)\wedge \init(y))\vee \{x,y\}\PNext (\,\top \,\PUntil\, (\init(x)\wedge \init(y)))\vspace{0.2cm}\\
  \init(z) & \DefinedAs   \{z\}\,\neg \PNext (\displaystyle{\bigvee_{p\in \AP}}(p[z]\vee \neg p[z]))\quad \text{ for each }z\in\Var 
\end{array} \]
\end{proof}
 
 The proof of the following proposition is a simplified version  of that of Proposition~\ref{FromSingleTRaceCHLTLToFOPLUS} and we omit the details here.

 \begin{proposition} Let $\varphi$ be a past $\CHLTL$ quantifier-free formula over $\AP$.  Then, one can construct in polynomial time a  $\FOINPLUS$ sentence $f(\varphi)$ over $\AP$ such that for each  trace $w$, it holds that $w\models f(\varphi)$ \emph{if and only if} $w\models  \varphi$.
\end{proposition}

 \subsection{Proof of Theorem~\ref{theo:UndecidabilityMC_CHLTL}}\label{APP:UndecidabilityMC_CHLTL}

 In this section, we provide a proof of Theorem~\ref{theo:UndecidabilityMC_CHLTL}.

\setcounter{aux}{\value{theorem}}
\setcounter{theorem}{\value{theorem-UndecidabilityMC_CHLTL}}
\setcounter{auxSec}{\value{section}}
\setcounter{section}{\value{sec-UndecidabilityMC_CHLTL}}

\begin{theorem}  Model-checking against the fragment $\U$ of $\CHLTL$ is $\Sigma_0^{1}$-hard.
\end{theorem}
\setcounter{theorem}{\value{aux}}
\setcounter{section}{\value{auxSec}}
\begin{proof}
The result is obtained by a polynomial time reduction from the halting problem of non-deterministic  Minsky $2$-counter machines~\cite{Harel86}, a well-known $\Sigma_0^{1}$-complete problem~\cite{Harel86}.

Recall that the fragment $\U$ consists of two-variable quantifier alternation-free formulas of the form $\exists x_1.\exists x_2.\, \psi_0\wedge \{x_2\}\Eventually \{x_1,x_2\}\psi$, where $\psi_0$ and $\psi$ are quantifier-free $\HLTL$ formulas.  Fix a Minsky $2$-counter machine $M = \tpl{Q,\Delta,\delta_\init,\delta_\halt}$. Note that $M$  is defined as at the
 end of Section~\ref{sec:AsynchronousHLTL} but the designated transition $\delta_\rec$ is replaced with the halting transition
 $\delta_\halt$. An halting computation of $M$ is a finite computation of $M$ starting at the initial configuration and leading to a
  configuration of the form $(\delta_\halt,\nu)$ for some counter valuation $\nu:\{1,2\}\rightarrow \N$. The halting problem consists in checking the existence of an halting computation for the given machine $M$.
   Without loss of generality, we assume that $\delta_\init\neq \delta_\halt$.
   We construct a  Kripke structure $\Ku_M$ and a $\CHLTL$ formula $\varphi_M$ in the fragment $\U$ such that there is an halting computation of $M$ iff  $\Lang(\Ku_M)\models \varphi_M$ (i.e., $\Ku_M$ is a model of $\varphi_M$).

We first provide a description of the reduction.   Codes of $M$-configurations (\emph{segments}) are similar to the ones described in the proof of Theorem~\ref{theo:UndecidabilitySIngleTraceAHLTL}, where the padding proposition allows to obtain codes of the same $M$-configuration having arbitrary length.
The finite Kripke structure $\Ku_M$ ensures that its set of traces contains the \emph{well-formed traces} $\pi$ of the
form $s_1 \ldots s_n \cdot \{\bot\}^{\omega}$, where   $s_1,\ldots,s_n$ are segments, segment $s_1$ encodes the initial configuration,
segment $s_n$ encodes some halting configuration, and the second  segment $s_2$ is marked by a special proposition. Then the  $\U$-formula
ensures that both the two variables, say $x_1$ and $x_2$, are bound to a well-formed trace $\pi=s_1 \ldots s_n \cdot \{\bot\}^{\omega}$ and $\pi$ is faithful to the evolution of $M$. This last requirement is ensured by moving in the context $\{x_2\}$ the position for the
$x_2$-trace to the beginning of the second segment $s_2$, and by requiring, by means of a lockstepwise traversal of the trace $\pi$ and the suffix of $\pi$ starting from the beginning of the second segment $s_2$, that all the segments have the same length. In this way, in this traversal, the $(i+1)^{th}$-segment of $\pi$ along the $x_2$-trace matches the $i^{th}$-segment along the $x_1$-trace, and one can ensure that the configuration encoded by the $(i+1)^{th}$-segment is a $M$-successor of the configuration encoded by the $i^{th}$-segment. Note that
the existence of an halting computation in $M$ entails the existence of a uniform upper bound  on the values of the counters along the computation. Thus, by exploiting the padding proposition, one can ensure that in the encoding of  such a computation all the segments have the same length.

We now proceed with the technical details of the reduction. We exploit the set $\AP$ of atomic propositions defined as   $\AP\DefinedAs \Delta \cup \{c_1,c_2,\#,\Beg,\Pad,\bot\}$. An  $M$-configuration $(\delta,\nu)$ is encoded by the finite words over $2^{\AP}$ (called \emph{segments}) of the form:
\[
\{\Beg,\delta\} P_1\ldots P_m \{\Pad\}^{k}
\]
where $k\geq 1$, $m= \max(\nu(1),\nu(2))$, and for all $i\in [1,m]$, (i) $\emptyset\neq  P_i\subseteq \{c_ 1,c_2\}$ and (ii)  for all $\ell\in\{1,2\}$, $c_\ell\in P_i$ iff $i\leq \nu(\ell)$. A \emph{marked segment} is defined as a segment but the first position is additionally labeled by the proposition $\#$. A \emph{well-formed} trace of $M$ is a trace of the
form $s_1\ldots s_n\cdot \{\bot\}^{\omega}$ such that (i) $n>1$ and $s_1,\ldots,s_n$ are $n$ segments, (ii)
$s_1$ encodes the initial configuration, (iii) $s_n$ encodes some halting configuration, and (iv) $s_2$ is a marked segment.
Thus, halting computations of $M$ are encoded by well-formed traces which, additionally, are faithful to the evolution of $M$.
The following result is straightforward.\vspace{0.2cm}

\noindent \textbf{Claim.} One can construct in polynomial time a  Kripke structure $\Ku_M$ such that
\begin{compactitem}
  \item $\Lang(\Ku_M)$ contains all the well-formed traces.
  \item For each $\pi\in \Lang(\Ku_M)$, $\pi$ is well-formed iff $\pi$ visits a position labeled by $\bot$.
\end{compactitem}\vspace{0.2cm}

We construct a $\U$-formula $\varphi_M$ such that $\Lang(\Ku_M)\models \varphi_M$ iff there is a well-formed trace
which is faithful to the evolution of $M$. Hence, the result follows. The $\U$-formula $\varphi_M$ is defined as follows:
\[
\varphi_M\DefinedAs \exists x_1.\exists x_2. \displaystyle{\bigwedge_{p\in \AP}}\Always
(p[x_1] \leftrightarrow p[x_2]) \, \wedge \, \{x_2\}\Eventually \{x_1,x_2\}\bigl( \#[x_2] \,\wedge\, \Eventually \bot[x_2] \,\wedge\, \psi_f \,\bigr)
\]
where $\psi_f$ is a quantifier-free $\HLTL$ formula defined in the following. Note that the first conjunct ensures that variables $x_1$ and $x_2$ are bound to the same trace $\pi$ of $\Ku_M$, while by the previous claim,  the subformula $\Eventually \bot[x_2]$ in the  second conjunct ensures that the trace $\pi= s_1\ldots s_n\cdot \{\bot\}^{\omega}$ is well-formed. Moreover,
the position for the $x_2$-trace is moved to the  beginning of the second segment of $\pi$ in the context $\{x_2\}$ and $\psi_f$
ensures by a lockstepwise traversal of the trace $\pi$ and the suffix of $\pi$ starting
from the beginning of the second segment that $\pi$ is faithful  to the evolution of $M$.
The  quantifier-free $\HLTL$ formula $\psi_f$ consists of  two conjuncts enforcing the following requirements:

\begin{itemize}
   \item The segments $s_1,\ldots,s_n$ have the same length. This ensures that in the lockstepwise
   traversal of $\pi$ (bound to variable $x_1$) and $s_2\ldots s_n \cdot \{\bot\}^{\omega}$ (bound to variable $x_2$), the
    $(i+1)^{th}$-segment for the $x_2$-trace matches in the traversal the  $i^{th}$-segment for the $x_1$-trace for all $1\leq i<n$.
   \[
   \Always(\Beg[x_2] \leftrightarrow \Beg[x_1])
   \]
   \item The configuration encoded by the $(i+1)^{th}$-segment for the $x_2$-trace is a $M$-successor
   of the configuration encoded by the  $i^{th}$-segment for the $x_1$-trace for all $1\leq i<n$.
\[
 \displaystyle{\bigwedge_{\delta=(q,(\instr,\ell),q')\in\Delta}}
 \Always \Bigl( \bigl(\Beg[x_1]\wedge \Beg[x_2] \wedge \delta[x_1]\bigr) \,\longrightarrow\, \Bigl(\displaystyle{\bigvee_{\delta'=(q',(\instr',\ell'),q'')\in\Delta}}\delta'[x_2] \,\wedge\,
\]
\[
   \hspace{2.5cm} \underbrace{\Next\bigl((c_{3-\ell}[x_1] \wedge c_{3-\ell}[x_2])  \Until (\neg c_{3-\ell}[x_1] \wedge \neg c_{3-\ell}[x_2])}_{\text{the value of counter $3-\ell$ does not change}}\bigr)\,\wedge\,
\]
\[
     \hspace{0.5cm} \underbrace{\instr =\inc \rightarrow \Next\bigl((c_{\ell}[x_1] \wedge c_{\ell}[x_2])  \Until (\neg  c_{\ell}[x_1] \wedge c_{\ell}[x_2]\wedge \neg \Next c_{\ell}[x_2])}_{\text{increment of counter $\ell$}}\bigr)\,\wedge\,
\]
\[
     \hspace{0.5cm} \underbrace{\instr =\dec \rightarrow \Next\bigl((c_{\ell}[x_1] \wedge c_{\ell}[x_2])  \Until (  c_{\ell}[x_1] \wedge \neg c_{\ell}[x_2]\wedge \neg \Next c_{\ell}[x_1])}_{\text{decrement of counter $\ell$}}\bigr)\,\wedge\,
\]
\[
     \hspace{2.5cm} \underbrace{\instr =\zero \rightarrow \Next (\neg c_{\ell}[x_1] \wedge\neg  c_{\ell}[x_2]) }_{\text{zero test of counter $\ell$}}\,\,\Bigr)\Bigr)
\]
\end{itemize}

By construction, if $\Lang(\Ku_M)\models \varphi_M$, then there exists a well-formed trace which is faithful to the evolution of
$M$, hence, there is an halting computation of $M$. Vice versa if there is an halting computation $\rho$ of $M$, then there is a uniform bound  on the counter values along $\rho$. Hence, by construction, there exists a well-formed trace $\pi$ whose segments have the same length which satisfies the quantifier-free part of $\varphi_M$ with $x_1$ and $x_2$ bound to $\pi$. By the previous claim, it follows that $\Lang(\Ku_M)\models \varphi_M$, and we are done.
 \end{proof}

\subsection{Proof of Theorem~\ref{theo:DecidabilityMCSimpleCHLTL}}\label{APP:DecidabilityMCSimpleCHLTL}

In this section, we provide  a proof of Theorem~\ref{theo:DecidabilityMCSimpleCHLTL}.
We show that a  simple $\CHLTL$ formula can be translated in polynomial time into an equivalent $\FOE$ formula (see Appendix~\ref{APP:LogicsWIthEqualLevelPredicate} for the definition of the syntax and semantics of $\FOE$). Thus, since model checking of $\FOE$ is decidable~\cite{CoenenFHH19}, Theorem~\ref{theo:DecidabilityMCSimpleCHLTL} directly follows.

\begin{proposition}\label{prop:TranslationSimpleCHLTLIntoHLTL} Given a simple $\CHLTL$ formula, one can construct in polynomial time an equivalent $\FOE$ formula.
\end{proposition}
\begin{proof} Let $\Var=\{x_1,\ldots,x_n\}$ and $\F$ be the set of quantifier-free $\CHLTL$ formulas over $\Var$ of the form $C\psi$ such that $C$ is a context and $\psi$ is a quantifier-free $\CHLTL$ formula whose contexts are all singletons, i.e.~consisting of a single trace variable. Evidently, each simple $\CHLTL$ quantifier-free formula $\varphi$ can be equivalently rewritten as a Boolean combination of quantifier-free formulas
in the fragment $\F$. Thus, it suffices to show that for each formula $C\psi$ in the fragment $\F$, there exists a  $\FOE$ formula $\psi'$ with free variables $x_1,\ldots,x_n$    such that for all sets $\Lang$ of traces and \emph{initial} pointed-trace assignments $\Pi$ over $\Var$ and $\Lang$, it holds that $\Pi \models C\psi$ iff $\Pi\models_\Lang \psi'$.

For each variable $x_i\in \Var$ and $\ell\in \N$, let $y_i^{\ell}$ be a fresh copy of $x_i$. For a quantifier-free $\CHLTL$ formula
$\psi$ whose contexts are all singletons, a context $\emptyset\neq C\subseteq \Var$, and $\ell\in \N$, we associate to the tuple
$(\psi,C,\ell)$, a $\FOE$ formula $f(\psi,C,\ell)$ with free variables
$y_1^{\ell},\ldots,y_n^{\ell}$. The mapping $f$ is inductively defined as follows (recall that for each $p\in \AP$, $\FOE$ provides the monadic predicate $P_p$):

\begin{itemize}
  \item $f(p[x_i],C,\ell)\DefinedAs P_{p}(y_i^{\ell})$.
  \item $f(\neg\psi,C,\ell)\DefinedAs \neg f(\psi,C,\ell)$.
  \item $f(\psi_1\wedge \psi_2,C,\ell)\DefinedAs f(\psi_1,C,\ell)\wedge
  f(\psi_2,C,\ell)$.
  \item
\[
  \begin{array}{l}
  f(\Next\psi,C,\ell)\DefinedAs   \exists y_1^{\ell+1}.\,\ldots \exists y_n^{\ell+1}.\,
  \bigl(\displaystyle{\bigwedge_{x_i\in C}}y_i^{\ell+1}= y_i^{\ell}+1\, \wedge\, \displaystyle{\bigwedge_{x_i\notin C}}  y_i^{\ell+1}= y_i^{\ell}\,\wedge\, f(\psi,C,\ell+1)\bigr)
\end{array}
\]
  Note that for variables $x$ and $y$, the predicate $y=x+1$ can be expressed in $\FOE$ as follows:
  $x<y\wedge \neg \exists z.\,x<z \wedge z<y$.
   \item
\[
  \begin{array}{ll}
  f(\psi_1\Until\psi_2,C,\ell)\DefinedAs & \exists y_1^{\ell+1}.\,\ldots \exists y_n^{\ell+1}.\,
  \Bigl(\displaystyle{\bigwedge_{x_i\in C}}y_i^{\ell}\leq y_i^{\ell+1}\, \wedge\, \displaystyle{\bigwedge_{x_i\notin C}}  y_i^{\ell+1}= y_i^{\ell}\,\wedge\, \theta_E(C,\ell+1)\,\wedge\, \vspace{0.1cm}\\
 &  f(\psi_2,C,\ell+1)\,\wedge\,\forall y_1^{\ell+2}.\,\ldots \forall y_n^{\ell+2}.\,\bigl[\displaystyle{\bigwedge_{x_i\in C}}
 (y_i^{\ell}\leq y_i^{\ell+2} \,\wedge\,y_i^{\ell+2}< y_i^{\ell+1}) \, \wedge\, \vspace{0.1cm}\\
 & \displaystyle{\bigwedge_{x_i\notin C}}  y_i^{\ell+2}= y_i^{\ell}\,\wedge\, \theta_E(C,\ell+2)  \bigr] \longrightarrow
 f(\psi_1,C,\ell+2) \Bigr)\vspace{0.1cm}\\
 \theta_E(C,\ell)\DefinedAs &  \displaystyle{\bigwedge_{x_i,x_j\in C \mid x_i\neq x_j}} E(y_i^{\ell},y_j^{\ell})
\end{array}
\]
\noindent Note that if $C$ is a singleton, then $\theta_E(C,\ell)$ is $\top$ (the empty conjunction is $\top$).
\item $f(\{x_i\}\psi,C,\ell)\DefinedAs  f(\psi,\{x_i\},\ell)$.
\end{itemize}

By construction, we obtain the following result, where a pointed-trace assignment $\Pi$ with domain $\Var$ is \emph{consistent with a context $C$} if there is a unique position $m$ such that  for all $x_i\in C$, it holds that $\Pi(x_i)$ is of the form $(\pi,m)$
for some trace $\pi$.\vspace{0.2cm}

\noindent \textbf{Claim.} Let $\psi$ be a quantifier-free $\CHLTL$ formula over $\Var$
whose contexts are all singletons, $C\subseteq \Var$ a context, and $\ell\in \N$. Then, for each set $\Lang$ of traces and each pointed-trace assignment $\Pi$ over $\Var$ and $\Lang$  which is consistent with the context $C$:
\[
\Pi\models \psi \text{ if and only if } \Pi_\ell \models_\Lang f(\psi,C,\ell)
\]
where $\Pi_\ell$ is the pointed-trace assignment over $\{y_1^{\ell},\ldots,y_n^{\ell}\}$ defined as:
$\Pi_\ell(y_i^{\ell})\DefinedAs \Pi(x_i)$.\vspace{0.1cm}

\noindent \textbf{Proof of the claim.} By a straightforward induction on the structure of the formula
$\psi$, (i) we first prove that the result holds when $C$ is a singleton, and then (ii) we generalize the result to the case where $C$ is not a singleton.  \qed \vspace{0.2cm}

For a formula $C\psi$ in the fragment $\F$, let $\psi'$ be the $\FOE$ formula with free variables $x_1\ldots,x_n$ given by
$\exists y_1^{1}.\,\ldots \exists y_n^{1}.\,(f(\psi,C,1)\wedge\displaystyle{\bigwedge_{x_i\in \Var}}y_i^{1}=x_i)$. By the previous claim, we obtain that
for all sets $\Lang$ of traces and \emph{initial} pointed-trace assignments $\Pi$ over $\Var$ and $\Lang$, it holds that $\Pi \models \psi$ iff $\Pi\models_\Lang \psi'$, and we are done.
\end{proof}

\section{Proofs from Section~\ref{sec:OverallExpressiveness}}\label{APP:OverallExpressiveness}

\subsection{Syntax and semantics of $\FOE$ and $\MSOE$}\label{APP:LogicsWIthEqualLevelPredicate}

We recall syntax and semantics of the logics $\FOE$~\cite{Finkbeiner017} and $\MSOE$~\cite{CoenenFHH19}.
$\FOE$~\cite{Finkbeiner017} extends standard first-order logic $\FO$  over infinite words or traces  with the equal-level binary predicate $E$. It is a first-order logic with equality over the signature
$\{<,E\}\cup \{P_a \mid a\in \AP\}$ for a given finite set $\AP$ of atomic propositions. $\MSOE$ extends $\FOE$ by second-order quantification over monadic predicates.

Given a finite set $\AP$ of atomic propositions, a finite set $V_1=\{x,y,\ldots\}$ of first-order variables, and
a finite set $V_2=\{X,Y,\ldots\}$ of second-order variables, the syntax of   $\MSOE$ formulas
over $\AP$, $V_1$, and $V_2$ is as follows:
\[
\varphi ::=    P_a(x)  \ | \ x\in X \ | \ x<y  \ | \ x=y  \ | \ E(x,y)  \ | \ \neg \varphi \ | \ \varphi \wedge \varphi \ | \ \exists x. \varphi   | \ \exists X. \varphi
\]
where $a\in \AP$, $x,y\in V_1$, and $X\in V_2$. A $\MSOE$ sentence is a $\MSOE$ formula where each (first-order or second-order) variable
occurs in the scope of a quantifier. 
$\FOE$ corresponds to the fragment of $\MSOE$ where the predicate $x\in X$ and the second-order quantifiers are disallowed.
 Note that standard first-order logic $\FO$ over traces and its monadic second-order extension $\MSO$ syntactically correspond to the fragments of $\FOE$ and $\MSOE$, respectively, where the equal-level predicate $E$ is disallowed.

 While $\MSO$ is interpreted over traces, $\MSOE$ is interpreted over  sets of traces.
 A set $\Lang$ of traces induces the relational structure with domain $\Lang\times \N$ (i.e., the set of pointed traces associated with $\Lang$), where (i) the
 binary predicate $<$ is interpreted as the set of pairs of pointed traces  in $\Lang\times \N$ of the form  $((\pi,i_1),(\pi,i_2))$ such that  $i_1<i_2$, and (ii) the equal-level predicate $E$ is interpreted as the set of pairs of pointed traces  in $\Lang\times \N$ of the form  $((\pi_1,i),(\pi_2,i))$. Hence, the predicate $<$ allows to compare distinct timestamps along the same trace, while
 the equal-level predicate  allows to compare distinct traces at the same timestamp.
 The semantics of $\MSOE$ is formally defined as follows.  A \emph{first-order valuation over $\Lang$} is a pointed-trace assignment $\Pi$ over $V_1$  and $\Lang$, and
 a \emph{second-order valuation over $\Lang$} is a mapping $\V$ assigning to each second-order variable in $V_2$ a set of pointed traces in $\Lang\times \N$.
  Given a $\MSOE$ formula $\varphi$, the satisfaction relation
 $\Pi,\V\models_\Lang \varphi$ is inductively defined as follows (we omit the semantics for the Boolean connectives which is standard):
\[ \begin{array}{ll}
\Pi,\V \models_\Lang  P_a(x)  &  \Leftrightarrow  \Pi(x)=(\pi,i) \text{ and }a\in \pi(i)\\
\Pi,\V \models_\Lang  x\in X  &  \Leftrightarrow  \Pi(x)\in \V(X)\\
\Pi,\V \models_\Lang  x<y  &  \Leftrightarrow  \Pi(x)=(\pi_1,i_1),\,\Pi(y)=(\pi_2,i_2),\,\pi_1=\pi_2,\, \text{ and }i_1<i_2\\
\Pi,\V \models_\Lang  x=y  &  \Leftrightarrow  \Pi(x)=\Pi(y) \\
\Pi,\V \models_\Lang  E(x,y)  &  \Leftrightarrow  \Pi(x)=(\pi_1,i_1),\,\Pi(y)=(\pi_2,i_2), \text{ and }i_1=i_2\\
\Pi,\V \models_\Lang  \exists x. \varphi  &  \Leftrightarrow  \Pi[x \mapsto (\pi,i)]\models_\Lang \varphi \text{ for some }(\pi,i)\in\Lang\times N\\
\Pi,\V \models_\Lang  \exists X. \varphi  &  \Leftrightarrow  \V[X \mapsto A]\models_\Lang \varphi \text{ for some }A\subseteq\Lang\times N
\end{array} \]
where  $\V[X \mapsto A](X)=A$ and $\V[X \mapsto A](Y)=\V(Y)$ for $Y\neq X$.
Note that if $\varphi$ is a sentence, then the satisfaction relation $\Pi,\V\models_\Lang \varphi$ is independent of $\Pi$ and $\V$.
 We say that $\Lang$ is a model of $\varphi$, written
$\Lang\models \varphi$ if $\Pi,\V\models_\Lang \varphi$ for some first-order  valuation $\Pi$ and second-order valuation $\V$ over $\Lang$.

\subsection{Proof of Lemma~\ref{lemma:MSOEkregular}}\label{APP:MSOEkregular}

In this section, we provide a proof of Lemma~\ref{lemma:MSOEkregular} (see Appendix~\ref{APP:LogicsWIthEqualLevelPredicate} for the definition of the syntax and semantics of $\MSOE$ and $\MSO$).
\setcounter{aux}{\value{lemma}}
\setcounter{lemma}{\value{lemma-MSOEkregular}}
\setcounter{auxSec}{\value{section}}
\setcounter{section}{\value{sec-MSOEkregular}}

\begin{lemma}Let $k\geq 1$ and  $\varphi$ be a $\MSOE$
  sentence over $\AP$. Then, one can construct a  $\MSO$ sentence $\varphi'$
  over $\AP\times [1,k]$ whose set of models is the $k$-language of
  $\varphi$.
  \end{lemma}
\setcounter{lemma}{\value{aux}}
\setcounter{section}{\value{auxSec}}
\begin{proof}
Fix $k\geq 1$. Let $\psi$ be a $\MSOE$ formula over $\AP$ with first-order variables $V_1=\{x_1,\ldots,x_n\}$ and second-order variables
$V_2 =\{X_1,\ldots,X_m\}$ for some $n\geq 1$ and $m\geq 0$. For each first-order variable $x_i$, we introduce
$k$ fresh first-order variables $y_{i,1},\ldots,y_{i,k}$. Moreover,
 for each second-order variable  $X_j$, we introduce $k$ fresh second-order variables
$Y_{j,1},\ldots,Y_{j,k}$. Let $V'_1=\{y_{1,1},\ldots,y_{1,k},\ldots,y_{n,1},\ldots,y_{n,k}\}$
and $V'_2=\{Y_{1,1},\ldots,Y_{1,k},\ldots,Y_{m,1},\ldots,Y_{m,k}\}$.
Let $\Lang=\{\pi_1,\ldots,\pi_k\}$ be a set of traces (over $\AP$) of cardinality $k$.
Intuitively, for a first-order variable $x_i$ mapped to a pointed trace $(\pi_\ell,h)$, we use the
variable $y_{i,\ell}$ in $\MSO$ for keeping track of the position $h$ associated to the trace $\pi_\ell$ in $\Lang$.
Moreover, for   a set $A\subseteq \Lang\times \N$ mapped to variable $X_i$, we use the second-order variables  $Y_{i,1},\ldots,Y_{i,k}$ for partitioning the set of timestamps in $A$ in accordance to the first component  of the  pointed traces in $A$ (recall that the second-order variables in $\MSO$ are mapped to sets of positions of the given trace).

Let $g: V_1 \mapsto [1,k]$ be a choice function associating to each first-order variable $x_i$ an  index $\ell\in [1,k]$. We associate to the $\MSOE$ formula $\psi$ and the choice function $g$ an $\MSO$ formula $f(\psi,g)$ over $\AP\times [1,k]$, $V'_1$, and $V'_2$.
The mapping $f(\psi,g)$ is inductively defined as follows:
\begin{itemize}
  \item $f$ is homomorphic for the Boolean connectives.\vspace{0.1cm}

  \item $f(P_a(x_i),g)\DefinedAs P_{(a,g(x_i))}(y_{i,g(x_i)})$ \vspace{0.1cm}

  \item $f(x_i\in X_j,g)\DefinedAs y_{i,g(x_i)}\in Y_{j,g(x_i)}$ \vspace{0.1cm}

  \item
  $
  f(x_i<x_j,g)\DefinedAs
          \left\{\begin{array}{ll}
                  \neg\top             & \text{ if }g(x_i)\neq g(x_j) \\
                 y_{i,g(x_i)}< y_{j,g(x_j)}            &  \text{ otherwise}
                \end{array}\right.
 $\vspace{0.1cm}

 \item
  $
  f(x_i=x_j,g)\DefinedAs
          \left\{\begin{array}{ll}
                  \neg\top             & \text{ if }g(x_i)\neq g(x_j) \\
                 y_{i,g(x_i)}= y_{j,g(x_j)}            &  \text{ otherwise}
                \end{array}\right.
 $\vspace{0.1cm}

 \item $f(E(x_i,x_j),g)\DefinedAs y_{i,g(x_i)}= y_{j,g(x_j)}$ \vspace{0.1cm}

 \item $f(\exists x_i.\,\psi,g)\DefinedAs \displaystyle{\bigvee_{\ell\in [1,k]}}\exists y_{i,\ell}.\, f(\psi,g[x_i \mapsto \ell])$\vspace{0.1cm}

 \item $f(\exists X_i.\, \psi,g)\DefinedAs \exists Y_{i,1}.\ldots \exists Y_{i,k}. f(\psi,g)$
\end{itemize}\vspace{0.1cm}

\noindent where $g[x_i \mapsto \ell](x_j)=\ell$ if $j=i$, $g[x_i \mapsto \ell](x_j)=g(x_j)$ otherwise.  Fix a well-formed trace $\nu$  over $\AP\times [1,k]$ and let $\Lang(\nu)$ be the set of traces over $\AP$ having cardinality $k$ associated with $\nu$. Then, $\Lang(\nu)$ can be written in the form $\Lang(\nu)= \{\pi_1,\ldots,\pi_k\}$ where
for each $\ell \in [1,k]$ and $i\geq 0$, $\pi_\ell(i) =\{p \in\AP \mid (p,\ell)\in \nu(i)\}$.
Let $g: V_1 \mapsto [1,k]$ be a choice function, $\V_1$ be a first-order $\MSO$-valuation over $V'_1$  (i.e., a function associating a position $i\geq 0$ to each variable in $V'_1$), and $\V_2$ be a second-order $\MSO$-valuation over $V'_2$ for $\MSO$
(i.e., a function associating a subset of $\N$ to each variable in $V'_2$).
We denote by $\Pi(\V_1,g)$  and $\V(\V_2)$ the first-order $\MSOE$-valuation and  the second-order $\MSOE$-valuation over $V_1$ and $V_2$, respectively, for the set of traces $\Lang(\nu)$, defined as follows:
\begin{itemize}
  \item for each $x_i\in V_1$, let $g(x_i)= \ell$: then $\Pi(\V_1,g)(x_i)\DefinedAs (\pi_\ell,\V_1(y_{i,\ell}))$;
  \item for each $X_i\in V_2$, $\V(\V_2)(X_i)\DefinedAs \displaystyle{\bigcup_{1\leq \ell\leq k} \{\pi_\ell\}\times \V_2(Y_{i,\ell})}$.
\end{itemize}
We first show the following.\vspace{0.2cm}

\noindent \textbf{Claim.} $\V_1,\V_2 \models_\nu f(\psi,g)$ if and only if $\Pi(\V_1,g),\V(\V_2)\models_{\Lang(\nu)} \psi$.\vspace{0.1cm}

\noindent \textbf{Proof of the claim: } by induction on the structure of $\psi$. The case of Boolean connectives directly follows from the induction hypothesis. For the other cases, we proceed as follows:
\begin{itemize}
  \item $\psi= P_a(x_i)$ for some $a\in \AP$ and $i\in [1,n]$. Let $g(x_i)=\ell$ and $\V_1(y_{i,\ell})=h$. By construction,
  $f(\psi,g)=P_{(a,\ell)}(y_{i,\ell})$ and $\Pi(\V_1,g)(x_i)=(\pi_\ell,h)$. Moreover,
  $a\in \pi_\ell(h)$ iff $(a,\ell)\in \nu(h)$. Hence, the result follows.   \vspace{0.1cm}

  \item $\psi = x_i\in X_j$ for some $i\in [1,n]$ and $j\in [1,m]$. Let $g(x_i)=\ell$, $\V_1(y_{i,\ell})=h$, and $\V_2(Y_{j,\ell})=A\subseteq \N$. By construction,
  $f(\psi,g)= y_{i,\ell}\in Y_{j,\ell}$, $\Pi(\V_1,g)(x_i)=(\pi_\ell,h)$, and $\V(\V_2)(X_j)\cap (\{\pi_\ell\}\times \N)= \{\pi_\ell\}\times A$. Hence, $\Pi(\V_1,g)(x_i)\in \V(\V_2)(X_j)$ iff $\V_1(y_{i,\ell})\in \V_2(Y_{j,\ell})$, and the result holds in this case as well. \vspace{0.1cm}

  \item $\psi=  x_i< x_j$ for some $i,j\in [1,n]$. Let $g(x_i)=\ell$, $g(x_j)=\ell'$,  $\V_1(y_{i,\ell})=h$, and $\V_1(y_{j,\ell'})=h'$. By construction,
$\Pi(\V_1,g)(x_i)=(\pi_\ell,h)$ and $\Pi(\V_1,g)(x_j)=(\pi_{\ell'},h')$. Moreover, \emph{either} $\ell\neq \ell'$ and $f(\psi,g)=\neg\top$, \emph{or} $\ell=\ell'$ and $f(\psi,g)=y_{i,\ell}<y_{j,\ell}$. Thus, since $\pi_\ell\neq \pi_{\ell'}$ if $\ell\neq \ell'$,  the result easily follows.   \vspace{0.1cm}

\item $\psi=  x_i= x_j$ for some $i,j\in [1,n]$. This case is similar to the previous one.\vspace{0.1cm}

\item $\psi=E(x_i,x_j)$ for some $i,j\in [1,n]$. Let $g(x_i)=\ell$, $g(x_j)=\ell'$,  $\V_1(y_{i,\ell})=h$, and $\V_1(y_{j,\ell'})=h'$. By construction,
$\Pi(\V_1,g)(x_i)=(\pi_\ell,h)$, $\Pi(\V_1,g)(x_j)=(\pi_{\ell'},h')$, and $f(\psi,g)= y_{i,\ell}=y_{j,\ell'}$.
Hence, (i) $\V_1,\V_2 \models_\nu f(\psi,g)$ iff $h=h'$ and (ii) $\Pi(\V_1,g),\V(\V_2)\models_{\Lang(\nu)} \psi$ iff $h=h'$, and the result follows. \vspace{0.1cm}

\item $\psi = \exists x_i.\, \theta$ for some $i\in [1,n]$. First, assume that $\V_1,\V_2 \models_\nu f(\exists x_i.\, \theta,g)$.
By construction, there exists $\ell\in [1,k]$ such that $\V_1,\V_2 \models_\nu  \exists y_{i,\ell}.\, f(\theta,g[x_i\mapsto \ell])$. Hence, for some $h\in\N$, $\V_1[y_{i,\ell} \mapsto h],\V_2\models_\nu f(\theta,g[x_i\mapsto \ell])$.
Let $\V'_1 = \V_1[y_{i,\ell} \mapsto h]$ and $g' = g[x_i\mapsto \ell]$. By the induction hypothesis, we have that
 $\Pi(\V'_1,g'),\V(\V_2)\models_{\Lang(\nu)} \theta$. Since $\Pi(\V'_1,g')=\Pi(\V_1,g)[x_i \mapsto (\pi_\ell,h)]$, we obtain that
  $\Pi(\V_1,g)[x_i \mapsto (\pi_\ell,h)],\V(\V_2)\models_{\Lang(\nu)} \theta$. Hence, $\Pi(\V_1,g),\V(\V_2)\models_{\Lang(\nu)} \exists x_i.\, \theta$. The converse implication is similar.\vspace{0.1cm}

  \item $\psi = \exists X_i.\, \theta$ for some $i\in [1,m]$. First, assume that $\V_1,\V_2 \models_\nu f(\exists X_i.\, \theta,g)$.
By construction, for each $\ell\in [1,k]$, there exists a set $A_\ell\subseteq \N$ such that  $\V_1,\V_2[Y_{i,1}\mapsto A_1,\ldots,Y_{i,k}\mapsto A_k] \models_\nu    f(\theta,g)$.
Let $\V'_2 = \V_2[Y_{i,1}\mapsto A_1,\ldots,Y_{i,k}\mapsto A_k]$. By the induction hypothesis, we have that
 $\Pi(\V_1,g),\V(\V'_2)\models_{\Lang(\nu)} \theta$. Let $A= \displaystyle{\bigcup_{\ell\in [1,k]}}\{\pi_\ell\}\times A_i$. Since $\V(\V'_2)=\V(\V_2)[X_i\mapsto A]$, we obtain that
  $\Pi(\V_1,g),\V(\V_2)[X_i\mapsto A]\models_{\Lang(\nu)} \theta$. Hence, $\Pi(\V_1,g),\V(\V_2)\models_{\Lang(\nu)} \exists X_i.\, \theta$. The converse implication is similar.
\end{itemize}
This concludes the proof of the claim.\qed\vspace{0.2cm}

Fix an arbitrary choice function $g_0: V_1\mapsto [1,k]$. For a $\MSOE$ sentence $\varphi$ over $\AP$, $V_1,$ and $V_2$, let $\varphi'$ be the $\MSO$ sentence over $\AP\times [1,k]$
given by
\[
\varphi':= f(\varphi,g_0)\wedge \displaystyle{\bigwedge_{\ell, \ell'\in [1,k]:\,\ell\neq \ell} \exists z_0.\,  \bigvee_{a\in\AP}}
(P_{(a,\ell)}(z_0) \leftrightarrow \neg P_{(a,\ell')}(z_0))
\]
Note that the second conjunct captures the well-formed traces over $\AP\times [1,k]$. By the previous claim, it follows that
the $k$-language of $\varphi$ is the set of traces (over $\AP\times [1,k]$) satisfying $\varphi'$.
This concludes the proof of Lemma~\ref{lemma:MSOEkregular}.
\end{proof}

\end{document}